\documentclass[12pt]{article}
\usepackage{amsmath}
\usepackage{graphicx}
\usepackage{enumerate}
\usepackage{natbib}
\usepackage{url} 
\usepackage{bbm}

\usepackage{xr-hyper}
\makeatletter
\newcommand*{\addFileDependency}[1]{
  \typeout{(#1)}
  \@addtofilelist{#1}
  \IfFileExists{#1}{}{\typeout{No file #1.}}
}
\makeatother

\newcommand*{\myexternaldocument}[1]{
    \externaldocument{#1}
    \addFileDependency{#1.tex}
    \addFileDependency{#1.aux}
}

\myexternaldocument{appendix}

\RequirePackage[colorlinks,citecolor=blue,urlcolor=blue]{hyperref}

\usepackage{amsmath, amssymb, amscd, amsthm, amsfonts}
\usepackage[perpage,symbol*]{footmisc}
\usepackage{float}
\usepackage{hyperref}
\usepackage{pgfplots}
\usepackage{color}  
\usepackage{tikz} 
\usepackage{algorithm,algcompatible,amsmath}
\DeclareMathOperator*{\argmax}{\arg\!\max}%
\DeclareMathOperator*{\argmin}{\arg\!\min}
\algnewcommand\INPUT{\item[\textbf{Input:}]}%
\algnewcommand\OUTPUT{\item[\textbf{Output:}]}%
\usepackage{authblk}
\usepackage{enumitem}
\usepackage[english]{babel}
\setcitestyle{authoryear,open={(},close={)}}
\usepackage{colortbl,dcolumn}
\usepackage{listings} 
\usepackage{xcolor} 
\usepackage{amsfonts,amssymb}
\usepackage{mathrsfs}
\usepackage{comment}

\newtheorem{theorem}{Theorem}

\newtheorem{remark}{Remark}

\newtheorem{condition}[theorem]{Condition}

\newtheorem{corollary}{Corollary}

\newtheorem{definition}{Definition}

\newtheorem{lemma}{Lemma}
\newtheorem{assumption}{Assumption}

\numberwithin{equation}{section} 
\numberwithin{theorem}{section}
\numberwithin{lemma}{section} 
\numberwithin{corollary}{section}
\numberwithin{definition}{section}
\numberwithin{proposition}{section} 
\numberwithin{remark}{section}
\numberwithin{example}{section}
\DeclareMathOperator\supp{supp}
\numberwithin{table}{section}
\numberwithin{figure}{section}

\newcommand{\T}{\intercal}

\newcommand{\blind}{0}

\usepackage{fancyhdr}
\begin{document}

\def\spacingset#1{\renewcommand{\baselinestretch}%
{#1}\small\normalsize} \spacingset{1}


\if0\blind
{
  \title{\bf Randomness of Shapes and Statistical Inference on Shapes via the Smooth Euler Characteristic Transform}
  \author[1,*]{Kun Meng}
	\author[2]{Jinyu Wang}
	\author[3,4]{Lorin Crawford}
	\author[3]{Ani Eloyan}
	
	\affil[1]{\small Division of Applied Mathematics, Brown University, RI, USA}	
	\affil[2]{\small Data Science Institute, Brown University, RI, USA}
	\affil[3]{\small Department of Biostatistics, Brown University School of Public Health, Providence, RI, USA}
	\affil[4]{\small Microsoft Research New England, Cambridge, MA, USA}
	\affil[*]{Address for correspondence: Kun Meng, Division of Applied Mathematics, Brown University, 182 George Street, Providence, RI 02912, USA. Email: \texttt{kun\_meng@brown.edu}.}
 
  \maketitle
} \fi

\if1\blind
{
  \bigskip
  \bigskip
  \bigskip
  \begin{center}
    {\Large\bf Randomness of Shapes and\\ Statistical Inference on Shapes via the Smooth Euler Characteristic Transform}
\end{center}
} \fi

\bigskip
\begin{abstract}
In this article, we establish the mathematical foundations for modeling the randomness of shapes and conducting statistical inference on shapes using the smooth Euler characteristic transform. Based on these foundations, we propose two chi-squared statistic-based algorithms for testing hypotheses on random shapes. Simulation studies are presented to validate our mathematical derivations and to compare our algorithms with state-of-the-art methods to demonstrate the utility of our proposed framework. As real applications, we analyze a data set of mandibular molars from four genera of primates and show that our algorithms have the power to detect significant shape differences that recapitulate known morphological variation across suborders. Altogether, our discussions bridge the following fields: algebraic and computational topology, probability theory and stochastic processes, Sobolev spaces and functional analysis, analysis of variance for functional data, and geometric morphometrics.
\end{abstract}

\noindent%
\vfill

\newpage
\spacingset{1.89}


\section{Introduction}\label{Introduction}

The quantification of shapes has become an important research direction. It has brought advances to many fields including geometric morphometrics \citep{boyer2011algorithms, gao2019gaussian, gao2019gaussianmorphometrics}, biophysics and structural biology \citep{wang2021statistical, tang2022topological}, and radiogenomics \citep{crawford2020predicting}. When shapes are considered as random variables, their corresponding quantitative summaries are also random, implying that such summaries of random shapes are statistics. The statistical inference on shapes based on these quantitative summaries has been of particular interest \citep{fasy2014confidence, roycraft2023bootstrapping}.

In this paper, we bring together mathematical and statistical approaches to make three significant contributions to shape statistics: (i) we provide mathematical foundations for the randomness of shapes encountered in applications, bridging algebraic topology \citep{hatcher2002algebraic} and stochastic processes \citep{hairer2009introduction}; (ii) we connect the statistical inference on shape-valued data to the well-studied analysis of variance for functional data \citep[fdANOVA,][]{zhang2013analysis}, bridging topological data analysis \citep[TDA,][]{edelsbrunner2010computational} and functional data analysis \citep[FDA,][]{hsing2015theoretical}; and (iii) our framework does not rely on any assumptions about diffeomorphisms or pre-specified landmarks.

\subsection{A Motivating Scientific Question}\label{section: A Motivation Question}

Through modeling the randomness of shapes, we aim to address the following statistical inference question: \textit{Is the observed difference between two groups of shapes statistically significant?} For example, the mandibular molars in Figure \ref{fig: Teeth} are from four genera of primates. A pertinent question from a morphological perspective is: \textit{In Figure \ref{fig: Teeth}, do the molars from genus \textit{Tarsius} exhibit significant differences from those from the other genera?} 

The primary objective of this paper is to propose a powerful approach for testing hypotheses on random shapes. This would help address morphology-motivated statistical inference questions like the one raised above. In achieving this objective, we lay down the mathematical foundations that justify our hypothesis testing methods. We take two key steps: In \textbf{Step 1}, we find the appropriate representations of shapes; and in \textbf{Step 2}, we test hypotheses on shapes using these representations. In Section \ref{section: Overview of The Shape and Topological Data Analysis}, we provide a literature review on shape representations and introduce the topological summary employed in this paper. Section \ref{section: Overview of Contributions and Paper Orgainization} begins by presenting the main theme of our hypothesis testing approach, followed by an overview of our contributions. Since the molars in Figure \ref{fig: Teeth} are diffeomorphic to the 2-dimensional unit sphere, some existing diffeomorphism-related methods can be considered for representing the molars \citep[e.g., parameterized surfaces;][]{kurtek2011elastic}. In contrast, we aim to propose an approach that does not rely on any diffeomorphic assumptions, allowing for a wider range of applications.
\begin{figure}[h]
\centering
\includegraphics[scale=0.141]{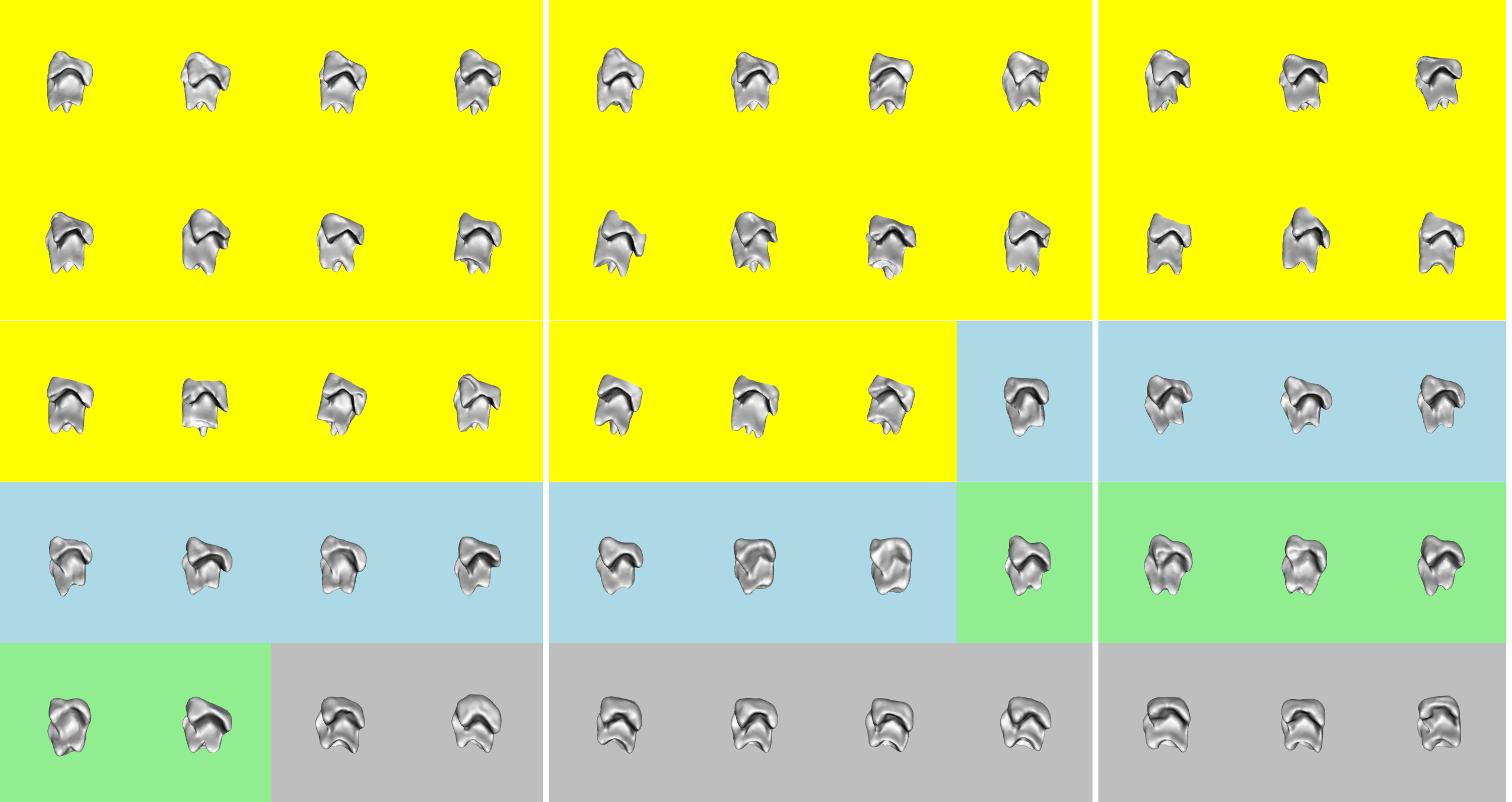}
\includegraphics[scale=0.145]{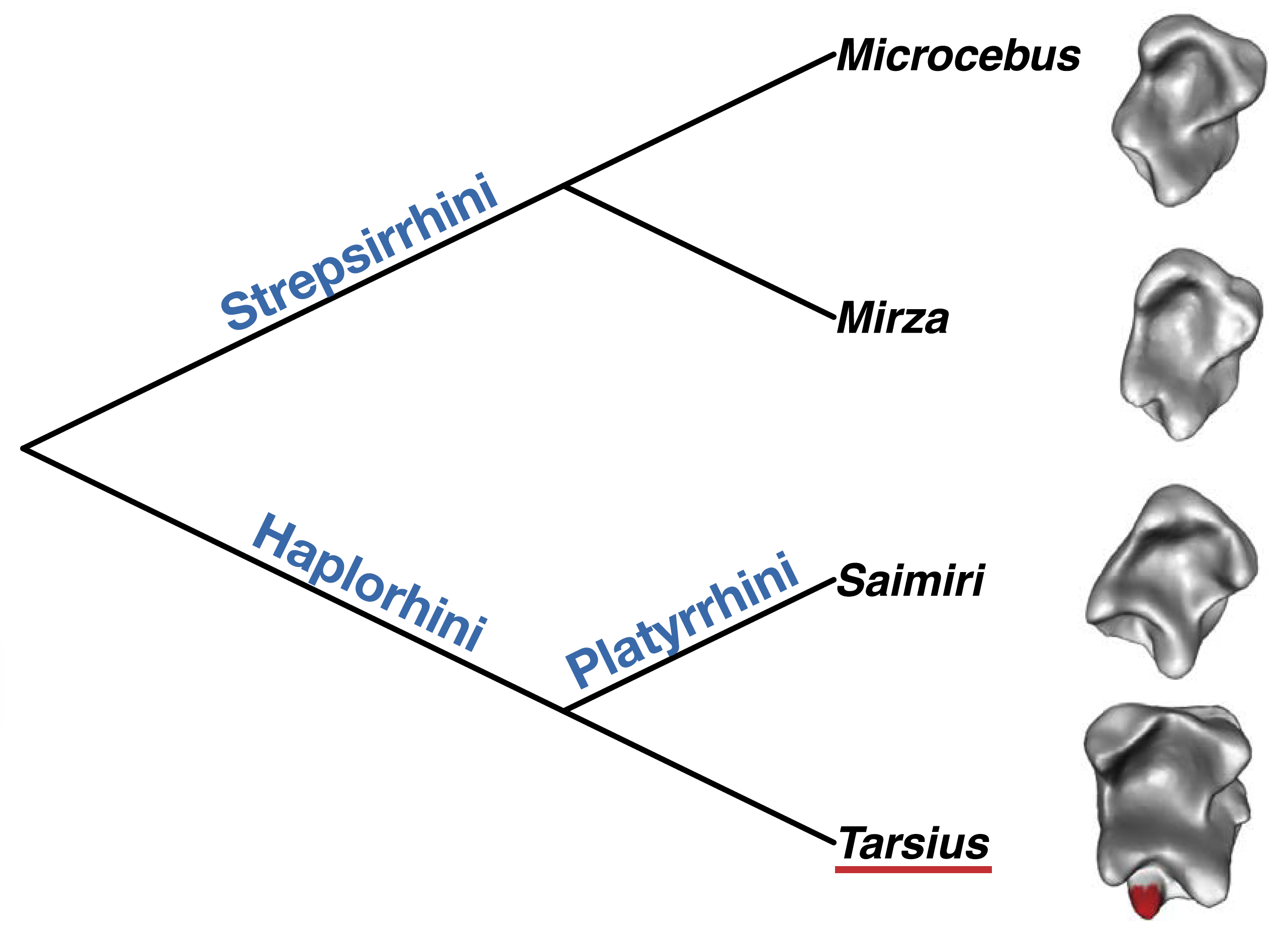}
        \caption{\footnotesize{Left: Molars from two suborders of the primates: Haplorhini and Strepsirrhini. The Haplorhini suborder has genera \textit{Tarsius} (yellow) and \textit{Saimiri} (grey). The Strepsirrhini suborder has genera \textit{Microcebus} (blue) and \textit{Mirza} (green). Right: Relationship between the four primate genera. Tarsier molars exhibit additional high cusps (highlighted in red). A similar figure was published in \cite{wang2021statistical}.}}
    \label{fig: Teeth}
\end{figure}

\subsection{Overview of Shape and Topological Data Analysis}\label{section: Overview of The Shape and Topological Data Analysis}

In classical geometric morphometrics, shapes are represented using prespecified points called landmarks \citep{kendall1989survey}. The manual landmarking of a collection of shapes requires domain knowledge, can be very labor intensive, and is subject to bias \citep{boyer2011algorithms}. Furthermore, an equal number of landmarks must be selected for each shape in a study in order to make comparisons (e.g., the Procrustes framework discussed in Section 2.1 of \cite{gao2019gaussianmorphometrics}). This necessitates comprehensive information about entire collections of shapes for consistency, which can be difficult to obtain (e.g., landmarking cancer tumors, which can have very different morphology across a population of patients). Unfortunately, many datasets do not come with prespecified landmarks \citep[e.g.,][]{goswami2015phenome10k}. Although many algorithms can automatically sample reasonable landmarks on shapes when their parameters are fine-tuned \citep[e.g.,][]{ gao2019gaussianmorphometrics,gao2019gaussian}, using a finite number of landmarks extracted from a continuum inevitably results in the loss of information. Diffeomorphism-based approaches \citep{dupuis1998variational, gao2019gaussianmorphometrics} are part of the ``computational anatomy" that was historically studied by the ``pattern theory school" pioneered by Ulf Grenander \citep{grenander1998computational}. They enable the comparison of (dis-)similarity between shapes with benefit of bypassing the need for landmarks. However, these approaches are based on the assumption that the shapes being compared are diffeomorphic to one another, making them unsuitable for many datasets (e.g., fruit fly wings in \cite{miller2015fruit}). Furthermore, parameterized curves and surfaces (PCS) provide a toolbox for assessing the heterogeneity of shapes with summary statistics that are invariant to reparameterizations \citep{kurtek2010parameterization, kurtek2011elastic, kurtek2012statistical}. Despite their effectiveness in analyzing real data \citep[e.g., DT-MRI brain fibers;][]{kurtek2012statistical}, PCS are based on assumptions about the diffeomorphism types of the shapes of interest. For example, \cite{kurtek2011elastic} focuses on surfaces that are diffeomorphic to the 2-dimensional unit sphere.

TDA opens the door for landmark-free and diffeomorphism-free representations of shapes. Motivated by differential topology, \cite{turner2014persistent} proposed the persistent homology transform (PHT) with the capability to sufficiently encode all information within shapes \citep{ghrist2018persistent}. To describe the PHT, we briefly provide some basics of TDA. One common statistical invariant in TDA is the persistence diagram \citep[PD,][]{edelsbrunner2010computational}. When equipped with the Wasserstein distance, the collection of PDs, denoted as \(\mathscr{D}\), is a Polish space \citep{mileyko2011probability}. Thus, probability measures can be applied, and the randomness of shapes can be represented using the probability measures on \(\mathscr{D}\). The PHT takes values in \(C(\mathbb{S}^{d-1};\mathscr{D}^d)=\{\mbox{continuous maps }F:\mathbb{S}^{d-1} \rightarrow \mathscr{D}^d\}\), where \(\mathbb{S}^{d-1}\) denotes the sphere \(\{x\in\mathbb{R}^d:\Vert x\Vert=1\}\) and \(\mathscr{D}^d\) is the \(d\)-fold Cartesian product of \(\mathscr{D}\) \citep[][Lemma 2.1 and Definition 2.1]{turner2014persistent}. A single PD does not preserve all information of a shape \citep{crawford2020predicting}. In contrast, the PHT is injective, which means it preserves all the information of a shape. However, since \(\mathscr{D}\) is not a vector space and the distances on \(\mathscr{D}\) are abstract \citep[e.g., the Wasserstein and bottleneck distances,][]{cohen2007stability}, many fundamental statistical concepts do not easily apply to summaries resulting from the PHT. For example, the definition of moments corresponding to probability measures on \(\mathscr{D}\) (e.g., means) is highly nontrivial \citep{mileyko2011probability}. The difficulty in defining these concepts hinders the application of PHT-based statistical methods in \(C(\mathbb{S}^{d-1};\mathscr{D}^d)\).

The smooth Euler characteristic transform \citep[SECT,][]{crawford2020predicting} offers an alternative summary statistic for shapes. The SECT not only preserves the information of shapes \citep[][Corollary 1]{ghrist2018persistent} but also represents shapes using continuous functions instead of PDs. More precisely, the values of the SECT are maps from the sphere $\mathbb{S}^{d-1}$ to a separable Banach space $\mathcal{B}\overset{\operatorname{def}}{=}C([0,T])$, the collection of continuous functions on a compact interval $[0,T]$ (values of $T$ will be given in Eq.~\eqref{eq: def of sublevel sets}). Hence, for any shape $K$, its SECT, denoted as $\{\operatorname{SECT}(K)(\nu)\}_{\nu\in\mathbb{S}^{d-1}}$, lies in $\mathcal{B}^{\mathbb{S}^{d-1}}= \{\text{maps }F: \mathbb{S}^{d-1}\rightarrow\mathcal{B}\}$. Specifically, $\operatorname{SECT}(K)(\nu)$ belongs to $\mathcal{B}$ for each $\nu\in\mathbb{S}^{d-1}$. As a result, the randomness of shapes $K$ is represented via the SECT by a collection of $\mathcal{B}$-valued random variables. Probability theory in separable Banach spaces is better developed than in $\mathscr{D}$ \citep[e.g.,][]{hairer2009introduction}. In particular, a $\mathcal{B}$-valued random variable is a stochastic process with its sample paths in $\mathcal{B}$. As we will demonstrate in Section \ref{The Definition of Smooth Euler Characteristic Transform}, $\mathcal{B}$ here can be replaced with a reproducing kernel Hilbert space (RKHS). The theory of stochastic processes has evolved over a century and FDA is a well-developed branch of statistics. Consequently, a myriad of tools are available to underpin both the randomness of shapes and the statistical inference on shapes.

From an application perspective, \cite{crawford2020predicting} applied the SECT to magnetic resonance images taken from tumors in a cohort of glioblastoma multiforme (GBM) patients. Using summary statistics derived from the SECT as predictors within Gaussian process regression, the authors demonstrated that the SECT can predict clinical outcomes more effectively than existing tumor shape quantification approaches and common molecular assays. The relative performance of the SECT in the GBM study suggests a promising future for its utility in medical imaging and broader statistical applications related to shape analyses. Similarly, \cite{wang2021statistical} utilized derivatives of the Euler characteristic transform (ECT) as predictors in statistical models for subimage analysis. This analysis is akin to variable selection, aiming to identify physical features that are important for distinguishing between two classes of shapes. Lastly, \cite{marsh2022detecting} highlighted that the SECT outperforms the standard measures employed in organoid morphology.

\subsection{Overview of Contributions and Paper Organization}\label{section: Overview of Contributions and Paper Orgainization}

Our goal is to address the hypothesis testing question posed in Section \ref{section: A Motivation Question} by employing a landmark-free and diffeomorphism-free approach, which opens up possibilities for further applications in the future. We formulate the question more generically here. Let $\mathbb{P}^{(1)}$ and $\mathbb{P}^{(2)}$ be two distributions that generate two collections of random shapes, $\{K_i^{(1)}\}_{i=1}^n$ and $\{K_i^{(2)}\}_{i=1}^n$. Detecting whether there is a significant difference between $\{K_i^{(1)}\}_{i=1}^n$ and $\{K_i^{(2)}\}_{i=1}^n$ is equivalent to rejecting the hypothesis $\mathbb{P}^{(1)} = \mathbb{P}^{(2)}$. Since each shape $K_i^{(j)}$ is random, $\operatorname{SECT}(K_i^{(j)})$ is a random variable taking values in a vector space (as discussed in Section \ref{section: Overview of The Shape and Topological Data Analysis}) and can be decomposed as follows (see Theorem \ref{thm: KL expansions of SECT} for a rigorous version)
\begin{align}\label{eq: decomposition for the main theme of HT}
    \operatorname{SECT}(K_i^{(j)}) = m^{(j)} + \text{random terms},\ \ \ \text{ for }j\in\{1,2\},
\end{align}
where $m^{(j)}$ denotes the mean of $\operatorname{SECT}(K_i^{(j)})$ with respect to the distribution $\mathbb{P}^{(j)}$. The random terms in Eq.~\eqref{eq: decomposition for the main theme of HT} can be characterized by the Karhunen–Loève (KL) expansion \citep[][Section 7.3]{hsing2015theoretical}. To reject $\mathbb{P}^{(1)} = \mathbb{P}^{(2)}$, it suffices to reject $m^{(1)}=m^{(2)}$. That is, the question posed in Section \ref{section: A Motivation Question} can be addressed by testing for the equality of two means. The important component of the test is the variance represented by the random terms in Eq.~\eqref{eq: decomposition for the main theme of HT}. In Section \ref{section: hypothesis testing}, we formulate this test as a two-sample problem in the fdANOVA literature \citep[][Section 5.2]{zhang2013analysis}. In addition, using the KL expansion, we provide a $\chi^2$-statistic in Section \ref{section: hypothesis testing} to test the hypothesis. Throughout the paper, our focus is on the two-sample problem. However, one may also consider employing the one-way fdANOVA to compare the means of three or more groups of shapes. The theoretical foundation and numerical experiments for this aspect are left for future research.

To develop our framework, we have to address the following mathematical foundation related questions: (i) \textit{What underlying probability spaces allow the randomness of shapes and their corresponding SECT?} and (ii) \textit{Are the conditions of the KL expansion satisfied in our setting?} We answer these questions in Sections \ref{The Definition of Smooth Euler Characteristic Transform} and \ref{section: distributions of Gaussian bridge} --- we model the randomness of shapes via the SECT using RKHS-valued random fields. The ``theory of random sets" is a well-established framework for characterizing set-valued random variables \citep{molchanov2005theory}. However, its application to persistent homology-based statistics (e.g., the SECT) remains underexplored. In this paper, we introduce a new probability space to characterize the randomness of shapes in a manner compatible with the SECT. 

We first propose a collection of shapes as our sample space on which the SECT is well-defined. We then demonstrate that every shape in this collection has its SECT in $C(\mathbb{S}^{d-1};\mathcal{H})=\{\mbox{continuous maps }F: \mathbb{S}^{d-1}\rightarrow\mathcal{H}\}$, where $\mathcal{H}  =  H_0^1([0,T])$ is not only a Sobolev space \citep{brezis2011functional} but also an RKHS (reasons for using $[0,T]$ instead of $(0,T)$ for $H^1_0([0,T])$ are in Appendix \ref{section: notation for closed vs. open}). Importantly, $C(\mathbb{S}^{d-1};\mathcal{H})$ is a separable Banach space (Theorem \ref{thm: the separability of C(Shere;H)}) and, hence, a Polish space. It helps construct a probability space to characterize the distributions of shapes. Building on this probability space, we define the mean and covariance of the SECT. Using the Sobolev embedding, we present some properties of the mean and covariance, which pave the way for the KL expansion of the SECT.

Traditionally, the statistical inference on shapes in TDA is conducted in the persistence diagram space $\mathscr{D}$, which is unsuitable for exponential family-based distributions and requires any corresponding statistical inference to be highly nonparametric \citep{fasy2014confidence, robinson2017hypothesis}. The PHT-based statistical inference in $C(\mathbb{S}^{d-1};\mathscr{D}^d)$ is even more difficult. With the KL expansion of the SECT, we propose a $\chi^2$-statistic for testing hypotheses on shapes. Beyond the mathematical foundations, we also provide simulation studies to illustrate the performance of our proposed hypothesis testing method. Lastly, we apply our proposed framework to answer the motivating question raised in Section \ref{section: A Motivation Question}.

We organize this paper as follows. In Section \ref{section: Notations and Mathematical Preparations}, we provide the mathematical preparations. In Section \ref{The Definition of Smooth Euler Characteristic Transform}, we define the SECT for a specific collection of shapes, highlighting its properties. In Section \ref{section: distributions of Gaussian bridge}, we construct a probability space to model shape distributions. In Section \ref{section: hypothesis testing}, we propose the KL expansion of the SECT, leading to a statistic for hypothesis testing. In Section \ref{section: Simulation experiments}, we conduct simulation studies to evaluate our method. In Section \ref{section: Applications}, we apply our method to real data. In Section \ref{Conclusions and Discussions}, we conclude the paper. The Appendix provides the proofs of theorems, further data analysis, and future research topics.



\section{Notations and Mathematical Preparations}\label{section: Notations and Mathematical Preparations}

To model the shapes discussed in our motivating question from Section \ref{section: A Motivation Question}, we need certain preparations regarding (i) topology and (ii) function spaces.

\noindent\textbf{Topology.} The first question we must address is: \textit{What are the ``shapes" in our framework?} \cite{ghrist2018persistent} and \cite{curry2022many} applied o-minimal structures \citep{van1998tame} to prove the injectivity of the PHT. Subsequent to this, o-minimal structures have been applied in many TDA studies to model shapes \citep[e.g.,][]{jiang2020weighted, kirveslahti2023representing}. To stay consistent with the existing literature, we also model shapes using o-minimal structures. An o-minimal structure is a sequence $\mathcal{S}=\{\mathcal{S}_n\}_{n\ge1}$ of subset collections $\mathcal{S}_n\subseteq 2^{\mathbb{R}^n}$ satisfying six set-theoretical axioms, where $2^{\mathbb{R}^n}$ denotes the power set of $\mathbb{R}^n$. The precise definition of o-minimal structures is available in \cite{van1998tame} and is provided in Appendix \ref{section: O-minimal Structures} for the reader's convenience.

A typical example of o-minimal structures is the collection of \textit{semialgebraic sets}. Specifically, a set $K\subseteq\mathbb{R}^n$ is semialgebraic if it can be expressed as a finite union of sets of the form $\{x\in\mathbb{R}^n \,\vert\ p(x)=0,\, q_1(x)>0,\,\ldots,\, q_k(x)>0\}$, where $p, q_1,\ldots,q_k$ are polynomial functions on $\mathbb{R}^n$. If we define $\mathcal{S}_n$ as the collection of semialgebraic subsets of $\mathbb{R}^n$, then $\mathcal{S}=\{\mathcal{S}_n\}_{n\ge1}$ is an o-minimal structure \citep[][Chapter 2]{van1998tame}. The unit sphere $\mathbb{S}^{d-1}$, open ball $B(0,R)=\{x\in\mathbb{R}^d \,\vert\,\Vert x\Vert^2<R^2\}$ for any $R>0$, and all polyhedra (e.g., polygon meshes in computer graphics) are semialgebraic, given that they can be represented using either the polynomial $\Vert x\Vert^2$ or affine functions. We assume the following:
\begin{assumption}\label{Assumption: basic requirements for o-minimal structures of interest}
    The o-minimal structure $\mathcal{S}$ of interest contains all semialgebraic sets.
\end{assumption}
\noindent Hereafter, a ``shape" refers to a compact set $K\in\bigcup_{n\ge1}\mathcal{S}_n$ for a prespecified o-minimal structure $\mathcal{S}=\{\mathcal{S}_n\}_{n\ge1}$ satisfying Assumption \ref{Assumption: basic requirements for o-minimal structures of interest}. Assumption \ref{Assumption: basic requirements for o-minimal structures of interest} incorporates many common shapes (e.g., balls and polyhedra) in our framework. More importantly, it implies the subsequent Theorem \ref{thm: finite triangularization} through the ``triangulation theorem" \citep[][Chapter 8]{van1998tame}. Although the definition of an o-minimal structure $\mathcal{S}$ is highly abstract (see Appendix \ref{section: O-minimal Structures}), each compact set in $\mathcal{S}$ resembles a polyhedron, which is precisely stated as follows.
\begin{theorem}\label{thm: finite triangularization}
Suppose $\mathcal{S}=\{\mathcal{S}_n\}_{n\ge1}$ is an o-minimal structure satisfying Assumption \ref{Assumption: basic requirements for o-minimal structures of interest} and $K\in\bigcup_{n\ge1}\mathcal{S}_n$. If $K$ is compact, there exists a finite simplicial complex $S$ such that the polyhedron $\vert S\vert \overset{\operatorname{def}}{=} \bigcup_{s\in S}s$ is homeomorphic to $K$, where each $s\in S$ denotes a simplex.
\end{theorem}
\noindent Herein, a finite simplicial complex $S$ is a finite collection of simplexes. Each face of a simplex $s\in S$ also belongs to $S$ (i.e., $S$ is a so-called ``closed complex" referred to in Chapter 8 of \cite{van1998tame}). Theorem \ref{thm: finite triangularization} directly results from the ``triangulation theorem" \citep{van1998tame}; hence, its proof is omitted. For the $d$-th component $\mathcal{S}_d$ of $\mathcal{S}=\{\mathcal{S}_n\}_{n\ge1}$, Theorem \ref{thm: finite triangularization} indicates that the compact sets $K\in\mathcal{S}_d$ are subsets of $\mathbb{R}^d$ that are homeomorphic to polyhedra. Theorem \ref{thm: finite triangularization} also implies that the homology groups of each compact $K\in\mathcal{S}_d$ are well-defined and finitely generated; hence, the Betti numbers and Euler characteristic of $K$ are well-defined and finite \citep[][Chapter 2]{hatcher2002algebraic}. 

\noindent\textbf{Function Spaces.} We apply the following notations throughout this paper:\\ 
(i) For any normed space $\mathcal{V}$, let $\Vert\cdot\Vert_{\mathcal{V}}$ denote its norm. Denote $\Vert\cdot\Vert_{\mathbb{R}^d}$ as $\Vert \cdot\Vert$ for succinctness. \\
(ii) Let $X$ be a compact metric space equipped with metric $d_X$, and let $\mathcal{V}$ denote a normed space. $C(X;\mathcal{V})$ is the collection of continuous maps from $X$ to $\mathcal{V}$. Furthermore, $C(X;\mathcal{V})$ is a normed space equipped with $\Vert f\Vert_{C(X;\mathcal{V})}  =  \sup_{x\in X}\Vert f(x)\Vert_{\mathcal{V}}$. The Hölder space $C^{0,\frac{1}{2}}(X;\mathcal{V})$ is defined as $\left\{f\in C(X;\mathcal{V}) \,\Big\vert\, \sup_{x,y\in X,\, x\ne y}\left(\frac{\Vert f(x)-f(y)\Vert_{\mathcal{V}}}{\sqrt{d_X(x,y)}}\right)<\infty\right\}$. Here, $C^{0,\frac{1}{2}}(X;\mathcal{V})$ is a normed space equipped with $\Vert f\Vert_{C^{0,\frac{1}{2}}(X;\mathcal{V})}  =  \Vert f\Vert_{C(X;\mathcal{V})}+\sup_{x,y\in X,\, x\ne y}\left(\frac{\Vert f(x)-f(y)\Vert_{\mathcal{V}}}{\sqrt{d_X(x,y)}}\right)$. Obviously, $C^{0,\frac{1}{2}}(X;\mathcal{V})\subseteq C(X;\mathcal{V})$. For simplicity, we denote $C(X)  =  C(X;\mathbb{R})$ and $C^{0,\frac{1}{2}}(X)  =  C^{0,\frac{1}{2}}(X;\mathbb{R})$. For a given $T>0$ (e.g., see Eq.~\eqref{eq: def of sublevel sets}), we denote $C([0,T])$ as $\mathcal{B}$. \\
(iii) The inner product of $\mathcal{H}  =  H_0^1([0,T]) = \{f\in L^2([0,1])\,\vert\, f'\in L^2([0,T]) \mbox{ and }f(0)=f(T)=0\}$ is defined as $\langle f, g \rangle = \int_0^T f'(t) g'(t)\, dt$ \citep[][Chapter 8.3, Remark 17]{brezis2011functional}. \\
(iv) Suppose $(Y, d_Y)$ is a metric space (not necessarily compact). Both $\mathscr{B}(Y)$ and $\mathscr{B}(d_Y)$ denote the Borel algebra generated by the metric topology corresponding to $d_Y$. \\
(v) $\{F(z)\}_{z\in Z}$ denotes a function $F$ defined on the set $Z$.\\
The following inequalities are useful for deriving many results presented in this paper
\begin{align}\label{eq: Sobolev embedding from Morrey}
    \Vert f\Vert_{\mathcal{B}}\le \Vert f\Vert_{C^{0,\frac{1}{2}}([0,T])}\le \Tilde{C}_T  \Vert f\Vert_{\mathcal{H}}, \ \ \mbox{ for all }f\in\mathcal{H},
\end{align}
where $\Tilde{C}_T $ is a constant depending only on $T$. The first inequality in Eq.~\eqref{eq: Sobolev embedding from Morrey} results from the definition of $\Vert\cdot\Vert_{C^{0,\frac{1}{2}}([0,T])}$, while the second inequality is from \cite{brezis2011functional} (Corollary 9.14; also see Appendix \ref{proof: simple proof of the Sobolev embedding}). Eq.~\eqref{eq: Sobolev embedding from Morrey} implies the following Sobolev embedding 
\begin{align}\label{eq: H, Holder, B embeddings}
    H_0^1([0,T]) \overset{\operatorname{def}}{=} \mathcal{H} \subseteq C^{0,\frac{1}{2}}([0,T]) \subseteq \mathcal{B} \overset{\operatorname{def}}{=} C([0,T]).
\end{align}


\section{Smooth Euler Characteristic Transform}\label{The Definition of Smooth Euler Characteristic Transform}

In this section, we give the background on the SECT and propose corresponding mathematical foundations. Notably, we specify the ``sample space" --- a collection of shapes on which the SECT is well-defined. The SECT of the shapes in this sample space has properties that are suitable for the probability theory developed in Section \ref{section: distributions of Gaussian bridge}. The molars in the motivating question from Section \ref{section: A Motivation Question} will be modeled as elements of the sample space.


Suppose an o-minimal structure $\mathcal{S}=\{\mathcal{S}_n\}_{n\ge 1}$ satisfying Assumption \ref{Assumption: basic requirements for o-minimal structures of interest} is given, and we focus on shapes in $d$-dimensional space $\mathbb{R}^d$. We assume the shape $K\in\mathcal{S}_d$ is compact and $K\subseteq B(0,R) = \{x\in\mathbb{R}^d: \Vert x\Vert< R\}$, e.g., the $K\subseteq\mathbb{R}^2$ in Figure \ref{fig: SECT_illustration} or the surfaces of the mandibular molars in $\mathbb{R}^3$ as presented by Figure \ref{fig: Teeth}. For each direction $\nu\in\mathbb{S}^{d-1}$, we define a filtration $\{K_t^\nu\}_{t\in[0,T]}$ of sublevel sets by the following (see Figure \ref{fig: SECT_illustration} for an illustration)
\begin{align}\label{eq: def of sublevel sets}
    K_t^\nu  \overset{\operatorname{def}}{=}  \left\{x\in K \,\vert\, x\cdot\nu \le t-R \right\},\ \ \mbox{ for all } t\in[0,T],\ \ \mbox{ where }T \overset{\operatorname{def}}{=} 2R.
\end{align}
We then have the following Euler characteristic curve (ECC, denoted as $\chi_t^\nu$) in direction $\nu$
\begin{align}\label{Eq: first def of Euler characteristic curve}
    \chi_t^\nu(K) \overset{\operatorname{def}}{=} \mbox{ the Euler characteristic of }K_{t}^\nu = \chi(K^\nu_t) =  \sum_{k=0}^{d-1} (-1)^{k}\cdot\beta_k(K_t^\nu),
\end{align}
for $t\in[0,T]$, where $\beta_k(K_t^\nu)$ is the $k$-th Betti number of $K_t^\nu$. The sum in Eq.~\eqref{Eq: first def of Euler characteristic curve} ends at $d-1$ because higher homology groups are trivial \citep[][Section 4]{curry2022many}. If $K_t^\nu$ is a triangle mesh, $\chi (K^\nu_t) =\#V-\#E+\#F$, where $\#V$, $\#E$, and $\#F$ denote the number of vertices, edges, and faces of the mesh, respectively. Due to Theorem \ref{thm: finite triangularization}, the compactness of $K$ guarantees that the Betti numbers in Eq.~\eqref{Eq: first def of Euler characteristic curve} are well-defined and finite.

\begin{figure}[h]
    \centering 
    \includegraphics[scale=0.185]{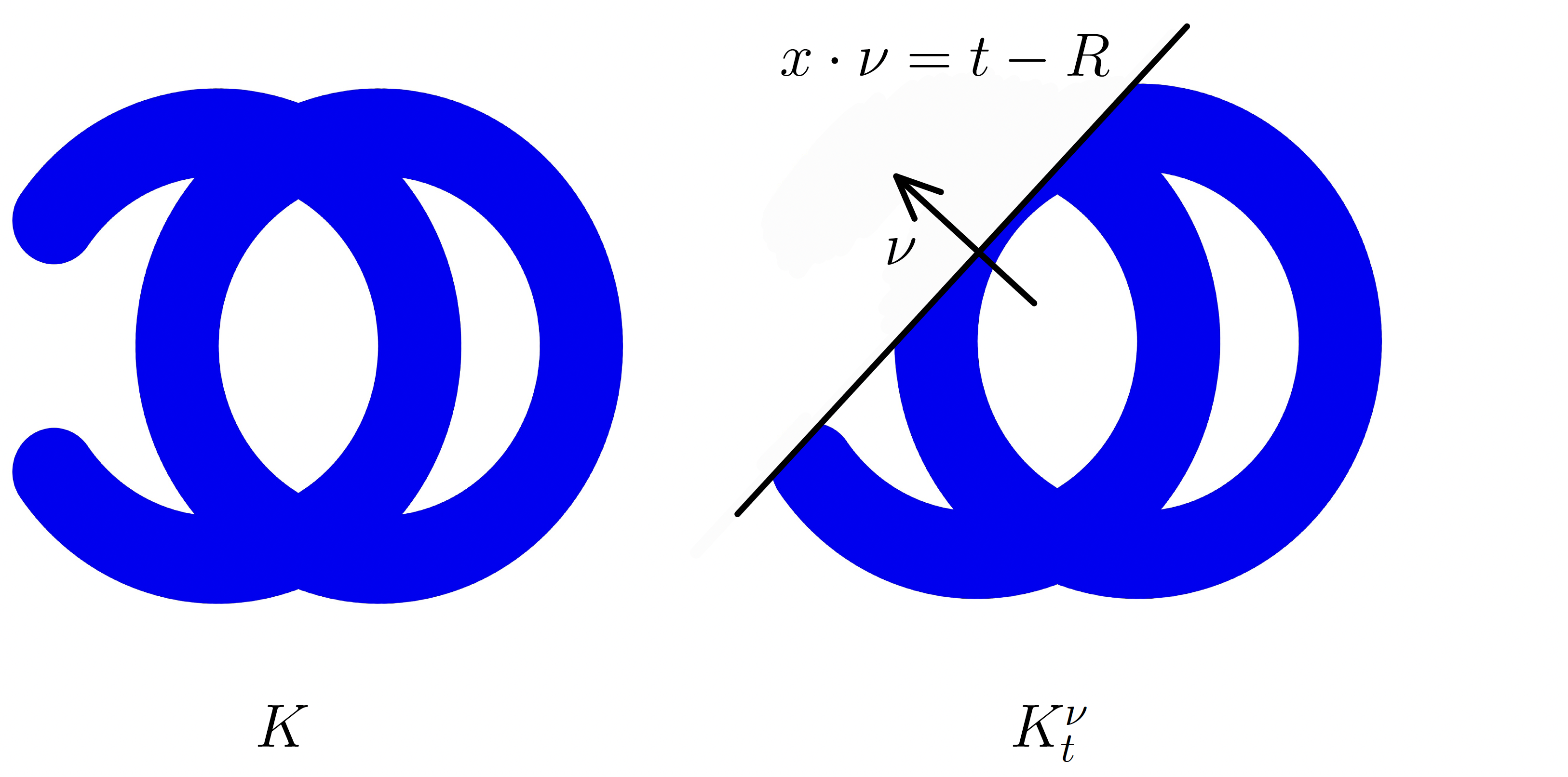}
    \includegraphics[scale=0.15]{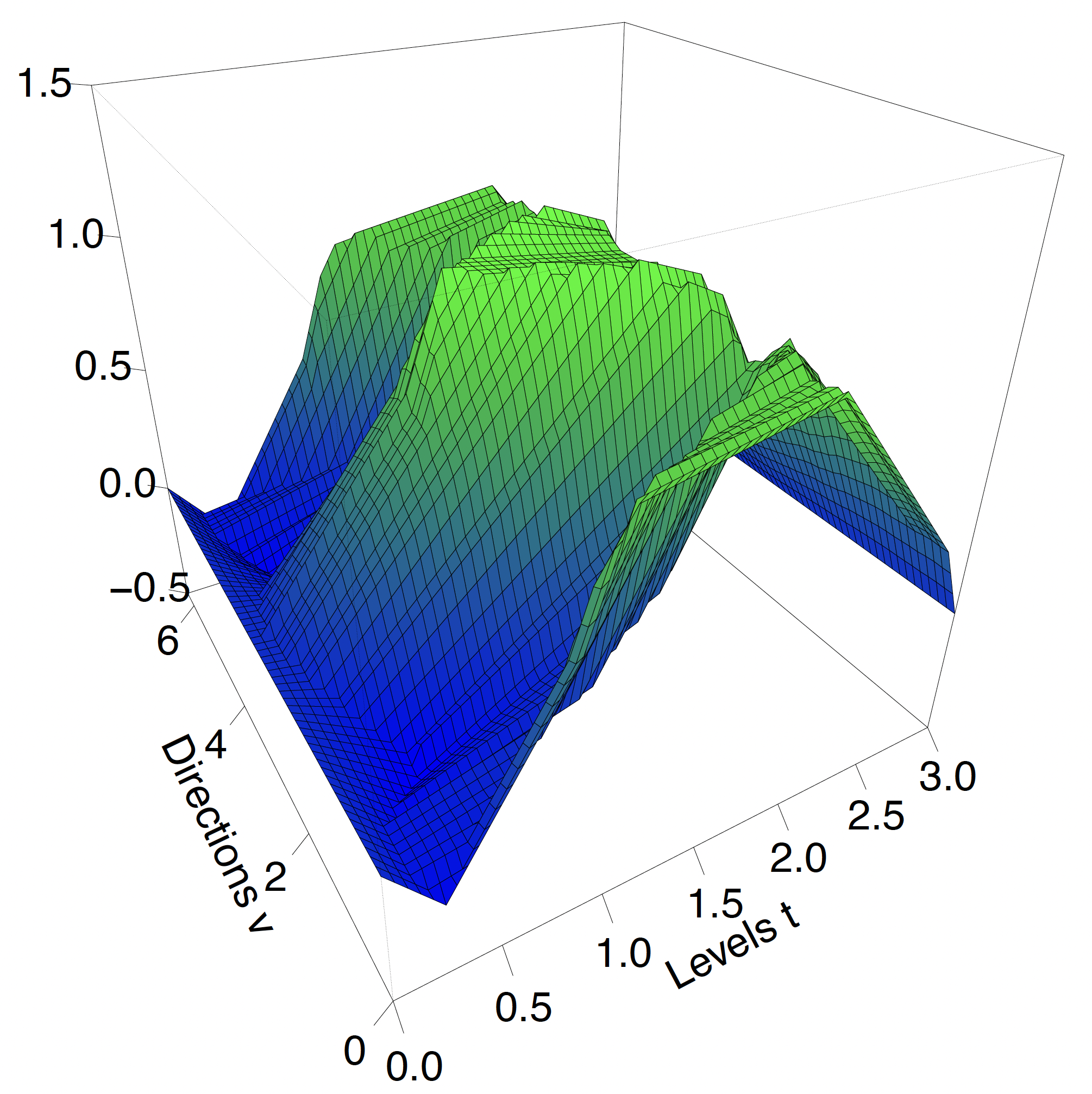}
    \caption{\footnotesize{Consider the 2-dimensional shape $K\in\mathcal{S}_2$ in the left panel. For each pair of $\nu$ and $t$, the equation $x\cdot\nu=t-R$ represents a straight line (or a hyperplane in a high-dimensional space). The subset $K_t^\nu$ denotes the region below this line. Let $\phi_\nu(x)=x\cdot\nu+R$, then $K_t^\nu=\{x\in  K\,\vert\,\phi_\nu(x)\le t\}$. The right panel presents the function $(\nu,t)\mapsto\operatorname{SECT}(K)(\nu,t)$, where $\nu\in\mathbb{S}^1$ is identified by $\theta\in[0,2\pi]$ through $\nu=(\cos\theta,\sin\theta)$. Procedures for generating the shape $K$ and the right panel are given in Appendix \ref{Proof-of-Concept Simulation Examples I: Deterministic Shapes}.}}
    \label{fig: SECT_illustration}
\end{figure}

The Euler characteristic transform (ECT) defined as $\operatorname{ECT}(K): \mathbb{S}^{d-1} \rightarrow \mathbb{Z}^{[0,T]},\, \nu \mapsto \{\chi_t^\nu(K) \}_{t\in[0,T]}$ was proposed by \cite{turner2014persistent} as an alternative to the PHT. Based on the ECT, \cite{crawford2020predicting} further proposed the SECT as follows
\begin{align}\label{Eq: definition of SECT}
\begin{aligned}
    & \operatorname{SECT}(K): \ \ \mathbb{S}^{d-1}\rightarrow\mathbb{R}^{[0,T]},\ \ \ \nu \mapsto \operatorname{SECT}(K)(\nu)  =  \left\{\operatorname{SECT}(K)(\nu, t) \right\}_{t\in[0,T]}, \\
    & \mbox{ where}\ \  \operatorname{SECT}(K)(\nu, t) \overset{\operatorname{def}}{=} \int_0^t \chi_{\tau}^\nu(K) \,d\tau-\frac{t}{T}\int_0^T \chi_{\tau}^\nu(K) \,d\tau.
\end{aligned}
\end{align}
A visualization of the function $(\nu,t)\mapsto\operatorname{SECT}(K)(\nu, t)$ is presented in Figure \ref{fig: SECT_illustration}. The following lemma implies that the Lebesgue integrals in Eq.~\eqref{Eq: definition of SECT} are well-defined.
\begin{lemma}\label{thm: tameness property}
For any fixed $K\in\mathcal{S}_d$ and $\nu\in\mathbb{S}^{d-1}$, the function $t\mapsto\chi(K_t^\nu)$ is piecewise constant with only finitely many discontinuities.
\end{lemma}
\noindent Through the ``cell decomposition theorem" \citep[][Chapter 3]{van1998tame}, Lemma \ref{thm: tameness property} directly follows from either Lemma 3.4 of \cite{curry2022many} or ``(2.10) Proposition" in Chapter 4 of \cite{van1998tame}. Hence, the proof of Lemma \ref{thm: tameness property} is omitted.


To investigate the distribution of $\operatorname{SECT}(K)$ over different shapes $K$, we introduce the following condition to restrict our attention to a subset of $\mathcal{S}_d$.
\begin{condition}\label{condition: the condition for defining S_{R,d}^M}
     Let $K\in\mathcal{S}_d$. The condition is that $K$ satisfies the following inequality
        \begin{align}\label{Eq: topological invariants boundedness condition}
    \sup_{k\in\{0,\,\cdots,\, d-1\}}\left[\sup_{\nu\in\mathbb{S}^{d-1}}\left(\#\Big\{\xi\in \operatorname{Dgm}_k(K;\phi_{\nu}) \, \Big\vert \, \operatorname{pers}(\xi)>0\Big\}\right)\right] \le\frac{M}{d},
\end{align}
where $\operatorname{Dgm}_k(K;\phi_{\nu})$ is the PD of $K$ associated with the function $\phi_\nu(x)=x\cdot \nu+R$ (also see Figure \ref{fig: SECT_illustration}), $\operatorname{pers}(\xi)$ is the persistence of the homology feature $\xi$, $\#\{\cdot\}$ denotes the cardinality of a multiset, and $M>0$ is a sufficiently large prespecified number.
\end{condition}
\noindent Condition \ref{condition: the condition for defining S_{R,d}^M} involves technicalities from computational topology \citep{edelsbrunner2010computational}. To maintain the flow of the paper, we relegate the details of this condition, as well as the definitions of $\operatorname{Dgm}_k(K;\phi_{\nu})$ and $\operatorname{pers}(\xi)$, to Appendix \ref{The Relationship between PHT and SECT}. Heuristically, Condition \ref{condition: the condition for defining S_{R,d}^M} implies the existence of a uniform upper bound on the number of nontrivial homology features of $K$ across all directions $\nu$. Hereafter, we focus on shapes in the following collection
\begin{align*}
    \mathcal{S}_{R,d}^M \overset{\operatorname{def}}{=} \left\{K \in\mathcal{S}_d \,\big\vert\, K\subseteq  B(0,R) \mbox{ is compact and satisfies Condition \ref{condition: the condition for defining S_{R,d}^M} with fixed }M>0 \right\}.
\end{align*}

Our proposed collection $\mathcal{S}_{R,d}^M$ is suitable for modeling shapes in many applications. For example, the surfaces of the molars in Figure \ref{fig: Teeth} are compact subsets of $\mathbb{R}^3$, bounded by a common open ball, and can be approximately represented by triangle meshes (hence, modeled by an o-minimal structure satisfying Assumption \ref{Assumption: basic requirements for o-minimal structures of interest}). In addition, the four genera of primates in Figure \ref{fig: Teeth} share a phylogentic relationship which implies that their molars have common baseline features and satisfy Condition \ref{condition: the condition for defining S_{R,d}^M} with a shared upper bound $M$. In each application, the dimension $d$ and radius $R$ of the ball $B(0,R)$ can easily be determined based on observed shapes. Although our mathematical framework requires the existence of such an $M$ in Eq.~\eqref{Eq: topological invariants boundedness condition}, the value of $M$ is not needed for our statistical methodology (see Section \ref{section: hypothesis testing}). Thus, Condition \ref{condition: the condition for defining S_{R,d}^M} does not hinder our proposed statistical methodology. 

Lemma \ref{thm: tameness property} implies that the function $\{\chi_t^\nu(K)\}_{t\in[0,T]}$ of $t$ belongs to $L^1([0,T])$. Therefore, the function $\operatorname{SECT}(K)(\nu)=\{\operatorname{SECT}(K)(\nu,t)\}_{t\in[0,T]}$ of $t$ is absolutely continuous on $[0,T]$. Furthermore, we have the following regularity result of the Sobolev type. 
\begin{lemma}\label{thm: Sobolev function paths}
For any $K\in\mathcal{S}_{R,d}^M$ and $\nu\in\mathbb{S}^{d-1}$, the function $\operatorname{SECT}(K)(\nu)$ belongs to $\mathcal{H}$.
\end{lemma}
\noindent Lemma \ref{thm: Sobolev function paths} is a special case of Lemma \ref{thm: Sobolev function paths; general, appendix}. It indicates $\operatorname{SECT}(\mathcal{S}_{R,d}^M) \subseteq \mathcal{H}^{\mathbb{S}^{d-1}}  =  \{\mbox{all maps }F:\mathbb{S}^{d-1}\rightarrow\mathcal{H}\}$, which is enhanced by the following result.
\begin{theorem}\label{lemma: The continuity lemma}
For each $K\in\mathcal{S}_{R,d}^M$, we have: (i) There exists a constant $C^*_{M,R,d}$ depending only on $M$, $R$, and $d$ such that the following inequality holds for any directions $\nu_1,\nu_2\in\mathbb{S}^{d-1}$, 
\begin{align}\label{Eq: continuity inequality}
    \begin{aligned}
     & \left\Vert\, \operatorname{SECT}(K)(\nu_1) - \operatorname{SECT}(K)(\nu_2) \,\right\Vert_{\mathcal{H}} \le C^*_{M,R,d} \cdot \sqrt{ \Vert \nu_1 - \nu_2\Vert + \Vert 
    \nu_1-\nu_2 \Vert^2 };
    \end{aligned}
\end{align}
and (ii) $\operatorname{SECT}(K) \in C^{0,\frac{1}{2}}(\mathbb{S}^{d-1};\mathcal{H})$, where $\mathbb{S}^{d-1}$ is equipped with the geodesic distance $d_{\mathbb{S}^{d-1}}$.
\end{theorem}
\noindent Results complementary to Theorem \ref{lemma: The continuity lemma} can be found in Theorem \ref{lemma: The continuity lemma; Appendix}, which imply that the function $(\nu,t)\mapsto \operatorname{SECT}(K)(\nu, t)$ belongs to $C^{0,\frac{1}{2}}(\mathbb{S}^{d-1}\times[0,T])$. Theorem \ref{lemma: The continuity lemma}(i) is an analog of Lemma 2.1 in \cite{turner2014persistent}. Theorem \ref{lemma: The continuity lemma}(ii) implies $\operatorname{SECT}(\mathcal{S}_{R,d}^M) \subseteq C^{0,\frac{1}{2}}(\mathbb{S}^{d-1}; \mathcal{H}) \subseteq C(\mathbb{S}^{d-1}; \mathcal{H}) \subseteq \mathcal{H}^{\mathbb{S}^{d-1}}$. As a result, Eq.~\eqref{Eq: definition of SECT} defines the following map
\begin{align}\label{Eq: final def of SECT}
    \operatorname{SECT}: \, \mathcal{S}_{R,d}^M \rightarrow C(\mathbb{S}^{d-1};\, \mathcal{H}),\ \ \ K \mapsto \operatorname{SECT}(K).
\end{align}
In Appendix \ref{Proof-of-Concept Simulation Examples I: Deterministic Shapes}, we provide detailed proof-of-concept examples (similar to Figure \ref{fig: SECT_illustration}) to visually illustrate the $\operatorname{SECT}$ and support the regularity results in Theorems \ref{lemma: The continuity lemma} and \ref{lemma: The continuity lemma; Appendix}.

Corollary 1 of \cite{ghrist2018persistent} implies the following result, which shows that the SECT preserves all the information of shapes $K\in \mathcal{S}_{R,d}^M$.
\begin{theorem}\label{thm: invertibility}
The map $\operatorname{SECT}$ defined in Eq.~\eqref{Eq: final def of SECT} is injective for all dimensions $d$.
\end{theorem}
\noindent The map \(\operatorname{SECT}: \mathcal{S}_{R,d}^M \rightarrow C(\mathbb{S}^{d-1}; \mathcal{H})\) is injective, but not surjective. Specifically, Theorem \ref{lemma: The continuity lemma} suggests that the image of \(\operatorname{SECT}\) does not lie outside of \(C^{0,\frac{1}{2}}(\mathbb{S}^{d-1}; \mathcal{H})\). An explicit characterization of the image \(\operatorname{SECT}(\mathcal{S}_{R,d}^M)\) remains a topic for future research. 

Inspired by Theorem \ref{thm: invertibility}, one may consider reconstructing a shape \(K\) from either the \(\operatorname{SECT}(K)\) or the \(\operatorname{ECT}(K)\). From a theoretical standpoint, a shape \(K\) can be reconstructed using the ``Schapira's inversion formula'' \citep{schapira1995tomography}. Further details are available in \cite{ghrist2018persistent}. From an algorithmic perspective, the proof of Theorem 3.1 in \cite{turner2014persistent} offers an algorithm to reconstruct low-dimensional meshes from their ECT. Nevertheless, effective algorithmic approaches to reconstructing shapes are still underdeveloped. Challenges in reconstructing shapes are extensively discussed in \cite{fasy2018challenges}. A comprehensive exploration of the reconstruction using SECT is also left for future research.

Together with Eq.~\eqref{Eq: final def of SECT}, Theorem \ref{thm: invertibility} allows us to represent each $K\in\mathcal{S}_{R,d}^M$ by $\operatorname{SECT}(K) \in C(\mathbb{S}^{d-1}; \mathcal{H})$. This perspective aids us in modeling the randomness of shapes using probability measures on the separable Banach space $C(\mathbb{S}^{d-1}; \mathcal{H})$. Here, we prefer $C(\mathbb{S}^{d-1}; \mathcal{H})$ over $\frac{1}{2}$-Hölder space $C^{0,\frac{1}{2}}(\mathbb{S}^{d-1}; \mathcal{H})$. This is because $\frac{1}{2}$-Hölder spaces are typically not separable \citep[][Remark 4.21]{hairer2009introduction}. The separability condition is essential for probability measures on Banach spaces to exhibit non-pathological behavior \citep[][Section 4]{hairer2009introduction}.


\section{Probabilistic Distributions over the SECT}\label{section: distributions of Gaussian bridge}

To address the motivating question outlined in Section \ref{section: A Motivation Question} using hypothesis testing, we need to view the observed shapes (e.g., the molars in Figure \ref{fig: Teeth}) as shape-valued random variables. In this section, we construct a probability space to model the randomness of shapes and make the SECT a random variable (in the measurable sense) taking values in $C(\mathbb{S}^{d-1}; \mathcal{H})$. More importantly, this probability space helps justify the KL expansion of the SECT, which lays the foundations for our hypothesis testing method in Section \ref{section: hypothesis testing}.

\noindent\textbf{Probability Space.} Suppose $\mathcal{S}_{R,d}^M$ is equipped with a $\sigma$-algebra $\mathscr{F}$. A distribution of shapes $K$ across $\mathcal{S}_{R,d}^M$ is represented by a probability measure $\mathbb{P}=\mathbb{P}(dK)$ on $\mathscr{F}$. Then, $(\mathcal{S}_{R,d}^M, \mathscr{F}, \mathbb{P})$ is a probability space. For each fixed $(\nu,t)$, the integer-valued map $\chi^\nu_t: K \mapsto \chi^\nu_t(K)$ is defined on $\mathcal{S}_{R,d}^M$. Hereafter, we assume the following:

\begin{assumption}\label{assumption: the measurability of ECC}
For each fixed $(\nu,t)\in\mathbb{S}^{d-1}\times [0,T]$, the map $\chi^\nu_t:\, (\mathcal{S}^M_{R,d},\, \mathscr{F} ) \rightarrow \left(\mathbb{R}, \, \mathscr{B}(\mathbb{R}) \right)$ is a measurable function and, hence, a real-valued random variable.
\end{assumption}

A $\sigma$-algebra $\mathscr{F}$ satisfying Assumption \ref{assumption: the measurability of ECC} exists. Here, we construct a metric $\rho$ on $\mathcal{S}_{R,d}^M$ and show that the Borel algebra $\mathscr{B}(\rho)$ induced by $\rho$ satisfies Assumption \ref{assumption: the measurability of ECC}. We define
\begin{align}\label{Eq: distance between shapes}
 \rho(K_1,\, K_2) \overset{\operatorname{def}}{=} \sup_{\nu\in\mathbb{S}^{d-1}} \left\{ \left( \int_0^T \Big\vert \, \chi_\tau^{\nu}(K_1)-\chi_\tau^{\nu}(K_2) \, \Big\vert^2 d\tau \right)^{1/2} \right\},\ \ \mbox{ for all }K_1, K_2 \in \mathcal{S}_{R,d}^M.
 \end{align}
\begin{theorem}\label{Thm: metric theorem for shapes}
The map $\rho$ is a metric on $\mathcal{S}_{R,d}^M$. Assumption \ref{assumption: the measurability of ECC} is satisfied if $\mathscr{F}=\mathscr{B}(\rho)$.
\end{theorem}

Under Assumption \ref{assumption: the measurability of ECC}, the ECC $\{\chi^\nu_t\}_{t\in[0,T]}$, for each $\nu\in\mathbb{S}^{d-1}$, is a stochastic process defined on the probability space $(\mathcal{S}_{R,d}^M, \mathscr{F}, \mathbb{P})$. Since each sample path $\{\chi^\nu_t(K)\}_{t\in[0,T]}$ has finitely many discontinuities (Lemma \ref{thm: tameness property}), $\int_0^t \chi_\tau^\nu(K) \,d\tau$ for each $t\in[0,T]$ is a Riemann integral, which is equal to the limit of Riemann sum $\int_0^t \chi_\tau^\nu(K) \,d\tau = \lim_{n\rightarrow\infty} \left\{\frac{t}{n} \sum_{l=1}^n \chi^\nu_{\frac{tl}{n}}(K)\right\}$. Given that each $\chi^\nu_{\frac{tl}{n}}$ is a random variable under Assumption \ref{assumption: the measurability of ECC}, the limit of the Riemann sum for each $t\in[0,T]$ is a random variable as well. Therefore, for each $\nu\in\mathbb{S}^{d-1}$, $\{\int_0^t \chi_\tau^\nu d\tau\}_{t\in[0,T]}$ with $\int_0^t \chi_\tau^\nu d\tau: K \mapsto \int_0^t \chi_\tau^\nu(K) d\tau$ is a stochastic process. Then, under Assumption \ref{assumption: the measurability of ECC}, Eq.~\eqref{Eq: definition of SECT} defines the following stochastic process on $(\mathcal{S}_{R,d}^M, \mathscr{F}, \mathbb{P})$ for each $\nu\in\mathbb{S}^{d-1}$
\begin{align}\label{Eq: def SECTs as stochastic processes}
\begin{aligned}
    & \operatorname{SECT}(\nu) \overset{\operatorname{def}}{=} \left\{\int_0^t\chi_\tau^\nu d\tau - \frac{t}{T} \int_0^T \chi_\tau^\nu d\tau   \overset{\operatorname{def}}{=} \operatorname{SECT}(\nu, t) \right\}_{t\in[0,T]}.
    \end {aligned}
\end{align}
Precisely, for each fixed $\nu$, we have the stochastic process $\operatorname{SECT}(\nu):K \mapsto\operatorname{SECT}(K)(\nu)=\{\operatorname{SECT}(K)(\nu,t)\}_{t\in[0,T]}$ defined on $(\mathcal{S}_{R,d}^M, \mathscr{F}, \mathbb{P})$; and, for each fixed $(\nu,t)$, we have the real-valued random variable $\operatorname{SECT}(\nu,t): K\mapsto\operatorname{SECT}(K)(\nu,t)$ defined on $(\mathcal{S}_{R,d}^M, \mathscr{F}, \mathbb{P})$.

Since $\mathcal{H}$ is an RKHS (Appendix \ref{section: notation for closed vs. open}), Lemma \ref{thm: Sobolev function paths} and Theorem \ref{lemma: The continuity lemma}, together with Theorem 7.1.2 of \cite{hsing2015theoretical}, imply the following. Its proof is omitted.
\begin{theorem}\label{thm: SECT distribution theorem in each direction}
(i) For each $\nu\in\mathbb{S}^{d-1}$, $\operatorname{SECT}(\nu)$ is a real-valued stochastic process with sample paths in $\mathcal{H}$. Equivalently, $\operatorname{SECT}(\nu)$ is a random variable taking values in $(\mathcal{H}, \mathscr{B}(\mathcal{H}))$. (ii) The map $\operatorname{SECT}$ defined in Eq.~\eqref{Eq: final def of SECT} is a random variable taking values in $C(\mathbb{S}^{d-1};\mathcal{H})$.
\end{theorem}
\noindent Using Theorem \ref{thm: SECT distribution theorem in each direction} in conjunction with Theorem \ref{thm: invertibility}, we can represent random shapes (which model the surfaces of the mandibular molars in Figure \ref{fig: Teeth}) as $C(\mathbb{S}^{d-1},\mathcal{H})$-valued random variables. This representation through the SECT has no loss of information.

In Appendix \ref{section: Proof-of-Concept Simulation Examples II: Random Shapes}, we provide proof-of-concept examples to illustrate random shapes and their SECT representations visually. These examples relate the SECT to Fréchet regression \citep{petersen2019frechet}, Wasserstein regression \citep{chen2021wasserstein}, and manifold learning \citep{dunson2021inferring, meng2021principal, li2022efficient}.

\noindent\textbf{Mean and Covariance of the SECT.} For deriving the KL expansion in Section \ref{section: hypothesis testing}, we define the mean and covariance of the SECT. To do so, we need the following lemma.
\begin{lemma}\label{assumption: existence of second moments}
For any probability measure $\mathbb{P}$ defined on the measurable space $(\mathcal{S}_{R,d}^M, \mathscr{F})$, we have $\mathbb{E}\left\{ \sup_{\nu\in\mathbb{S}^{d-1}} \Vert \operatorname{SECT}(\nu)\Vert^2_{\mathcal{H}} \right\}=\int_{\mathcal{S}_{R,d}^M} \left\{ \sup_{\nu\in\mathbb{S}^{d-1}} \Vert \operatorname{SECT}(K)(\nu)\Vert^2_{\mathcal{H}} \right\} \,\mathbb{P}(dK)<\infty.$
\end{lemma}
\noindent Lemma \ref{assumption: existence of second moments}, together with Eq.~\eqref{eq: Sobolev embedding from Morrey}, implies that $\mathbb{E}\vert \operatorname{SECT}(\nu, t)\vert^2 \le \Tilde{C}^2_T \cdot \mathbb{E}\Vert \operatorname{SECT}(\nu)\Vert^2_{\mathcal{H}}<\infty$ for all $(\nu,t)\in\mathbb{S}^{d-1}\times[0,T]$. Then, we define the mean and covariance functions as follows
\begin{align}\label{Eq: mean and covariance functions}
    \begin{aligned}
        & m_\nu(t)  =  \mathbb{E}\left\{ \operatorname{SECT}(\nu, t) \right\} = \int_{\mathcal{S}_{R,d}^M} \operatorname{SECT}(K)(\nu,t) \,\mathbb{P}(dK),\\
        & \Xi_\nu(s,t)  =  \operatorname{Cov}\Big(\operatorname{SECT}(\nu, s), \, \operatorname{SECT}(\nu, t)\Big),\ \ \mbox{ for }s,t\in[0,T] \text{ and }\nu\in\mathbb{S}^{d-1}.
    \end{aligned}
\end{align}
Lemma \ref{thm: mean is in H} provides several properties of the mean $m_\nu$ and covariance $\Xi_\nu$ that validate our KL expansion of $\operatorname{SECT}(\nu)$ in Section \ref{section: hypothesis testing}. Additionally, Lemma \ref{thm: mean is in H} demonstrates that the mean $\pmb{m}\overset{\operatorname{def}}{=}\{m_\nu\}_{\nu\in\mathbb{S}^{d-1}}$ of SECT belongs to $C(\mathbb{S}^{d-1};\mathcal{H})$. A tentative discussion on the ``pseudo-inverse" $\operatorname{SECT}^{-1}(\pmb{m})$ is provided after Lemma \ref{thm: mean is in H} in Appendix \ref{section: Further Theorems}.

In most shape analysis studies, data are preprocessed by alignment. In Appendix \ref{section: ECT Alignment}, we introduce the ``ECT alignment" as a preprocessing step before any statistical inference. Throughout the manuscript, we assume that the data have been aligned using this method. The ECT alignment exploits rigid motions, does not rely on landmarks, and is equivalent to the alignment approach outlined in \cite{wang2021statistical} (Supplementary Section 4). The primary objective of the ECT alignment is to minimize the differences between two shapes that arise from rigid motions. For instance, the molars in Figure \ref{fig: Teeth} were aligned using the ECT alignment. Furthermore, the ECT alignment is compatible with our SECT framework. Appendix \ref{section: ECT Alignment} demonstrates that the ECT alignment does not alter the qualitative properties of SECT (e.g., the measurability, Sobolev-regularity, and $\frac{1}{2}$-Hölder continuity).

In applications, it is infeasible to sample infinitely many directions $\nu \in \mathbb{S}^{d-1}$ and levels $t \in [0,T]$. For given shapes $K$, we compute $\operatorname{SECT}(K)(\nu, t)$ for finitely many directions $\{\nu_1, \cdots, \nu_\Gamma\} \subseteq \mathbb{S}^{d-1}$ and levels $\{t_1, \cdots, t_\Delta\} \subseteq [0,T]$. To retain information about shapes $K$, one needs to properly set the numbers of directions and levels (i.e., $\Gamma$ and $\Delta$). From a theoretical viewpoint, \cite{curry2022many} comprehensively discussed the number $\Gamma$ of directions needed to recover shapes $K$ from $\operatorname{ECT}(K)$ when $K$ are ``piecewise linearly embedded shapes with plausible geometric bounds.'' From the numerical perspective, we note the following: (i) \cite{wang2021statistical} provided detailed simulation studies on the choices of $\Gamma$ and $\Delta$ in sub-image analysis, and a general guideline for setting $\Gamma$ and $\Delta$ in practice was presented in Supplementary Table 1 therein; and (ii) in our Appendix \ref{section: Runtime}, we provide detailed numerical experiments on the trade-offs between the choices of $\Gamma$ and $\Delta$, the statistical power of our proposed algorithms (Algorithms \ref{algorithm: testing hypotheses on mean functions; Appendix} and \ref{algorithm: permutation-based testing hypotheses on mean functions}), and computational cost.

\section{Testing Hypotheses on Shapes}\label{section: hypothesis testing}

In this section, we apply the probabilistic formulation from Section \ref{section: distributions of Gaussian bridge} and Lemma \ref{thm: mean is in H} to test hypotheses on shapes. Suppose $\mathbb{P}^{(1)}$ and $\mathbb{P}^{(2)}$ are two distributions on the measurable space $(\mathcal{S}_{R,d}^M, \mathscr{F})$. Let $\mathbb{P}^{(1)}\otimes \mathbb{P}^{(2)}$ be the product probability measure defined on the product $\sigma$-algebra $\mathscr{F}\otimes\mathscr{F}$, satisfying $\mathbb{P}^{(1)}\otimes \mathbb{P}^{(2)}(A\times B)=\mathbb{P}^{(1)}(A)\cdot \mathbb{P}^{(2)}(B)$ for all $A,B\in\mathscr{F}$. To address the motivating question from Section \ref{section: A Motivation Question}, we test the following hypotheses
\begin{align}\label{eq: original hypotheses}
    H_0^*:\ \ \mathbb{P}^{(1)} = \mathbb{P}^{(2)},\ \ \ vs. \ \ \ H_1^*: \ \ \mathbb{P}^{(1)} \ne \mathbb{P}^{(2)},
\end{align}
e.g., suppose $\mathbb{P}^{(1)}$ and $\mathbb{P}^{(2)}$ model the distributions of molars from two genera of primates (Figure \ref{fig: Teeth}). Rejecting the $H_0^*$ in Eq.~\eqref{eq: original hypotheses} helps distinguish the two genera of primates.

Define $m_\nu^{(j)}(t)  =  \int_{\mathcal{S}_{R,d}^M} \operatorname{SECT}(K)(\nu, t) \,\mathbb{P}^{(j)}(dK)$ for $j\in\{1,2\}$ as the mean functions corresponding to $\mathbb{P}^{(1)}$ and $\mathbb{P}^{(2)}$. To reject the null $H_0^*$ in Eq.~\eqref{eq: original hypotheses} (equivalently, distinguish two collections of shapes), it suffices to reject the null hypothesis $H_0$ in the following
\begin{align}\label{eq: the main hypotheses}
    \begin{aligned}
        & H_0: m_\nu^{(1)}(t)=m_\nu^{(2)}(t)\mbox{ for all }(\nu,t),\ \ \ vs.\ \ \ H_1:  m_\nu^{(1)}(t)\ne m_\nu^{(2)}(t)\mbox{ for some }(\nu,t).
    \end{aligned}
\end{align}

\noindent\textbf{Analysis of Variance for Functional Data (fdANOVA).} Considering the hypotheses in Eq.~\eqref{eq: the main hypotheses} for all directions $\nu\in\mathbb{S}^{d-1}$ results in simultaneous multiple-comparisons and inflation of the type I error. To address this issue, we focus on a specific direction, motivated by the observation that the null hypothesis $H_0$ in Eq.~\eqref{eq: the main hypotheses} is equivalent to $\sup_{\nu\in\mathbb{S}^{d-1}}\{\Vert m_{\nu}^{(1)}-m_{\nu}^{(2)} \Vert_{\mathcal{B}}\}=0$. Hence, the direction of interest is defined as
\begin{align}\label{eq: def of distinguishing direction}
    \nu^* \overset{\operatorname{def}}{=} \argmax_{\nu\in\mathbb{S}^{d-1}} \left\{\Vert m_{\nu}^{(1)}-m_{\nu}^{(2)} \Vert_{\mathcal{B}} \right\}.
\end{align}
Lemma \ref{thm: mean is in H} and Eq.~\eqref{eq: H, Holder, B embeddings} imply $\{m_\nu^{(j)}\}_{j=1}^2\subseteq\mathcal{B}$ for all $\nu$. Lemma \ref{thm: mean is in H}, together with Eq.~\eqref{eq: Sobolev embedding from Morrey}, confirms the existence of a maximizer in Eq.~\eqref{eq: def of distinguishing direction}. The maximizer in Eq.~\eqref{eq: def of distinguishing direction} may not be unique. If there are multiple maximizers, we arbitrarily choose one, as this choice does not influence our framework. The null hypothesis \( H_0 \) in Eq.~\eqref{eq: the main hypotheses} is then equivalent to \( \Vert m_{\nu^*}^{(1)} - m_{\nu^*}^{(2)} \Vert_{\mathcal{B}} = 0 \), where the \( \nu^* \) defined in Eq.~\eqref{eq: def of distinguishing direction} is called a \textit{distinguishing direction}. Hereafter, we investigate the distribution of \( \operatorname{SECT}(\nu^*) \).

Based on the discussion above, testing the hypotheses in Eq.~\eqref{eq: the main hypotheses} is equivalent to testing $m_{\nu^*}^{(1)}(t) = m_{\nu^*}^{(2)}(t)$ for $t\in[0,T]$ using $\operatorname{SECT}(\nu^*)$, which is a fdANOVA problem that has been well-studied in the literature \citep[e.g.,][Section 5.2]{zhang2013analysis}. However, many state-of-the-art fdANOVA approaches are incompatible with $\operatorname{SECT}(\nu^*)$. For example, the Gaussianity of $\operatorname{SECT}(\nu^*)$ is not guaranteed (Remark \ref{remark: The Gaussianity of the SECT is not guaranteed.}), and the ``two-sample problem assumptions" in Section 5.2 of \cite{zhang2013analysis} may not be satisfied. Besides, the $L^2$-norm-based test \citep{zhang2007statistical} and F-type test \citep{shen2004f} are not preferred when the functional data are not Gaussian \citep[][Chapter 5]{zhang2013analysis}. Additionally, many fdANOVA methods are time-consuming. For example, tests based on random projections \citep[TRP,][]{cuesta2010simple} require the computation of (at least 30) $L^2$-projections for each observed function, followed by the application of appropriate ANOVA tests to these projections. To address the Gaussianity issue and achieve computational efficiency, we propose a method for fdANOVA using the KL expansion. Our test has a foundation that aligns with the probabilistic framework of $\operatorname{SECT}(\nu^*)$ in Section \ref{section: distributions of Gaussian bridge}; it is comparable with the existing methods in terms of size and power (see Appendix \ref{appendix: Numerical Experiments on One-way ANOVA --- Existing Methods vs. Our Proposed Methods}); and it is also computationally efficient, allowing for the permutation test used with our method.

\noindent\textbf{Karhunen–Loève Expansion.} Let $\Xi_{\nu^*}^{(j)}(s,t)$ be the covariance function of the stochastic process $\operatorname{SECT}(\nu^*)$ corresponding to $\mathbb{P}^{(j)}$, for $j\in\{1,2\}$ (see Eq.~\eqref{Eq: mean and covariance functions}). Hereafter, we assume the following, which is true under the null hypothesis \( H_0^*: \mathbb{P}^{(1)} = \mathbb{P}^{(2)} \) in Eq.~\eqref{eq: original hypotheses}.
\begin{assumption}[Homoscedasticity]\label{assumption: equal covariance assumption}
$\Xi_{\nu^*} \overset{\operatorname{def}}{=}\Xi_{\nu^*}^{(1)}=\Xi_{\nu^*}^{(2)}$, where $\nu^*$ is defined in Eq.~\eqref{eq: def of distinguishing direction}.
\end{assumption}
\noindent This is a standard assumption in the fdANOVA literature \citep[e.g.,][Section 5.2]{zhang2013analysis} and can be tested using the methods proposed by \cite{guo2018new}. 

We define an integral operator on $L^2([0,T]^2)$ as $f\mapsto \int_0^T f(s)\cdot\Xi_{\nu^*}(s,\cdot\,) \,ds$. This operator is compact and self-adjoint \citep[][Theorems 4.6.2 and Example 3.3.4]{hsing2015theoretical}. Moreover, the Hilbert-Schmidt theorem \citep[][Theorem VI.16]{reed2012methods} suggests that there is a complete orthonormal basis $\{\phi_l\}_{l=1}^\infty$ for $L^2([0,T])$ so that (i) each $\phi_l$ is an eigenfunction with eigenvalue $\lambda_l$, (ii) $\lambda_1\ge\lambda_2\ge\cdots\ge0$, and (iii) $\lim_{l\rightarrow\infty}\lambda_l=0$. Lemma \ref{thm: mean is in H} and Theorem 7.3.5 of \cite{hsing2015theoretical} imply the following KL expansion:
\begin{theorem}[Karhunen–Loève expansion]\label{thm: KL expansions of SECT}
(i) For each fixed $j\in\{1,2\}$, we have
\begin{align}\label{eq: rigorous KL expansions of SECT}
    \begin{aligned}
        \lim_{L\rightarrow\infty}\sup_{t\in[0,T]}\mathbb{E}^{(j)}\left[ \, \operatorname{SECT}(\nu^*,\,t) - m^{(j)}_{\nu^*}(t) - \sum_{l=1}^L \sqrt{\lambda_l} \cdot Z_{l}^{(j)} \cdot \phi_l(t) \, \right]^2 =0,
    \end{aligned}
\end{align}
where $Z_{l}^{(j)}(K) \overset{\operatorname{def}}{=} \frac{1}{\sqrt{\lambda_l}}\int_0^T \{\operatorname{SECT}(K)(\nu^*,\, t)-m_{\nu^*}^{(j)}(t) \}\cdot \phi_l(t) \,dt$ for $l=1,2,\ldots$, and $\mathbb{E}^{(j)}$ is the expectation associated with $\mathbb{P}^{(j)}$. For each $j\in\{1,2\}$, the random variables $\{Z_{l}^{(j)}\}_{l=1}^\infty$ are defined on the probability space $(\mathcal{S}_{R,d}^M, \mathscr{F}, \mathbb{P}^{(j)})$, are mutually uncorrelated, and have mean 0 and variance 1. (ii)
There exists $\mathcal{N}\in\mathscr{F}\otimes \mathscr{F}$ so that $\mathbb{P}^{(1)}\otimes \mathbb{P}^{(2)}(\mathcal{N})=0$ and
\begin{align}\label{eq: KL expansions of SECT}
    \begin{aligned}
        &  \delta_l\left(K^{(1)}, \, K^{(2)}\right) \overset{\operatorname{def}}{=} \frac{1}{\sqrt{2 \lambda_l}} \int_0^T \left\{ \operatorname{SECT}(K^{(1)})(\nu^*,\,t)-\operatorname{SECT}(K^{(2)})(\nu^*,\,t)\right\}\cdot \phi_l(t) \,dt \\
        &\ \ \ \ \ \ \ \ \ \ \ \ \ \ \ \ \ \, = \theta_l + \left( \frac{Z_{l}^{(1)}(K^{(1)})-Z_{l}^{(2)}(K^{(2)})}{\sqrt{2}} \right),\\ & \mbox{where }\theta_l  \overset{\operatorname{def}}{=} \frac{1}{\sqrt{2 \lambda_l}} \int_0^T \left\{ m_{\nu^*}^{(1)}(t) - m_{\nu^*}^{(2)}(t) \right\}\cdot \phi_l(t) \,dt,
    \end{aligned}
\end{align}
for any $(K^{(1)},\, K^{(2)})\notin \mathcal{N}$ and each fixed $l=1,2,\ldots$. The null set $\mathcal{N}$ is allowed to be empty.
\end{theorem}

\noindent Using the KL expansion in Eq.~\eqref{eq: rigorous KL expansions of SECT}, the random sampling of shapes may be considered, which is discussed in Appendix \ref{appendix: Generative Models for Complex Shapes} and left for future research.

\noindent\textbf{Our Approach.} Consider two independent collections of random shapes $\{K_i^{(j)}\}_{i=1}^n\overset{\operatorname{i.i.d.}}{\sim}\mathbb{P}^{(j)}$, for $j\in\{1,2\}$ (i.e., $\{(K_i^{(1)}, K_i^{(2)})\}_{i=1}^n\overset{\operatorname{i.i.d.}}{\sim}\mathbb{P}^{(1)}\otimes \mathbb{P}^{(2)}$). The pairing in $(K_i^{(1)}, K_i^{(2)})$ is arbitrary for the following reasons: (i) pairs $(K_i^{(1)}, K_i^{(2)})$ and $(K_i^{(1)}, K_{i'}^{(2)})$ with $i\ne i'$ have the same distribution $\mathbb{P}^{(1)}\otimes \mathbb{P}^{(2)}$, and (ii) numerical experiments in Sections \ref{section: Simulation experiments} and \ref{section: Applications} demonstrate that the performance of our proposed algorithms is numerically invariant to shuffling the index $i$ within each collection $\{K_i^{(j)}\}_{i=1}^n$. Without loss of generality, we assume that all the shapes have been aligned using the ``ECT alignment" (Appendix \ref{section: ECT Alignment}). Here, we present the 
theoretical foundation for employing $\{(K_i^{(1)},\, K_i^{(2)})\}_{i=1}^n$ to test the hypotheses in Eq.~\eqref{eq: the main hypotheses}. This foundation helps address the motivating question from Section \ref{section: A Motivation Question}. 

Without loss of generality, we assume $(K_i^{(1)}, K_i^{(2)})\notin \mathcal{N}$, for all $i=1,2,\ldots,n$, where $\mathcal{N}$ is the null set in Theorem \ref{thm: KL expansions of SECT} satisfying $\mathbb{P}^{(1)}\otimes \mathbb{P}^{(2)}(\mathcal{N})=0$. Then, we have
\begin{align}\label{eq: def of the xi statistic}
    \begin{aligned}
        \xi_{l,i} \overset{\operatorname{def}}{=} \delta_l \left(K_i^{(1)}, \, K_i^{(2)}\right) = \theta_l + \left( \frac{Z_{l}^{(1)}(K_i^{(1)})-Z_{l}^{(2)}(K_i^{(2)})}{\sqrt{2}} \right),
    \end{aligned}
\end{align}
where $\delta_l$ and $\theta_l$ are defined in Eq.~\eqref{eq: KL expansions of SECT}. Theorem \ref{thm: KL expansions of SECT} implies that, for each fixed $l$, the random variables $\{\xi_{l,i}\}_{i=1}^n$ are i.i.d. across $i=1, \ldots,n$ with mean $\theta_l$ and variance 1; for each fixed $i$, the random variables $\{\xi_{l,i}\}_{l=1}^\infty$ are mutually uncorrelated across $l=1,2,3,\ldots$. The following lemma represents the null $H_0$ in Eq.~\eqref{eq: the main hypotheses} using the means $\{\theta_l\}_{l=1}^\infty$.
\begin{lemma}\label{lemma: representing H0}
The null $H_0$ in Eq.~\eqref{eq: the main hypotheses} is equivalent to $\theta_l=0$ for all positive integers $l$.
\end{lemma}

Recall that $\lim_{l\rightarrow\infty}\lambda_l=0$. When eigenvalues $\lambda_l$ in the denominator of Eq.~\eqref{eq: KL expansions of SECT} are close to zero for large $l$, the estimated $\theta_l$ becomes unstable. Specifically, even if $ m_{\nu^*}^{(1)}(t) \approx m_{\nu^*}^{(2)}(t)$, an extremely small $\lambda_l$ can move the corresponding estimated $\theta_l$ far away from zero. Using the standard approach in principal component analysis, we focus on $\{\theta_l\}_{l=1}^L$ with 
\begin{align}\label{eq: def of L}
    L\overset{\operatorname{def}}{=}\max\{1,\, \Tilde{L}\},\ \ \ \text{ where } \Tilde{L} \overset{\operatorname{def}}{=} \min \left\{ l\in\mathbb{N}\, \bigg\vert\, \frac{\sum_{l'=1}^l \lambda_{l'}}{\sum_{l^{''}=1}^\infty \lambda_{l^{''}}} >0.95\right\}.
\end{align}
Hence, to test the hypotheses in Eq.~\eqref{eq: the main hypotheses} via Lemma \ref{lemma: representing H0}, we test the following
\begin{align}\label{eq: approximate hypotheses}
    \begin{aligned}
        & \widehat{H}_0: \theta_1=\cdots=\theta_L=0, \ \ \ vs. \ \ \ \widehat{H}_1: \mbox{ there exists } l'\in\{1,\cdots,L\} \mbox{ such that }\theta_{l'}\ne0.
    \end{aligned}
\end{align}
Under the null $\widehat{H}_0$ in Eq.~\eqref{eq: approximate hypotheses}, for each $l\in\{1,\cdots,L\}$, the central limit theorem indicates that $\frac{1}{\sqrt{n}}\sum_{i=1}^n \xi_{l,i}$ is asymptotically $N(0,1)$ when $n$ is large. The mutual uncorrelation in Theorem \ref{thm: KL expansions of SECT} and the asymptotic normality of $\frac{1}{\sqrt{n}}\sum_{i=1}^n \xi_{l,i}$ provide the asymptotic independence of $\{\frac{1}{\sqrt{n}}\sum_{i=1}^n \xi_{l,i} \}_{l=1}^L$ across $l=1, \ldots,L$. Then, $\sum_{l=1}^L (\frac{1}{\sqrt{n}}\sum_{i=1}^n \xi_{l,i} )^2$ is asymptotically $\chi_L^2$ under the $\widehat{H}_0$ in Eq.~\eqref{eq: approximate hypotheses}. At the asymptotic significance $\alpha\in(0,1)$, we reject the $\widehat{H}_0$ if
\begin{align}\label{eq: rejection region}
    \sum_{l=1}^L \left(\frac{1}{\sqrt{n}}\sum_{i=1}^n \xi_{l,i}\right)^2 > \chi^2_{L, 1-\alpha} = \mbox{ the $1-\alpha$ lower quantile of the $\chi^2_L$ distribution}.
\end{align}

In applications, neither the mean $m_{\nu}^{(j)}(t)$ nor the covariance $\Xi_{\nu}(s,t)$ is known. Hence, the KL expansions in Eq.~\eqref{eq: KL expansions of SECT} cannot be directly used and must be estimated. In Appendix \ref{The Numerical foundation for Hypothesis Testing}, we propose a numerical foundation for conducting the asymptotic $\chi^2$-test in Eq.~\eqref{eq: rejection region} and encapsulate the numerical procedures for the test in \textbf{Algorithm \ref{algorithm: testing hypotheses on mean functions; Appendix}}. In all our analyses in Sections \ref{section: Simulation experiments} and \ref{section: Applications}, the numerical estimates $\widehat{L}$ (see Eq.~\eqref{eq: estimated L} in Appendix \ref{The Numerical foundation for Hypothesis Testing}) of the $L$ in Eq.~\eqref{eq: def of L} are smaller than 10. When the $\widehat{L}$ values are large (e.g., several hundred), one may also consider applying the adaptive Neyman test proposed by \cite{fan1996test}.

In the simulation studies presented in Tables \ref{table: epsilon vs. rejection rates} and \ref{table: comparison with the existing ANOVA methods}, our Algorithm \ref{algorithm: testing hypotheses on mean functions; Appendix} has comparable performance with more than ten existing state-of-the-art fdANOVA methods. Nonetheless, both Algorithm \ref{algorithm: testing hypotheses on mean functions; Appendix} and the existing methods exhibit type I error inflation (e.g., the rejection rate of Algorithm \ref{algorithm: testing hypotheses on mean functions; Appendix} is $0.118$, whereas the significance is $0.05$). To mitigate this inflation, we may consider applying the permutation test using one of these methods that is computationally efficient. For example, \cite{gorecki2015comparison} proposed a permutation test based on an F-type statistic (FP). Specifically, \cite{gorecki2015comparison} approximated each observed function by basis functions via information criteria, and the F-type statistic was approximated by a form conducive to efficiently computing permutation-based p-values. However, the FP also exhibits type I error inflation (see Tables \ref{table: epsilon vs. rejection rates} and \ref{table: comparison with the existing ANOVA methods}). Motivated by the FP, we apply the permutation test to the $\chi^2$-statistic defined in Eq.~\eqref{eq: rejection region} in the following way: we first apply Algorithm \ref{algorithm: testing hypotheses on mean functions; Appendix} to our original shapes $K_i^{(j)}$ and then repeatedly re-apply Algorithm \ref{algorithm: testing hypotheses on mean functions; Appendix} to the shapes with shuffled group labels $j$. The $\chi^2$-test statistic derived from the original data is then compared to that from the shuffled data. A detailed description of our permutation-based approach is presented in \textbf{Algorithm \ref{algorithm: permutation-based testing hypotheses on mean functions}} in Appendix \ref{The Numerical foundation for Hypothesis Testing}. Simulations in Section \ref{section: Simulation experiments} demonstrate that our permutation-based approach eliminates the type I error inflation encountered by Algorithm \ref{algorithm: testing hypotheses on mean functions; Appendix}. The permutation nature of Algorithm \ref{algorithm: permutation-based testing hypotheses on mean functions} is also advantageous for small sample sizes. Note, however, that the power of Algorithm \ref{algorithm: permutation-based testing hypotheses on mean functions} under the alternative is moderately weaker than that of Algorithm \ref{algorithm: testing hypotheses on mean functions; Appendix}. Lastly, the runtimes of Algorithms \ref{algorithm: testing hypotheses on mean functions; Appendix} and \ref{algorithm: permutation-based testing hypotheses on mean functions}, when applied to simulations, are studied in Appendix \ref{section: Runtime}. We present the runtimes when applying the algorithms to real data in Table \ref{tab: Mandibular Molars Database}.

\section{Experiments Using Simulations}\label{section: Simulation experiments}

We present simulations showing the performance of our Algorithms \ref{algorithm: testing hypotheses on mean functions; Appendix} and \ref{algorithm: permutation-based testing hypotheses on mean functions}. In addition, we compare our algorithms with the ``randomization-style null hypothesis significance test (NHST)" \citep{robinson2017hypothesis}, the TRP using Wald-type permutation statistic \citep[TRP-WTPS,][]{cuesta2010simple, pauly2015asymptotic}, and the FP. Details of the randomization-style NHST are given in Appendix \ref{Randomization-style Null Hypothesis Significance Test} and referred to as Algorithm \ref{algorithm: randomization-style NHST}. The application of the FP and TRP to the SECT is described in Section \ref{section: hypothesis testing}. We implement the FP and TRP-WTPS using the \texttt{R} package \texttt{fdANOVA} with its default parameters as recommended by \cite{gorecki2019fdanova}. Additional simulations comparing our proposed algorithms and other existing fdANOVA methods are presented in Appendix \ref{appendix: Numerical Experiments on One-way ANOVA --- Existing Methods vs. Our Proposed Methods}.

We focus on a family of distributions $\{\mathbb{P}^{(\varepsilon)}\}_{0\le\varepsilon\le0.1}$ with shapes $\{K_i^{(\varepsilon)}\}_{i=1}^n \overset{\operatorname{i.i.d.}}{\sim} \mathbb{P}^{(\varepsilon)}$ via
\begin{align}\label{eq: explicit P varepsilon}
\begin{aligned}
        K_i^{(\varepsilon)} \overset{\operatorname{def}}{=} \left\{x\in\mathbb{R}^2 \, \Bigg\vert\, \inf_{y\in S_i^{(\varepsilon)}}\Vert x-y\Vert\le 0.2\right\}, \ \ \text{ where}
\end{aligned}
\end{align}
$S_i^{(\varepsilon)} \overset{\operatorname{def}}{=} \left\{\left(\frac{2}{5}+a_{1,i}\cdot\cos t,\ b_{1,i}\cdot\sin t\right) \Big\vert \frac{1-\varepsilon}{5}\pi\le t\le\frac{9+\varepsilon}{5}\pi\right\}
         \bigcup\left\{\left(-\frac{2}{5}+a_{2,i}\cdot\cos t,\ b_{2,i}\cdot\sin t\right) \Big\vert \frac{6\pi}{5}\le t\le\frac{14\pi}{5}\right\}$ and $\{a_{1,i}, a_{2,i}, b_{1,i}, b_{2,i}\}_{i=1}^n \overset{\operatorname{i.i.d.}}{\sim} N(1, 0.05^2)$. The $\varepsilon$ denotes the dissimilarity between $\mathbb{P}^{(\varepsilon)}$ and $\mathbb{P}^{(0)}$. For each $\varepsilon\in[0,\,0.1]$, through the discussion in Section \ref{section: hypothesis testing}, we test the following hypotheses via fdANOVA methods (i.e., FP, TRP-WTPS, Algorithms \ref{algorithm: testing hypotheses on mean functions; Appendix}, and \ref{algorithm: permutation-based testing hypotheses on mean functions})
\begin{align*}
    \begin{aligned}
        & H_0: m_\nu^{(0)}(t)=m_\nu^{(\varepsilon)}(t)\mbox{ for all }(\nu,t)\in\mathbb{S}^{d-1}\times[0,T]\ \ \ vs. \ \ \ H_1:  m_\nu^{(0)}(t)\ne m_\nu^{(\varepsilon)}(t)\mbox{ for some }(\nu,t),
    \end{aligned}
\end{align*}
where the mean $m_\nu^{(\varepsilon)}(t) \overset{\operatorname{def}}{=} \int_{\mathcal{S}_{R,d}^M} \operatorname{SECT}(K)(\nu, t) \,\mathbb{P}^{(\varepsilon)}(dK)$, and the null hypothesis $H_0$ is true when $\varepsilon=0$. We also test $H^*_0:\mathbb{P}^{(0)}=\mathbb{P}^{(\epsilon)}\ vs.\ \mathbb{P}^{(0)} \ne \mathbb{P}^{(\epsilon)}$ using Algorithm \ref{algorithm: randomization-style NHST}.

We set $T=3$, directions $\nu_p=(\cos\frac{p-1}{4}\pi, \sin\frac{p-1}{4}\pi)^\T$ for $p\in\{1,2,3,4\}$, levels $t_q=\frac{T}{50}q$ for $q\in\{1,\cdots,50\}$ (i.e., $\Gamma=4$ and $\Delta=50$ in Algorithms \ref{algorithm: testing hypotheses on mean functions; Appendix}, \ref{algorithm: permutation-based testing hypotheses on mean functions}, and \ref{algorithm: randomization-style NHST}), the confidence level $95\%$ (i.e., $\alpha=0.05$), and the number of permutations $\Pi=1000$. For each $\varepsilon\in$ \{0, 0.01, 0.02, 0.03, 0.04, 0.05, 0.06, 0.08, 0.1\}, we independently generate two collections $\{K_i^{(0)}\}_{i=1}^n\overset{\operatorname{i.i.d.}}{\sim} \mathbb{P}^{(0)}$ and $\{K_i^{(\varepsilon)}\}_{i=1}^n\overset{\operatorname{i.i.d.}}{\sim} \mathbb{P}^{(\varepsilon)}$ through Eq~\eqref{eq: explicit P varepsilon} with the number of shape pairs set to $n=100$, and we compute the SECT of each generated shape in directions $\{\nu_p\}_{p=1}^4$ and at levels $\{t_q\}_{q=1}^{50}$. We then implement the fdANOVA methods and Algorithm \ref{algorithm: randomization-style NHST} to these computed SECT statistics and get the corresponding \texttt{Accept}/\texttt{Reject} outputs. We repeat this procedure 1000 times and report the rejection rates across all 1000 replicates for each $\varepsilon$ in Table \ref{table: epsilon vs. rejection rates}. The rejection rates are also visually presented in Figure \ref{fig: simulation visualizations}. We choose $\Gamma=4$ as the number of directions in our simulations based on the following observation: in Appendix \ref{section: Runtime}, we experiment with all combinations of $\Gamma\in\{2,4,8\}$, $\Delta\in\{25, 50, 100\}$, and $n\in\{25, 50, 100\}$. When $\Delta=50$ and $n=100$, the number $\Gamma=4$ is sufficiently large for our Algorithms \ref{algorithm: testing hypotheses on mean functions; Appendix} and \ref{algorithm: permutation-based testing hypotheses on mean functions} to distinguish $\mathbb{P}^{(0)}$ from $\mathbb{P}^{(\varepsilon)}$ with $\varepsilon>0$ using the significance level $\alpha=0.05$. Moreover, this choice allows us to demonstrate that even a relatively small number of directions (e.g., $\Gamma=4$) is sufficient for implementing our Algorithm \ref{algorithm: testing hypotheses on mean functions; Appendix} and \ref{algorithm: permutation-based testing hypotheses on mean functions}.

\begin{table}[h]
\centering
\caption{\footnotesize{Rejection rates (from 1000 experiments) for different indices $\varepsilon$ (significance $\alpha=0.05$). Appendix \ref{appendix: Numerical Experiments on One-way ANOVA --- Existing Methods vs. Our Proposed Methods} provides a comparison of Algorithms \ref{algorithm: testing hypotheses on mean functions; Appendix}, \ref{algorithm: permutation-based testing hypotheses on mean functions}, and \ref{algorithm: randomization-style NHST} to other existing fdANOVA methods.}}
    \label{table: epsilon vs. rejection rates}
    \vspace*{0.5em}
\def\arraystretch{0.8}
\small
\begin{tabular}{|c|c|c|c|c|c|c|c|c|c|}
\hline
Indices $\varepsilon$ & 0.00 & 0.01 & 0.02 & 0.03 & 0.04 & 0.05 & 0.06 & 0.08 & 0.10 \\ \hline
Algorithm \ref{algorithm: testing hypotheses on mean functions; Appendix} & 0.118 & 0.161 & 0.315 & 0.519 & 0.785 & 0.910 &  0.975  & 0.990  & 1.000   \\
Algorithm \ref{algorithm: permutation-based testing hypotheses on mean functions} & 0.046  & 0.054 & 0.162 & 0.343  & 0.612   & 0.789  & 0.931   & 0.994 & 1.000 \\
Algorithm \ref{algorithm: randomization-style NHST} & 0.050 & 0.050 & 0.111 & 0.185 & 0.335 & 0.535 & 0.739 & 0.983 & 0.999 \\
FP         & 0.136 & 0.153 & 0.308  & 0.539  & 0.810  & 0.924 & 0.986 & 0.997 & 1.000  \\
TRP-WTPS   & 0.075 & 0.091 & 0.261  & 0.515  & 0.790  & 0.929 & 0.980 & 0.997  & 1.000  \\ \hline
\end{tabular}
\end{table}

\begin{figure}[h]
    \centering
    \includegraphics[scale=0.5]{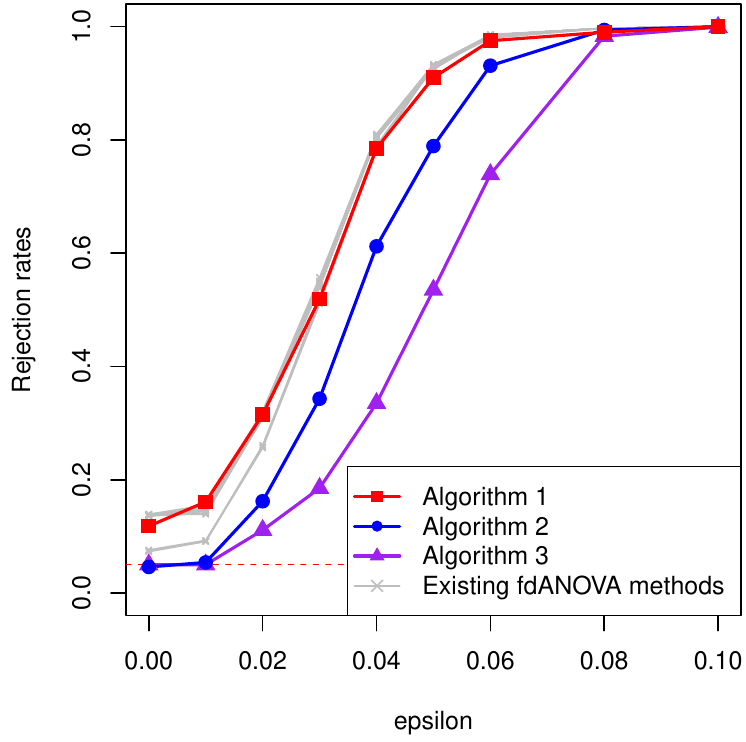}
    \includegraphics[scale=0.42]{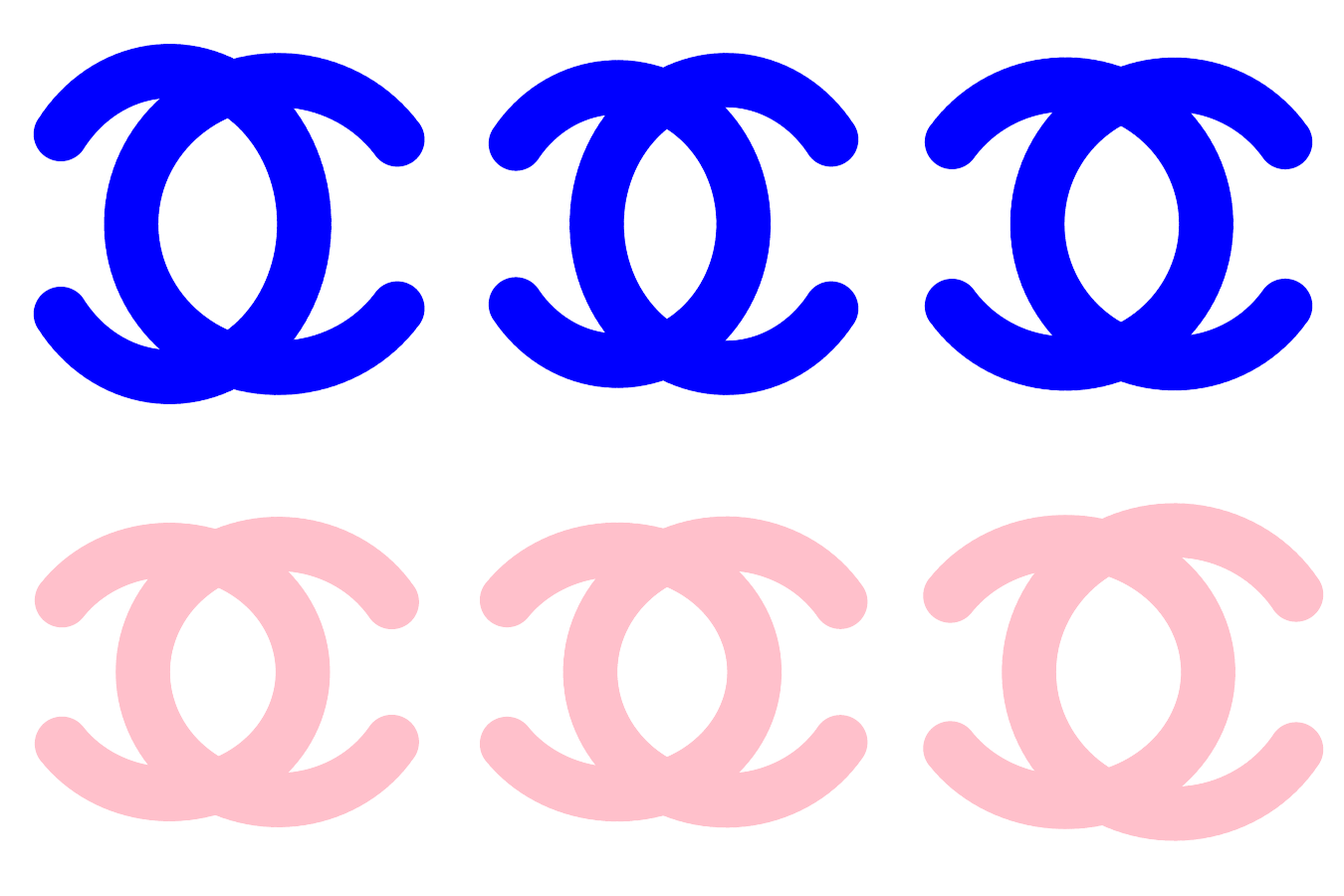}
\caption{\footnotesize{(Left panel) The relationship between $\varepsilon$ and the rejection rates computed via Algorithms \ref{algorithm: testing hypotheses on mean functions; Appendix}, \ref{algorithm: permutation-based testing hypotheses on mean functions}, \ref{algorithm: randomization-style NHST} (see Table \ref{table: epsilon vs. rejection rates}), and twelve existing fdANOVA methods (see Table \ref{table: comparison with the existing ANOVA methods} in Appendix \ref{appendix: Numerical Experiments on One-way ANOVA --- Existing Methods vs. Our Proposed Methods} for details on the existing fdANOVA methods). The (red) dashed line presents the significance level $\alpha=0.05$. (Right panel) The shapes in the first row are from $\mathbb{P}^{(0)}$, and the shapes in the second row are from $\mathbb{P}^{(0.08)}$.}}
    \label{fig: simulation visualizations}
\end{figure}

The results in Table \ref{table: epsilon vs. rejection rates} and Figure \ref{fig: simulation visualizations} demonstrate that our proposed algorithms are effective at detecting the difference between $\mathbb{P}^{(\varepsilon)}$ and $\mathbb{P}^{(0)}$ in terms of distinguishing their mean functions. Notably, our algorithms (especially Algorithm \ref{algorithm: permutation-based testing hypotheses on mean functions}) tend to avoid falsely detecting differences between shape-generating distributions under the null hypothesis (i.e., $\varepsilon=0$). As $\varepsilon$ increases, $\mathbb{P}^{(\varepsilon)}$ deviates from $\mathbb{P}^{(0)}$, and the power of our algorithms in detecting the deviation increases. When $\varepsilon\ge 0.08$, the power of Algorithms \ref{algorithm: testing hypotheses on mean functions; Appendix} and \ref{algorithm: permutation-based testing hypotheses on mean functions} exceeds 0.99. For all the $\varepsilon$, it is difficult to see the deviation of $\mathbb{P}^{(\varepsilon)}$ from $\mathbb{P}^{(0)}$ visually. For instance, by merely observing the shapes in Figure \ref{fig: simulation visualizations}, one might find it hard to differentiate between the shape collections generated by $\mathbb{P}^{(0)}$ (blue) and $\mathbb{P}^{(0.08)}$ (pink). However, in more than $99\%$ of the simulations, our algorithms detect the difference between the two distributions. We also randomly shuffle the index $i$ within each collection $\{K_i^{(\varepsilon)}\}_{i=1}^n$ and apply Algorithms \ref{algorithm: testing hypotheses on mean functions; Appendix} and \ref{algorithm: permutation-based testing hypotheses on mean functions} to the shuffled collections. The results obtained from the unshuffled and shuffled shape collections, respectively, are nearly identical. Algorithm \ref{algorithm: randomization-style NHST} performs well in detecting the discrepancy between $\mathbb{P}^{(0)}$ and $\mathbb{P}^{(\varepsilon)}$. However, its power under the alternative hypotheses (i.e., $\varepsilon>0$) is weaker than that of our Algorithms \ref{algorithm: testing hypotheses on mean functions; Appendix} and \ref{algorithm: permutation-based testing hypotheses on mean functions}. Moreover, Algorithms \ref{algorithm: testing hypotheses on mean functions; Appendix} and \ref{algorithm: permutation-based testing hypotheses on mean functions} exhibit performance comparable to twelve existing state-of-the-art fdANOVA methods (see Table \ref{table: epsilon vs. rejection rates}, Figure \ref{fig: simulation visualizations}, and Table \ref{table: comparison with the existing ANOVA methods} in Appendix \ref{appendix: Numerical Experiments on One-way ANOVA --- Existing Methods vs. Our Proposed Methods}).

\section{Applications}\label{section: Applications}

We first apply our proposed Algorithms \ref{algorithm: testing hypotheses on mean functions; Appendix} and \ref{algorithm: permutation-based testing hypotheses on mean functions} to the MPEG-7 shape silhouette database \citep{sikora2001mpeg} as a toy example. Details of this are provided in Appendix \ref{section: Silhouette Database}. This analysis shows that our proposed algorithms can distinguish between shape classes in the silhouette database and do not falsely identify signals when there are no differences between groups.

In this section, we apply our algorithms to address the motivating question in Section \ref{section: A Motivation Question}. Specifically, we utilize Algorithms \ref{algorithm: testing hypotheses on mean functions; Appendix} and \ref{algorithm: permutation-based testing hypotheses on mean functions} to distinguish between the four categories of mandibular molars in Figure \ref{fig: Teeth} that are from four genera of primates. The shapes in Figure \ref{fig: Teeth} come from two suborders of primates: Haplorhini
and Strepsirrhini (see Figure \ref{fig: Teeth}). In the haplorhine suborder collection, 29 molars came from the genus \textit{Tarsius} (yellow panels in Figure \ref{fig: Teeth}), and 9 molars came from the genus \textit{Saimiri} (grey panels in Figure \ref{fig: Teeth}). In the strepsirrhine collection, 11 molars came from the genus \textit{Microcebus} (blue panels in Figure \ref{fig: Teeth}), and 6 molars came from the genus \textit{Mirza} (green panels in Figure \ref{fig: Teeth}).

Before applying Algorithms \ref{algorithm: testing hypotheses on mean functions; Appendix} and \ref{algorithm: permutation-based testing hypotheses on mean functions}, we preprocess the raw triangle mesh data of the surfaces of the molars by aligning them through the ECT alignment approach detailed in Appendix \ref{section: ECT Alignment}. The aligned molars are presented in Figure \ref{fig: Teeth}. We apply our Algorithms \ref{algorithm: testing hypotheses on mean functions; Appendix} and \ref{algorithm: permutation-based testing hypotheses on mean functions} to the preprocessed molars. For each aligned molar, we compute its SECT for 2918 directions; in each direction, we use 200 sublevel sets. To compare any pair of molar groups, as a proof of concept, we select the smaller size of the two groups as the sample size input $n$ in our algorithms. For example, when comparing the \textit{Tarsius} and \textit{Microcebus} groups, we choose $n=11$; that is, we compare the first 11 molars of the \textit{Tarsius} group to all the molars in the \textit{Microcebus} group. We apply our algorithms to the four groups of molars and present the results in Table \ref{tab: Mandibular Molars Database}. The p-values in Table \ref{tab: Mandibular Molars Database} are either  $\chi^2$-test p-values (Algorithm \ref{algorithm: testing hypotheses on mean functions; Appendix}) or permutation-test p-values (Algorithm \ref{algorithm: permutation-based testing hypotheses on mean functions} with $1000$ permutations). The small p-values ($P<0.05$) in Table \ref{tab: Mandibular Molars Database} show that our proposed algorithms can distinguish the four different genera of primates. Since the genera \textit{Microcebus} and \textit{Mirza} belong to the same suborder Strepsirrhini (see Figure \ref{fig: Teeth}), the p-value from Algorithm \ref{algorithm: permutation-based testing hypotheses on mean functions} is comparatively large when comparing molars from these two groups. In comparison, although the \textit{Tarsius} and \textit{Saimiri} both belong to the suborder Haplorhini, the molars of the two genera look different. Specifically, the paraconids (i.e., the cusp highlighted in red in Figure \ref{fig: Teeth}) are only retained by the genus \textit{Tarsius} and, thus, are a key reason for the small p-values ($P<10^{-3}$) when comparing with molars from the \textit{Saimiri}. Other small p-values ($P<10^{-3}$) in our analyses are a result of the corresponding genera belonging to different suborders.

\begin{table}[h]
\caption{\footnotesize{P-values of Algorithms \ref{algorithm: testing hypotheses on mean functions; Appendix}, \ref{algorithm: permutation-based testing hypotheses on mean functions}, and \ref{algorithm: Permutation test using landmark-based Procrustes distances} for the data set of mandibular molars. In the last row, we present the overall runtime for conducting all hypothesis testing tasks using each algorithm.}}
\label{tab: Mandibular Molars Database}
\vspace*{0.5em}
\centering
\small
\def\arraystretch{0.8}
\begin{tabular}{lllllll}
\hline
& & Algorithm \ref{algorithm: testing hypotheses on mean functions; Appendix} & &  Algorithm \ref{algorithm: permutation-based testing hypotheses on mean functions} & & Algorithm \ref{algorithm: Permutation test using landmark-based Procrustes distances} \\[2pt] \hline
\textit{Tarsius} vs. \textit{Microcebus}            &           & $<10^{-3}$ & &  $<10^{-3}$ & & $<10^{-3}$ \\[2pt]
\textit{Tarsius} vs. \textit{Mirza}          &          & $<10^{-3}$  & &  $<10^{-3}$ & & $0.001$              \\ [2pt]
\textit{Tarsius} vs. \textit{Saimiri} & & $<10^{-3}$ & & $<10^{-3}$ & & $<10^{-3}$\\[2pt]
\textit{Microcebus} vs. \textit{Mirza} & &     $<10^{-3}$  & &  $0.009$ & & $0.004$\\[2pt]
\textit{Microcebus} vs. \textit{Saimiri} & &     $<10^{-3}$ & & $<10^{-3}$  & &  $<10^{-3}$ \\[2pt]
\textit{Mirza} vs. \textit{Saimiri} & & $<10^{-3}$ & & $<10^{-3}$ & & $<10^{-3}$\\[2pt]
\hline
\textit{Tarsius} vs. \textit{Tarsius} & & 0.196 (0.220) & & 0.496 (0.294) & & 0.527 (0.273)\\[2pt]
\hline
Overall runtimes (in hours) & & $\approx 3$ & & $\approx 3$ & & $\approx 20$\\[2pt]
\hline
\end{tabular}
\end{table}


In addition to testing the difference between genera, we apply our algorithms within the genus \textit{Tarsius}. Specifically, we focus on the first 28 molars in the \textit{Tarsius} group. We randomly split the 28 molars into two halves and apply Algorithms \ref{algorithm: testing hypotheses on mean functions; Appendix} and \ref{algorithm: permutation-based testing hypotheses on mean functions} to test the difference between the two halves. We repeat the random splitting procedure 100 times and present the corresponding p-values in Table \ref{tab: Mandibular Molars Database}. The results are summarized by their mean and standard deviation (in parentheses). These p-values show that our proposed Algorithm \ref{algorithm: permutation-based testing hypotheses on mean functions} tends to avoid the type I error for the molars from the genus \textit{Tarsius}.

Landmark methods are widely used in geometric morphometrics. One state-of-the-art approach is the ``Gaussian process landmarking (GPL)'' algorithm \citep{gao2019gaussianmorphometrics,gao2019gaussian} which can automatically sample landmarks on the surfaces of the molars in Figure \ref{fig: Teeth}. \cite{gao2019gaussianmorphometrics} showed that these sampled landmarks could induce a continuous Procrustes distance to measure the dissimilarity between molars. A permutation test can be derived using the Procrustes distance induced by the GPL algorithm. This test is detailed in Appendix \ref{section: Landmark-based Permutation Test} and is encapsulated by Algorithm \ref{algorithm: Permutation test using landmark-based Procrustes distances}. We use the GPL-based Algorithm \ref{algorithm: Permutation test using landmark-based Procrustes distances} to differentiate the four collections of molars. For this, we utilize the \texttt{MATLAB} code from the \texttt{GitHub} repository provided by \cite{gao2019gaussianmorphometrics} to compute the Procrustes distance. Performance of Algorithm \ref{algorithm: Permutation test using landmark-based Procrustes distances} is in Table \ref{tab: Mandibular Molars Database}, which shows that the GPL-based method and our Algorithm \ref{algorithm: permutation-based testing hypotheses on mean functions} have comparable performance. However, repeatedly computing the Procrustes distance is time-consuming. Hence, Algorithm \ref{algorithm: permutation-based testing hypotheses on mean functions} is more computationally efficient than Algorithm \ref{algorithm: Permutation test using landmark-based Procrustes distances} while achieving similar performance (see the last row of Table \ref{tab: Mandibular Molars Database}). 

We want to note that, in addition to the GPL algorithm, many other existing methods can be applied to measure dissimilarity between molars, including parameterized surfaces \citep{kurtek2010parameterization, kurtek2011elastic} and the approaches from computational anatomy \citep{grenander1998computational}. Similarly, the parameterized curves \citep{kurtek2012statistical} can also be used to analyze the silhouette database in Appendix \ref{section: Silhouette Database}. An even more comprehensive comparison of our algorithms with the entire edifice of existing methods is left for future research.

\section{Conclusions and Discussions}\label{Conclusions and Discussions}

In this paper, we established the mathematical foundations for the randomness of shapes via the SECT. Specifically, (i) $(\mathcal{S}_{R,d}^M, \mathscr{B}(\rho), \mathbb{P})$ was constructed as the underlying probability space; (ii) the SECT was modeled as a $C(\mathbb{S}^{d-1};\mathcal{H})$-valued random variable. We further demonstrated several properties of the SECT ensuring its KL expansion, which led to a $\chi^2$-statistic for testing hypotheses on random shapes. We bridged the fdANOVA and TDA. Simulation studies corroborated our mathematical derivations and showed the performance of our hypothesis testing algorithms. Our approach was shown to be powerful in detecting the difference between two shape-generating distributions. We applied our proposed algorithms to silhouette and primate molar datasets. Importantly, our simulations when $\varepsilon=0$, together with the applications to the molars and the silhouette database, indicate that our algorithms tend to avoid falsely detecting differences between shape-generating distributions when there are none. Using the molars in Figure \ref{fig: Teeth}, we compared the performance of our algorithms to a permutation test based on a state-of-the-art landmarking algorithm \citep{gao2019gaussianmorphometrics,gao2019gaussian}, underscoring the efficiency of our algorithms. 
We enumerate potential future research areas in Appendix \ref{section: potential future research areas}, e.g., the fdANOVA methods can be utilized for functional connectivity \citep{chen2024gradient, meng2024population} via topological summaries.


\if0\blind
{
\section*{Software Availability}
The source code for implementing the simulation studies and applications is publicly available online at \url{https://github.com/JinyuWang123/TDA.git}.
} \fi

\if1\blind
{
  \bigskip
} \fi

\if0\blind
{
\section*{Acknowledgments}

We are grateful to the Editor, Associate Editor, and three Referees of the \textit{Journal of the American Statistical Association} for their thorough review of our article and the insightful suggestions that have tremendously improved its quality. We want to thank Dr. Matthew T.~Harrison from the Division of Applied Mathematics at Brown University for useful comments and suggestions. KM wants to thank Mattie Ji from the Department of Mathematics at Brown University for her insightful comments. LC would like to acknowledge the support of a David \& Lucile Packard Fellowship for Science and Engineering. Research reported in this publication was partially supported by the National Institute On Aging of the National Institutes of Health under Award Number R01AG075511. The content is solely the responsibility of the authors and does not necessarily represent the official views of the National Institutes of Health.

\section*{Disclosure Statement}

The authors report there are no competing interests to declare.
} \fi

\if1\blind
{
  \bigskip
} \fi


\newpage

\begin{appendix}


\begin{center}
    \Large{\textbf{Appendix}}
\end{center}

\tableofcontents

\newpage

\section{Mathematical Remarks}\label{section: mathematical remarks}

In this section, we provide some mathematical remarks.

\subsection{Remark on $H^1_0([0,T])$ vs. $H^1_0((0,T))$}\label{section: notation for closed vs. open}

Strictly speaking, the functions in Sobolev space $\mathcal{H}$ are defined on the open interval $(0,T)$ instead of the closed interval $[0,T]$ \citep[][Chapter 8.2]{brezis2011functional}. Hence, the rigorous notation of $\mathcal{H}$ should be $H_0^1((0,T))$. However, Theorem 8.8 of \cite{brezis2011functional} indicates that each function in $H_0^1((0,T))$ can be uniquely represented by a continuous function defined on the closed interval $[0,T]$, which implies that functions in $H_0^1((0,T))$ can be viewed as being defined on the closed interval $[0,T]$. Therefore, to implement the boundary values on $\partial (0,T)=\{0,T\}$, we use the notation $H_0^1([0,T])$ throughout this paper to indicate that all functions in $\mathcal{H}$ are viewed as defined on $[0,T]$. The same reasoning is also applied for the space $W_0^{1,p}([0,T])$ implemented in Appendix \ref{section: Further Theorems} \citep[also see][Theorem 8.8 and Remark 8 after Proposition 8.3]{brezis2011functional}. 

Notably, $\mathcal{H}=H^1_0([0,T])$ is the RKHS generated by the kernel $\kappa(s,t)=\min\{s,t\}-\frac{st}{T}$ \citep[][Example 4.9]{lifshits2012lectures}.

Sobolev spaces have traditionally been used to study partial differential equations \citep{lu1996embedding, hadac2009well, li2016global, lv2019well, wei2021transition}. In statistics, the Sobolev spaces most commonly used are those that are also RKHSs \citep{duchon1977splines, wahba1990spline}. \cite{hairer2009introduction} offers a detailed exploration of the relationship between Sobolev spaces and RKHSs through Gaussian measures, noting that RKHSs can be broadly understood as equivalent to compact Sobolev embeddings.

\subsection{Homeomorphism Critical Points}\label{section: Homeomorphism Critical Points}

To introduce the concept of \textit{homeomorphism critical points} (HCPs), we need the following lemma
\begin{lemma}\label{lemma: Lemma 3.4 of Curry}
    For any $K\in\mathcal{S}_{R,d}^M$ and any direction $\nu\in\mathbb{S}^{d-1}$, the homeomorphism type of $K_t^\nu$  can only change finitely many times as a function of $t$.
\end{lemma}
\noindent Lemma \ref{lemma: Lemma 3.4 of Curry} is a direct consequence of Lemma 3.4 from \cite{curry2022many}; hence, its proof is omitted.

For clarity of the presentation in the rest of the Appendix, we define the HCPs as follows: \textit{For any $K\in\mathcal{S}_{R,d}^M$ and any direction $\nu\in\mathbb{S}^{d-1}$, each value of $t$ where the homeomorphism type of $K_t^\nu$ changes is called an HCP of $K$ in direction $\nu$.} 

Lemma \ref{lemma: Lemma 3.4 of Curry} implies that there are only finitely many HCPs of $K$ in each direction. Furthermore, due to the homeomorphic invariance of the Euler characteristic and Betti numbers, Lemma \ref{lemma: Lemma 3.4 of Curry} implies that the functions $t \mapsto \beta_k(K_t^\nu)$ and $t \mapsto \chi(K_t^\nu)$ are piecewise constant functions. The discontinuities of these functions are HCPs in direction $\nu$.

\subsection{O-minimal Structures}\label{section: O-minimal Structures}

The definition of o-minimal structures is available in \cite{van1998tame} and rephrased as follows.

\begin{definition}[\cite{van1998tame}]\label{def: definability}
An o-minimal structure is a sequence $\mathcal{S}=\{\mathcal{S}_n\}_{n\ge1}$ satisfying the following: 
\begin{enumerate}
    \item for each $n$, $\mathcal{S}_n$ is a Boolean algebra of subsets of $\mathbb{R}^n$;

    \item $A\in\mathcal{S}_n$ implies $A\times \mathbb{R}\in\mathcal{S}_{n+1}$ and $\mathbb{R}\times A\in\mathcal{S}_{n+1}$;

    \item $\{(x_1,\ldots,x_n)\in\mathbb{R}^n \vert\, x_i=x_j\}\in\mathcal{S}_n$ for all $1\le i < j\le n$;

    \item $A\in\mathcal{S}_{n+1}$ implies $\pi(A)\in\mathcal{S}_{n}$, where $\pi:\mathbb{R}^{n+1}\rightarrow\mathbb{R}^n$ is the usual projection map;

    \item $\{r\}\in\mathcal{S}_1$ for all $r\in\mathbb{R}$, and $\{(x,y)\in\mathbb{R}^2 \vert\, x<y\}\in\mathcal{S}_2$; and

    \item the only sets in $\mathcal{S}_1$ are the finite unions of open intervals (with $\pm\infty$ endpoints allowed) and points. 
\end{enumerate}
\end{definition}

\section{Overview of Persistence Diagrams}\label{The Relationship between PHT and SECT}

In Appendix \ref{The Relationship between PHT and SECT}, we provide an overview of PDs in the literature. The overview is provided for the following purposes: 
\begin{itemize}
    \item We provide the details of the definition of $\mathcal{S}_{R,d}^M$, particularly Condition \ref{condition: the condition for defining S_{R,d}^M}.
    \item The PD framework is the necessary tool for several proofs in Appendix \ref{section: appendix, proofs}.
\end{itemize}
Most of the materials in the overview come from or are modified from \cite{mileyko2011probability} and \cite{turner2013means}.

\subsection{Definition of Persistence Diagrams}

Let $\mathbb{K}$ be a compact topological space and $\varphi$ be a real-valued continuous function defined on $\mathbb{K}$. Because of the compactness of $\mathbb{K}$ and continuity of $\varphi$, we assume $\varphi(\mathbb{K})\subseteq[0,T]$ without loss of generality. For each $t\in [0,T]$, denote 
\begin{align*}
    \mathbb{K}^\varphi_t \overset{\operatorname{def} }{=} \{x\in\mathbb{K} \,\vert \, \varphi(x)\le t\}.
\end{align*}
Then $\mathbb{K}^\varphi_{t_1} \subseteq\mathbb{K}^\varphi_{t_2}$ for all $0\le t_1\le t_2 \le T$, and $i_{t_1 \rightarrow t_2}$ denotes the corresponding inclusion map. Hereafter, we assume the following
\begin{assumption}\label{assumption: tameness; Appendix}
    $\mathbb{K}^\varphi_{t}$ falls into finitely many homeomorphism types as $t$ ranges over $[0,T]$.
\end{assumption}
If we take $\mathbb{K}=K\in\mathcal{S}_{R,d}^M$ and
\begin{align}\label{Eq: Morse function 1}
    \varphi(x)=x\cdot \nu+R \overset{\operatorname{def}}{=} \phi_\nu(x),\ \ \ x\in K,\ \ \nu\in\mathbb{S}^{d-1},
\end{align}
we have the scenario discussed in Section \ref{The Definition of Smooth Euler Characteristic Transform}. Lemma \ref{lemma: Lemma 3.4 of Curry} implies that Assumption \ref{assumption: tameness; Appendix} is satisfied by $\phi_\nu$ for every $\nu\in\mathbb{S}^{d-1}$.

The inclusion maps $i_{t_1 \rightarrow t_2}: \mathbb{K}_{t_1}^\varphi \rightarrow \mathbb{K}_{t_2}^\varphi$ induces the group homomorphisms
\begin{align*}
    i^{\#}_{t_1 \rightarrow t_2}: H_k(\mathbb{K}_{t_1}^\varphi) \rightarrow H_k(\mathbb{K}_{t_2}^\varphi), \ \ \mbox{ for all }k\in\mathbb{Z},
\end{align*}
where $H_k(\cdot)=H_k(\,\cdot\,;\mathbb{Z}_2)$ denotes the $k$-th homology group with respect to field $\mathbb{Z}_2$, and $\mathbb{Z}_2$ is omitted for succinctness. Under Assumption \ref{assumption: tameness; Appendix}, for any $t_1 \le t_2$, we have that the image
\begin{align*}
    \operatorname{im} \left(i^{\#}_{(t_1-\delta) \rightarrow t_2} \right)=\operatorname{im} \left(i^{\#}_{(t_1-\delta) \rightarrow t_1} \circ i^{\#}_{t_1 \rightarrow t_2} \right)
\end{align*}
does not depend on $\delta>0$ when $\delta$ is sufficiently small, and then this constant image is denoted as $\operatorname{im}(i^{\#}_{(t_1-) \rightarrow t_2})$. For any $t$, the $k$-th \textit{birth group} at $t$ is defined as the quotient group
\begin{align*}
    B_k^{t} \overset{\operatorname{def}}{=} H_k(\mathbb{K}_t^\varphi)/\operatorname{im}(i^{\#}_{(t-)\rightarrow t}),
\end{align*}
and $\pi_{B_k^t}: H_k(\mathbb{K}_t^\varphi) \rightarrow B_k^{t}$ denotes the corresponding quotient map. For any $\alpha\in H_k(\mathbb{K}_t^\varphi)$, we say $\alpha$ is born at $t$ if $\pi_{B_k^t}(\alpha)\ne 0$ in $B_k^t$. Assumption \ref{assumption: tameness; Appendix} implies that $B_k^t$ is a nontrivial group only for finitely many $t$. For any $t_1<t_2$, we denote the quotient group
\begin{align*}
    E_k^{t_1, t_2} \overset{\operatorname{def}}{=} H_k(\mathbb{K}_{t_2}^\varphi)/\operatorname{im}(i^{\#}_{(t_1-)\rightarrow t_2})
\end{align*}
and the corresponding quotient map $\pi_{E_k^{t_1, t_2}}: H_k(\mathbb{K}_{t_2}^\varphi) \rightarrow E_k^{t_1, t_2}$. Furthermore, we define the following map 
\begin{align*}
    g_k^{t_1, t_2}:\ \  B_k^{t_1} \rightarrow E_k^{t_1, t_2},\ \ \ \ \pi_{B_k^{t_1}}(\alpha) \mapsto \pi_{E_k^{t_1, t_2}}\left(i^{\#}_{t_1 \rightarrow t_2}(\alpha)\right),
\end{align*}
for all $\alpha \in H_k(\mathbb{K}_{t_1}^\varphi)$. Then, we define the \textit{death group} 
\begin{align*}
    D_k^{t_1, t_2} \overset{\operatorname{def}}{=} \operatorname{ker}(g_k^{t_1, t_2}).
\end{align*}
We say a homology class $\alpha\in H_k(\mathbb{K}_{t_1}^\varphi)$ is born at $t_1$ and dies at $t_2$ if
\begin{enumerate}
    \item $\pi_{B_k^{t_1}}(\alpha)\ne 0$,

    \item $\pi_{B_k^{t_1}}(\alpha)\in D_{k}^{t_1, t_2}$, and

    \item $\pi_{B_k^{t_1}}(\alpha)\notin D_{k}^{t_1, t_2-\delta}$ for any $\delta\in(0, t_2-t_1)$.
\end{enumerate}
If $\alpha$ does not die, we artificially say that it dies at $T$ as $K_T^\varphi=\mathbb{K}$. Then we denote $\operatorname{birth}(\alpha)=t_1$ and $\operatorname{death}(\alpha)=t_2$, and the persistence of $\alpha$ is defined as 
\begin{align}\label{eq: definition of pers}
    \operatorname{pers}(\alpha) \overset{\operatorname{def}}{=} \operatorname{death}(\alpha) - \operatorname{birth}(\alpha).
\end{align}

With the notions of $\operatorname{death}(\alpha)$ and $\operatorname{birth}(\alpha)$, the $k$-th PD of $\mathbb{K}$ with respect to $\varphi$ is defined as the following multiset of 2-dimensional points \citep[][Definition 2]{mileyko2011probability}. 
\begin{align}\label{Eq: def of PD}
    \begin{aligned}
        & \operatorname{Dgm}_k(\mathbb{K};\varphi) \\
        & \overset{\operatorname{def}}{=} \bigg\{\big( \operatorname{birth}(\alpha), \operatorname{death}(\alpha)\big) \,\bigg\vert\, \alpha\in H_k(\mathbb{K}_t) \mbox{ for some }t\in[0,T] \mbox{ with }\operatorname{pers}(\alpha)>0 \bigg\}\bigcup \mathfrak{D},
    \end{aligned}
\end{align}
where $\big( \operatorname{birth}(\alpha_1), \operatorname{death}(\alpha_1) \big)$ and $\big( \operatorname{birth}(\alpha_2), \operatorname{death}(\alpha_2) \big)$ for $\alpha_1\ne\alpha_2$ are counted as two points even if $\alpha_1$ and $\alpha_2$ are born and die at the same times, respectively; that is, the multiplicity of the point $\big( \operatorname{birth}(\alpha_1), \operatorname{death}(\alpha_1) \big) = \big( \operatorname{birth}(\alpha_2), \operatorname{death}(\alpha_2) \big)$ is at least $2$; $\mathfrak{D}$ denotes the diagonal $\{(t,t)\,|\, t\in\mathbb{R}\}$ with the multiplicity of each point on this diagnal is the cardinality of $\mathbb{Z}$. Since $\operatorname{birth}(\alpha)$ is no later than $\operatorname{death}(\alpha)$, the PD $\operatorname{Dgm}_k(\mathbb{K};\varphi)$ is contained in the triangular region $\{(s,t)\in\mathbb{R}^2 \,\vert\, 0\le s\le t\le T\}$.

\subsection{Eq.~\eqref{Eq: topological invariants boundedness condition} in the definition of $\mathcal{S}_{R,d}^M$}

The following ingredients give the details of  Condition \ref{condition: the condition for defining S_{R,d}^M}.
\begin{itemize}
    \item The function $\phi_\nu$ defined in Eq.~\eqref{Eq: Morse function 1},
    \item the corresponding PDs defined by Eq.~\eqref{Eq: def of PD}, and 
    \item the definition of $\operatorname{pers}(\cdot)$ given in Eq.~\eqref{eq: definition of pers}.
    \item The notation $\#\{\cdot\}$ counts the multiplicity of the corresponding multiset.
\end{itemize}

\subsection{Bottleneck Stability}

Generally, a persistence diagram is a countable multiset of points in triangular region $\{(s,t)\in\mathbb{R}^2 \,\vert\, 0\le s, t\le T \mbox{ and }s\le t\}$ along with $\mathfrak{D}$ \citep[][Definition 2]{mileyko2011probability}. The collection of all persistence diagrams is denoted as $\mathscr{D}$. Obviously, all the $\operatorname{Dgm}_k(\mathbb{K};\varphi)$ defined in Eq.~\eqref{Eq: def of PD} is in $\mathscr{D}$. The following definition and stability result of the \textit{bottleneck distance} are from \cite{cohen2007stability}, and they play important roles in the proofs of Theorems \ref{lemma: The continuity lemma} and \ref{lemma: The continuity lemma; Appendix}.
\begin{definition}\label{def: bottleneck distance}
Let $\mathbb{K}$ be a compact topological space. $\varphi_1$ and $\varphi_2$ are two continuous real-valued functions on $\mathbb{K}$ such that $\mathbb{K}$ is tame with respect to both  $\varphi_1$ and $\varphi_2$. The bottleneck distance between PDs $\operatorname{Dgm}_k(\mathbb{K};\varphi_1)$ and $\operatorname{Dgm}_k(\mathbb{K};\varphi_2)$ is defined as 
\begin{align*}
    W_\infty \Big(\operatorname{Dgm}_k(\mathbb{K};\varphi_1), \, \operatorname{Dgm}_k(\mathbb{K};\varphi_2) \Big) \overset{\operatorname{def}}{=} \inf_{\gamma} \Big(\sup \left\{\Vert \xi - \gamma(\xi) \Vert_{l^\infty} \, \Big\vert \, \xi\in \operatorname{Dgm}_k(\mathbb{K};\varphi_1)\right\} \Big),
\end{align*}
where $\gamma$ ranges over bijections from $\operatorname{Dgm}_k(\mathbb{K};\varphi_1)$ to $\operatorname{Dgm}_k(\mathbb{K};\varphi_2)$, and
\begin{align}\label{eq: def of l infinity norm}
    \Vert \xi\Vert_{l^\infty} \overset{\operatorname{def}}{=} \max\{\vert \xi_1\vert , \vert \xi_2\vert\},\ \ \mbox{ for all }\xi=(\xi_1, \xi_2)^\T\in\mathbb{R}^2.
\end{align}
\end{definition}
\begin{theorem}\label{thm: bottleneck stability}
Let $\mathbb{K}$ be a compact and finitely triangulable topological space. $\varphi_1$ and $\varphi_2$ are two continuous real-valued functions on $\mathbb{K}$ such that $\mathbb{K}$ is tame with respect to both  $\varphi_1$ and $\varphi_2$. Then, we have the bottleneck stability as follows 
\begin{align*}
    W_\infty \Big(\operatorname{Dgm}_k(\mathbb{K};\varphi_1), \, \operatorname{Dgm}_k(\mathbb{K};\varphi_2) \Big) \le \sup_{x\in\mathbb{K}} \left\vert \varphi_1(x) - \varphi_2(x) \right\vert.
\end{align*}
\end{theorem}

\section{Additional Theorems}\label{section: Further Theorems}

In this section, we provide some further theorems and lemmas which may deepen our understanding of the SECT framework. The proofs of these theorems and lemmas are given in Appendix \ref{section: appendix, proofs}.
\begin{itemize}
\item The following theorem shows that $C(\mathbb{S}^{d-1};H_0^1([0,T]))$ is a separable Banach space, hence, a Polish space.
\begin{theorem}\label{thm: the separability of C(Shere;H)}
(i) Let $\mathcal{H}$ be a separable Hilbert space. Then, $C(\mathbb{S}^{d-1}; \, \mathcal{H})$ is separable.\\
(ii) Let $\mathcal{H}$ be a separable Hilbert space. Then, $C(\mathbb{S}^{d-1};\mathcal{H})$ is a Banach space.\\
(iii) $C(\mathbb{S}^{d-1};H_0^1([0,T]))$ is a separable Banach space.
\end{theorem}

\noindent\textbf{Remark:} The results in Theorem \ref{thm: the separability of C(Shere;H)} are well-known. The proof of Theorem \ref{thm: the separability of C(Shere;H)} is provided in Appendix \ref{section: appendix, proofs} for the convenience of the reader.

\item The following boundedness can be derived from Condition \ref{condition: the condition for defining S_{R,d}^M}.
\begin{theorem}\label{thm: boundedness topological invariants theorem}
If $K\in\mathcal{S}_{R,d}^M$, we have 
\begin{align*}
    \sup_{\nu\in\mathbb{S}^{d-1}}\left\{\sup_{0\le t\le T}\left\vert\chi_t^\nu(K) \right\vert\right\} \le M.
\end{align*}
\end{theorem}

\item Lemma \ref{thm: tameness property} states that $t\mapsto \chi(K_t^\nu)$, for each $\nu\in\mathbb{S}^{d-1}$, is a piecewise constant function with only finitely many discontinuities. The following result that follows directly from \cite{ji2023euler} clarifies the behaviors of $t\mapsto \chi(K_t^\nu)$ at the discontinuities.
\begin{lemma}\label{lemma: right continuity of the ECT}
    For any shape $K\in\mathcal{S}_d$ and fixed direction $\nu\in\mathbb{S}^{d-1}$, the function $t\mapsto\chi(K_t^\nu)$ is right continuous.
\end{lemma}

\item The following lemma indicates that the SECT is uniformly bounded.
\begin{lemma}\label{lemma: uniform boundedness of the SECT}
    For any $K\in\mathcal{S}_{R,d}^M$, we have 
    \begin{align*}
        \left\Vert \operatorname{SECT}(K) \right\Vert_{C(\mathbb{S}^{d-1};\,\mathcal{H})} \le 2M \cdot\sqrt{T}.
    \end{align*}
\end{lemma}

\item Lemma \ref{thm: Sobolev function paths} is a special case of the following lemma.
\begin{lemma}\label{thm: Sobolev function paths; general, appendix}
For any fixed $K\in\mathcal{S}_{R,d}^M$ and $\nu\in\mathbb{S}^{d-1}$, the function $\operatorname{SECT}(K)(\nu)$ (i.e., $t\mapsto\operatorname{SECT}(K)(\nu,t)$ with fixed $\nu$) belongs to $W^{1,p}_0([0,T]) \subseteq \mathcal{B}$ for all $p\in[1,\infty)$.
\end{lemma}

Here, $W^{1,p}_0([0,T])$ is a Sobolev space defined as 
\begin{align*}
    W^{1,p}_0([0,T]) \overset{\operatorname{def}}{=} \left\{f\in L^p([0,T]) \, \Big\vert \, \mbox{weak derivative }f'\in L^p([0,T]) \mbox{ and }f(0)=f(T)=0 \right\}
\end{align*}
\citep[see][Theorem 8.12]{brezis2011functional}. We focus on the case $p=2$ where $\mathcal{H}=H_0^1([0,T])=W^{1,2}_0([0,T])$, which implies Lemma \ref{thm: Sobolev function paths}. 

    \item The following theorem is a companion result of Theorem \ref{lemma: The continuity lemma}.
    \begin{theorem}\label{lemma: The continuity lemma; Appendix}
        For each $K\in\mathcal{S}_{R,d}^M$, we have the following:\\ 
There exists a constant $C^*_{M,R,d}$ depending only on $M$, $R$, and $d$ such that the following inequality holds for any two directions $\nu_1,\nu_2\in\mathbb{S}^{d-1}$, 
\begin{align}\label{Eq: lemma for the continuity inequality}
    \begin{aligned}
       \left( \int_0^T \Big\vert\chi_\tau^{\nu_1}(K)-\chi_\tau^{\nu_2}(K)\Big\vert^2 d\tau \right)^{1/2} \le C^*_{M,R,d} \cdot \sqrt{\Vert \nu_1-\nu_2\Vert}.
    \end{aligned}
\end{align}
Furthermore, the constant $\Tilde{C}_T$ in Eq.~\eqref{eq: Sobolev embedding from Morrey} provides the following inequality 
\begin{align}\label{eq: bivariate Holder continuity}
    \begin{aligned}
        & \Big\vert \operatorname{SECT}(K)(\nu_1; t_1)-\operatorname{SECT}(K)(\nu_2; t_2)\Big\vert\\ 
    & \le \Tilde{C}_T \left\{2M\sqrt{T}\cdot \sqrt{\vert t_1-t_2\vert} + C^*_{M,R,d} \cdot \sqrt{ \Vert \nu_1 - \nu_2\Vert + \Vert 
    \nu_1-\nu_2 \Vert^2 } \right\},
    \end{aligned}
\end{align}
for all $\nu_1, \nu_2\in\mathbb{S}^{d-1}$ and $t_1, t_2\in[0,T]$, which implies that $(\nu,t)\mapsto \operatorname{SECT}(K)(\nu, t)$, as a function on $\mathbb{S}^{d-1}\times[0,T]$, belongs to $C^{0,\frac{1}{2}}(\mathbb{S}^{d-1}\times[0,T]; \, \mathbb{R})$.
    \end{theorem}

\item For any given o-minimal structure $\mathcal{S}$, elements of $\mathcal{S}$ are called definable sets. Furthermore, compact definable sets are called constructible sets. The collection of constructible subsets of $\mathbb{R}^d$ is denoted by $\operatorname{CS}(\mathbb{R}^d)$ \citep[e.g., see][Section 2]{curry2022many}. If the o-minimal structure $\mathcal{S}$ satisfies Assumption \ref{Assumption: basic requirements for o-minimal structures of interest}, it is straightforward that our proposed $\mathcal{S}_{R,d}^M$ is a subset of $ \operatorname{CS}(\mathbb{R}^d)$.

Using the results originated from Euler calculus \citep{schapira1988cycles, viro1988some, schapira1991operations, schapira1995tomography, van1998tame}, particularly the ``Schapira's inversion formula" \citep{schapira1995tomography}, \cite{ghrist2018persistent} and \cite{curry2022many} independently proved the following injectivity of the ECT.
\begin{theorem}[Theorem 1 of \cite{ghrist2018persistent} or Theorem 3.5 of \cite{curry2022many}]\label{thm: injectivity of the ECT}
    The following map is injective for all dimensions $d$
    \begin{align*}
        \operatorname{ECT}:\ \ & \operatorname{CS}(\mathbb{R}^d) \rightarrow \mathbb{Z}^{\mathbb{S}^{d-1}\times\mathbb{R}}, \\
        & K \mapsto \{\chi(K_t^\nu)\}_{(\nu,t) \in \mathbb{S}^{d-1}\times\mathbb{R}}.
    \end{align*}
\end{theorem}
\noindent The following result from \cite{ghrist2018persistent} is a corollary of Theorem \ref{thm: injectivity of the ECT}, and it guarantees the injectivity of the SECT.
\begin{corollary}[Corollary 1 of \cite{ghrist2018persistent}]\label{corollary: Corollary 1 of Ghrist et all.(2018)}
    The smooth Euler characteristic transform of \cite{crawford2020predicting} is injective on constructible subsets of $\mathbb{R}^d$ for all dimensions $d$.
\end{corollary}

\item We introduce another topological summary --- the primitive Euler characteristic transform (PECT) --- which is related to the SECT. 

The PECT is defined as follows
\begin{align}\label{Eq: def of PECT}
\begin{aligned}
    & \operatorname{PECT}: \ \mathcal{S}_{R,d}^M \rightarrow C(\mathbb{S}^{d-1};\mathcal{H}_{BM}),\ \ \ K \mapsto \operatorname{PECT}(K) \overset{\operatorname{def}}{=} \{\operatorname{PECT}(K)(\nu)\}_{\nu\in\mathbb{S}^{d-1}}, \\
    & \mbox{where }\ \operatorname{PECT}(K)(\nu) \overset{\operatorname{def}}{=} \left\{\int_0^t \chi_\tau^\nu(K) \,d\tau\right\}_{t\in[0,T]},
    \end{aligned}
\end{align}
and $\mathcal{H}_{BM} \overset{\operatorname{def}}{=} \{f\in L^2([0,T]) \,\vert\, \mbox{weak derivative }f' \mbox{ exists, }f'\in L^2([0,T]), \mbox{ and } f(0)=0 \}$ is a separable Hilbert space equipped with the inner product 
\begin{align*}
    \langle f, g \rangle_{\mathcal{H}_{BM}}  =  \int_0^T f'(t)g'(t) dt.
\end{align*}
Eq.~\eqref{Eq: lemma for the continuity inequality}, together with Lemma \ref{lemma: weak derivative formula}, implies that
\begin{align*}
    \left\Vert \operatorname{PECT}(K)(\nu_1)-\operatorname{PECT}(K)(\nu_2)\right\Vert_{\mathcal{H}_{BM}}\le C^*_{M,R,d} \cdot \sqrt{\Vert \nu_1-\nu_2\Vert},
\end{align*}
for any $\nu_1, \nu_2\in\mathbb{S}^{d-1}$. Therefore, we have $\operatorname{PECT}(K)\in C(\mathbb{S}^{d-1}; \mathcal{H}_{BM})$ in Eq.~\eqref{Eq: def of PECT}. Section \ref{proof: measurability of SECT and PECT; appendix} shows that the following map is continuous, hence, measurable
\begin{align*}
    \operatorname{PECT}:\ \ & \mathcal{S}_{R,d}^M\rightarrow\mathbb{R}, \\
    & K\mapsto \operatorname{PECT}(K)(\nu,t).
\end{align*}
Then, for each fixed direction $\nu\in\mathbb{S}^{d-1}$, $\operatorname{PECT}(\nu) \overset{\operatorname{def}}{=}\{\operatorname{PECT}(\nu,t)\}_{t\in[0,T]}$ is a stochastic process.

The following theorem regarding the PECT indicates that the stochastic process $\operatorname{SECT}(\nu)=\{\operatorname{SECT}(\nu,t)\}_{t\in[0,T]}$ tends to be non-Gaussian for every fixed direction $\nu$.
\begin{theorem}\label{eq: the SECT is non-Gaussian}
    Let $\nu\in\mathbb{S}^{d-1}$ be a fixed direction. If the stochastic processes $\{\chi_t^\nu\,\vert\, t_0<t\le t_1\},\,\{\chi_t^\nu\,\vert\, t_1<t\le t_2\},\,\ldots,\, \{\chi_t^\nu\,\vert\, t_{l-1}<t\le t_l\}$ are independent for any partition $0=t_0<t_1<\cdots<t_l=T$ of $[0,T]$, the stochastic process $\operatorname{PECT}(\nu)=\{\operatorname{PECT}(\nu,t)\}_{t\in[0,T]}$ is a Gaussian process.
\end{theorem}
\begin{remark}\label{remark: The Gaussianity of the SECT is not guaranteed.}
    The indepedence condition in Theorem \ref{eq: the SECT is non-Gaussian} cannot hold, because the sample paths of the stochastic process $\{\chi_t^\nu\}_{t\in[0,T]}$ are piecewise constant functions (see Lemmas \ref{thm: tameness property} and \ref{lemma: right continuity of the ECT}). Hence, there is no guarantee for the Gaussianity of the PECT. Note that the PECT relates to the SECT via the following
\begin{align}\label{eq: relationship between PECT and SECT}
    \operatorname{SECT}(K)(\nu, t)=\operatorname{PECT}(K)(\nu, t)-\frac{t}{T}\operatorname{PECT}(K)(\nu, T),
\end{align}
for all $\nu\in\mathbb{S}^{d-1}$ and $t\in[0,T]$. The Gaussianity of the SECT is not guaranteed.
\end{remark}

\item The following lemma provides several properties of the mean $m_\nu$ and covariance $\Xi_\nu$ that validate our KL expansion of $\operatorname{SECT}(\nu)$ in Section \ref{section: hypothesis testing} \citep[also see][Section 7.2]{hsing2015theoretical}.
\begin{lemma}\label{thm: mean is in H}
\begin{enumerate}
    \item For each $\nu$, the function $m_\nu \overset{\operatorname{def}}{=} \{m_\nu(t)\}_{t\in[0,T]}$ of $t$ is in $\mathcal{H}$.

    \item The map $(K, t)\mapsto \operatorname{SECT}(K)(\nu, t)$ is in 
    \begin{align*}
        L^2 \Big(\mathcal{S}_{R,d}^M\times[0,T],\, \mathscr{F}\otimes \mathscr{B}([0,T]),\, \mathbb{P}(dK)\otimes dt \Big),
    \end{align*}
    where $\mathbb{P}(dK)\otimes dt$ denotes the product measure generated by $\mathbb{P}(dK)$ and the Lebesgue measure $dt$.

    \item $\operatorname{SECT}(\nu)$ is mean-square continuous, i.e.,
\begin{align*}
    \lim_{\epsilon\rightarrow0}\mathbb{E}\vert \operatorname{SECT}(\nu,t+\epsilon)-\operatorname{SECT}(\nu, t)\vert^2=0.
\end{align*}

\item $(s,t)\mapsto\Xi_\nu(s,t)$ is continuous on $[0,T]^2$.

\item The map $\nu\mapsto m_\nu$ is an element of $C(\mathbb{S}^{d-1};\mathcal{H})$. 
\end{enumerate}
\end{lemma}

\textbf{Remark:} Although the mean $\pmb{m}\overset{\operatorname{def}}{=}\{m_\nu\}_{\nu\in\mathbb{S}^{d-1}}$ of SECT belongs to $C(\mathbb{S}^{d-1};\mathcal{H})$ (see Lemma \ref{thm: mean is in H}), it may not belong to $\operatorname{SECT}(\mathcal{S}_{R,d}^M)$ (note that the SECT is not surjective). Hence, there does not necessarily exist a shape in $\mathcal{S}_{R,d}^M$ whose SECT is $\pmb{m}$. Under some topological conditions, the ``Schapira's inversion formula" \citep{schapira1995tomography} can be applied to obtain a pseudo-inverse $\operatorname{SECT}^{-1}(\pmb{m})$. However, the $\operatorname{SECT}^{-1}(\pmb{m})$ can be a ``grayscale image" (e.g., shape with blurred boundary) instead of a shape in $\mathcal{S}_{R,d}^M$. Specifically, \cite{meng2023Inference} generalized the SECT from shapes to grayscale images, enlarging the image of the SECT. The mean $\pmb{m}$ may correspond to a grayscale image under some conditions. Detailed discussions on the relationship between the SECT and grayscale images are provided in \cite{meng2023Inference}. In addition, a relevant discussion from the Fréchet mean viewpoint is given in Appendix \ref{section: Proof-of-Concept Simulation Examples II: Random Shapes}. Complete studies on the pseudo-inverse $\operatorname{SECT}^{-1}(\pmb{m})$ are left for future research.
    
\end{itemize}

\section{Proof-of-Concept Simulation Examples}

In Section \ref{Proof-of-Concept Simulation Examples I: Deterministic Shapes}, we provide proof-of-concept examples to visually illustrate the SECT of deterministic shapes and support the $\frac{1}{2}$-Hölder regularity stated in Theorems \ref{lemma: The continuity lemma} and \ref{lemma: The continuity lemma; Appendix}.  Then, in Section \ref{section: Proof-of-Concept Simulation Examples II: Random Shapes}, we provide proof-of-concept examples to illustrate random shapes and their SECT representations visually; the examples in Section \ref{section: Proof-of-Concept Simulation Examples II: Random Shapes} provide potential relationships between the Fréchet mean \citep{frechet1948elements}, Fréchet regression \citep{petersen2019frechet}, Wasserstein regression \citep{chen2021wasserstein}, and manifold learning \citep{dunson2021inferring, meng2021principal, li2022efficient}. Lastly, in Section \ref{section: Computation of SECT Using the Čech Complexes}, we provide an approach to approximately computing the SECT, which is implemented in Sections \ref{section: Simulation experiments}, \ref{Proof-of-Concept Simulation Examples I: Deterministic Shapes}, and \ref{section: Proof-of-Concept Simulation Examples II: Random Shapes} to compute the SECT values.

\subsection{Proof-of-Concept Simulation Examples I: Deterministic Shapes}\label{Proof-of-Concept Simulation Examples I: Deterministic Shapes}

In this subsection, we compute the SECT of two simulated shapes $K^{(1)}$ and $K^{(2)}$ of dimension $d=2$. These shapes are defined as the following and presented in Supplementary Figures \ref{fig: SECT visualizations, deterministic}(a) and (g), respectively:
\begin{align}\label{eq: example shapes K1 and K2}
\begin{aligned}
    & K^{(j)}  =  \left\{x\in\mathbb{R}^2 \,\Bigg\vert \, \inf_{y\in S^{(j)}}\Vert x-y\Vert\le \frac{1}{5}\right\},\ \ \mbox{ where }j\in\{1,2\}, \\
    & S^{(1)}  =  \left\{\left(\frac{2}{5}+\cos t, \sin t\right) \,\Bigg\vert \, \frac{\pi}{5}\le t\le\frac{9\pi}{5}\right\}\bigcup\left\{\left(-\frac{2}{5}+\cos t, \sin t\right) \,\Bigg\vert \, \frac{6\pi}{5}\le t\le\frac{14\pi}{5}\right\},\\
        & S^{(2)}  =  \left\{\left(\frac{2}{5}+\cos t, \sin t\right) \,\Bigg\vert \, 0\le t\le 2\pi\right\}\bigcup\left\{\left(-\frac{2}{5}+\cos t, \sin t\right) \,\Bigg\vert \, \frac{6\pi}{5}\le t\le\frac{14\pi}{5}\right\}.
\end{aligned}
\end{align}
We compute $\operatorname{SECT}(K^{(j)})(\nu, t)$ across directions $\nu\in\mathbb{S}^{1}$ and sublevel sets $t\in[0,T]$ with $T=3$. For the following visualization, we identify $\nu\in\mathbb{S}^1$ through the parametrization $\nu=(\cos\vartheta, \sin\vartheta)$ with $\vartheta\in[0,2\pi)$. 

\begin{itemize}
    \item The surfaces of the bivariate maps $(\vartheta, t)\mapsto \operatorname{SECT}(K^{(j)})(\nu, t)$, for $j\in\{1,2\}$, are presented in Supplementary Figures \ref{fig: SECT visualizations, deterministic}(b), (c), (h), and (i).

    \item The curves of the univariate maps $t\mapsto \operatorname{SECT}(K^{(j)})\left((1,0)^\T;t\right)$, for $j\in\{1,2\}$, are presented by the black solid lines in Supplementary Figures \ref{fig: SECT visualizations, deterministic}(d) and (j); while the curves of the univariate maps $t\mapsto \operatorname{SECT}(K^{(j)})\left((0,1)^\T;t\right)$, for $j\in\{1,2\}$, are presented by black solid lines in Supplementary Figures \ref{fig: SECT visualizations, deterministic}(e) and (k).

    \item Lastly, the curves of the univariate maps $\vartheta\mapsto \operatorname{SECT}(K^{(j)})\left((\cos\vartheta, \sin\vartheta)^\T;\frac{3}{2}\right)$, for $j\in\{1,2\}$, are presented by the black solid lines in Supplementary Figures \ref{fig: SECT visualizations, deterministic}(f) and (l).
\end{itemize}

\begin{figure}
    \centering
    \includegraphics[scale=0.53]{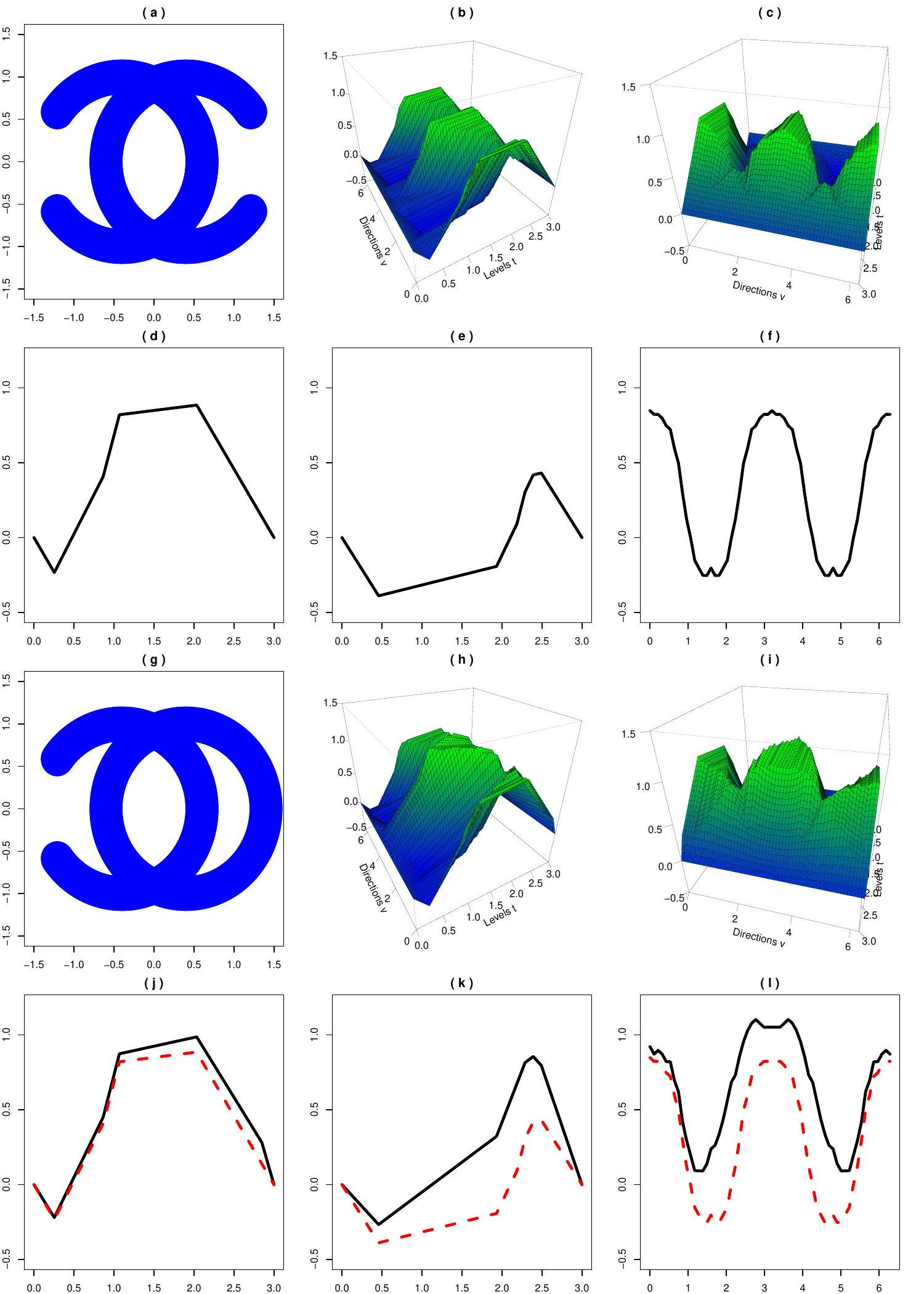}
    \caption{{\footnotesize Visualizations of $\operatorname{SECT}(K^{(j)})$ for $j\in\{1,2\}$, where $K^{(1)}$ and $K^{(2)}$ are defined in Eq.~\eqref{eq: example shapes K1 and K2}. Panels (b) and (c) present the same surface from different angles. Panels (h) and (i) present the same surface from different angles. The similarity between the curves in panel (j) partially indicates that the SECT in only one direction does not preserve all the geometric information of a shape. Panels (g) and (h) also appear in the main text under Figure \ref{fig: SECT_illustration}.}}
    \label{fig: SECT visualizations, deterministic}
\end{figure}

These figures illustrate the continuity of $(\nu,t)\mapsto \operatorname{SECT}(K^{(j)})(\nu, t)$ stated in Theorem \ref{lemma: The continuity lemma; Appendix} (Eq.~\eqref{eq: bivariate Holder continuity} therein). Specifically, the curves and surfaces in these figures look smoother than the sample path of Brownian motions, while they are not differentiable everywhere. With probability one, the sample paths of a Brownian motion are not locally $C^{0,\frac{1}{2}}$-continuous \citep[][Remark 22.4]{klenke2013probability}. Hence, based on Supplementary Figure \ref{fig: SECT visualizations, deterministic}, the regularity of $(\nu,t)\mapsto \operatorname{SECT}(K^{(j)})(\nu, t)$ is likely to be better than that of Brownian motion sample paths, but worse than continuously differentiable functions. Therefore, Supplementary Figure \ref{fig: SECT visualizations, deterministic} supports the $C^{0,\frac{1}{2}}$-continuity in Theorem \ref{lemma: The continuity lemma; Appendix}.

Theorem \ref{thm: invertibility} (i.e., the SECT as a map is injective) indicates that all information of $K^{(1)}$ and $K^{(2)}$ is stored in the surfaces presented by Supplementary Figures \ref{fig: SECT visualizations, deterministic}(b), (c), (h), and (i). The red dashed curves in Supplementary Figures \ref{fig: SECT visualizations, deterministic}(j), (k), and (l) are the counterparts of $K^{(1)}$ (see the curves in Supplementary Figures \ref{fig: SECT visualizations, deterministic}(d), (e), and (f)). The discrepancy between the solid black and dashed red curves illustrates the ability of the SECT to distinguish shapes, which motivates us to develop the hypothesis testing approach in Section \ref{section: hypothesis testing}.

\subsection{Proof-of-Concept Simulation Examples II: Random Shapes}\label{section: Proof-of-Concept Simulation Examples II: Random Shapes}

In this subsection, we compute the SECT for a collection of random shapes $\{K_i\}_{i=1}^n$ of dimension $d=2$. These shapes are randomly generated as follows
\begin{align}\label{eq: randon shapes under null}
        \begin{aligned}
K_i  = & \left\{x\in\mathbb{R}^2 \, \Bigg\vert\, \inf_{y\in S_i}\Vert x-y\Vert\le \frac{1}{5}\right\},\ \ \mbox{ where } \\
        S_i  = & \left\{\left(\frac{2}{5}+a_{1,i}\times\cos t,\ \ b_{1,i}\times\sin t\right) \, \Bigg\vert\, \frac{\pi}{5}\le t\le\frac{9\pi}{5}\right\}\\ & \bigcup\left\{\left(-\frac{2}{5}+a_{2,i}\times\cos t,\ \ b_{2,i}\times\sin t\right) \, \Bigg\vert\, \frac{6\pi}{5}\le t\le\frac{14\pi}{5}\right\},
        \end{aligned}
\end{align}
and $\{a_{1,i}, a_{2,i}, b_{1,i}, b_{2,i}\}_{i=1}^n \overset{\operatorname{i.i.d.}}{\sim} N(1, 0.05^2)$ follow a normal distribution. One element of the shape collection $\{K_i\}_{i=1}^n$ is presented in Supplementary Figure \ref{fig: SECT visualizations, random}(a). The underlying distribution on $\mathcal{S}_{R,d}^M$ generating $\{K_i\}_{i=1}^n$ is denoted by $\mathbb{P}$, and the expectation associated with $\mathbb{P}$ is denoted by $\mathbb{E}$. We estimate the expected value $\mathbb{E}\{\operatorname{SECT}(\nu, t)\}$ by the sample average $\frac{1}{n}\sum_{i=1}^n \operatorname{SECT}(K_i)(\nu, t)$ with $n=100$. We identify each direction $\nu\in\mathbb{S}^1$ through the parametrization $\nu=(\cos\vartheta, \sin\vartheta)$ with some $\vartheta\in[0,2\pi)$ as we did in Section \ref{Proof-of-Concept Simulation Examples I: Deterministic Shapes}. 

\begin{itemize}
    \item The surface of the map $(\vartheta,t)\mapsto \mathbb{E}\{\operatorname{SECT}(\nu, t)\}$ is presented in Supplementary Figures \ref{fig: SECT visualizations, random}(b) and (c).

    \item The black solid curves in Supplementary Figure \ref{fig: SECT visualizations, random}(d) present the 100 sample paths $t\mapsto \operatorname{SECT}(K_i)\left((1,0)^\T, t\right)$; the black solid curves in Supplementary Figure \ref{fig: SECT visualizations, random}(e) present sample paths $t\mapsto \operatorname{SECT}(K_i)\left((0,1)^\T, t\right)$; and the black solid curves in Supplementary Figure \ref{fig: SECT visualizations, random}(f) present paths $\vartheta\mapsto \operatorname{SECT}(K_i)\left((\cos\vartheta, \sin\vartheta)^\T, \frac{3}{2}\right)$, for $i\in\{1,\cdots,100\}$.

    \item The red solid curves in Supplementary Figures \ref{fig: SECT visualizations, random}(d), (e), and (f) present mean curves 
    \begin{align*}
        & t\mapsto \mathbb{E}\{\operatorname{SECT}\left((1,0)^\T, t\right)\}, \\
        & t\mapsto \mathbb{E}\{\operatorname{SECT}\left((0,1)^\T, t\right)\}, \text{ and }\\
        & \vartheta\mapsto \mathbb{E}\left\{\operatorname{SECT}\left((\cos\vartheta, \sin\vartheta)^\T, \frac{3}{2}\right) \right\},
    \end{align*}
    respectively.
\end{itemize} 

The smoothness of the red solid curves in Supplementary Figures \ref{fig: SECT visualizations, random}(d) and (e) supports the regularity of $\{m_\nu(t)\}_{t\in[0,T]}$ in Lemma \ref{thm: mean is in H}. The finite variance of $\operatorname{SECT}(\nu, t)$ for $\nu=(1,0)^\T, (0,1)^\T$ and $t=3/2$, visually presented in Supplementary Figures \ref{fig: SECT visualizations, random}(d), (e), and (f), supports Lemma \ref{assumption: existence of second moments}.

In addition, the blue dashed curves in Supplementary Figures \ref{fig: SECT visualizations, random}(d), (e), and (f) present curves 
\begin{align*}
    & t\mapsto \operatorname{SECT}(K^{(1)})\left((1,0)^\T,t\right), \\
    & t\mapsto \operatorname{SECT}(K^{(1)})\left((0,1)^\T,t\right), \ \ \text{ and }\\
    & \vartheta\mapsto \operatorname{SECT}(K^{(1)})\left((\cos\vartheta, \sin\vartheta)^\T,\frac{3}{2}\right),
\end{align*}
respectively, where shape $K^{(1)}$ is defined in Eq.~\eqref{eq: example shapes K1 and K2}. Since $\mathbb{E}\{a_{1,i}\}=\mathbb{E}\{a_{2,i}\}=\mathbb{E}\{b_{1,i}\}=\mathbb{E}\{b_{2,i}\}=1$, the shape $K^{(1)}$ defined in Eq.~\eqref{eq: example shapes K1 and K2} can be somewhat viewed as the ``mean shape" of the random collection $\{K_i\}_{i=1}^n$. The similarity between the red solid curves and blue dashed curves in \ref{fig: SECT visualizations, random}(d), (e), and (f) supports the ``mean shape" role of $K^{(1)}$. The rigorous definition of a ``mean shape" and its relationship to the mean function $\mathbb{E}\{\operatorname{SECT}(\nu, t)\}$ is outside the scope of this work. A potential approach for defining mean shapes is through the following Fréchet mean \citep{frechet1948elements}
\begin{align}\label{Frechet mean shape}
    K_\oplus \overset{\operatorname{def}}{=} \argmin_{K\in\mathcal{S}_{R,d}^M} \mathbb{E}\left[\left\{\rho(\,\cdot, K)\right\}^2\right] = \argmin_{K\in\mathcal{S}_{R,d}^M} \left[ \int_{\mathcal{S}_{R,d}^M} \left\{\rho(K', K)\right\}^2 \mathbb{P}(dK') \right],
\end{align}
where $\rho$ can be either the metric on $\mathcal{S}_{R,d}^M$ defined in Eq.~\eqref{Eq: distance between shapes} or any other metrics generating $\sigma$-algebras satisfying Assumption \ref{assumption: the measurability of ECC}. 

The existence and uniqueness of the minimizer $K_\oplus$ in Eq.~\eqref{Frechet mean shape}, the relationship between $\operatorname{SECT}(K_\oplus)$ and $\mathbb{E}\{\operatorname{SECT}\}$, and the extension of Eq.~\eqref{Frechet mean shape} to Fréchet and Wasserstein regression \citep{petersen2019frechet, chen2021wasserstein} for random shapes are left for future research. The study of the existence of $K_\oplus$ will be an analog of Section 4 in \cite{mileyko2011probability}.

In the scenarios where the SECT of shapes from distribution $\mathbb{P}$ are computed only in finitely many directions and at finitely many levels (see the end of Section \ref{section: distributions of Gaussian bridge}), the mean surface $(\vartheta, t) \mapsto \mathbb{E}\{\operatorname{SECT}(\nu, t)\}$ in Supplementary Figures \ref{fig: SECT visualizations, random}(b) and (c) can also be potentially estimated using manifold learning methods \citep{dunson2021inferring, meng2021principal, li2022efficient}.

\begin{figure}[h]
    \centering
    \includegraphics[scale=0.6]{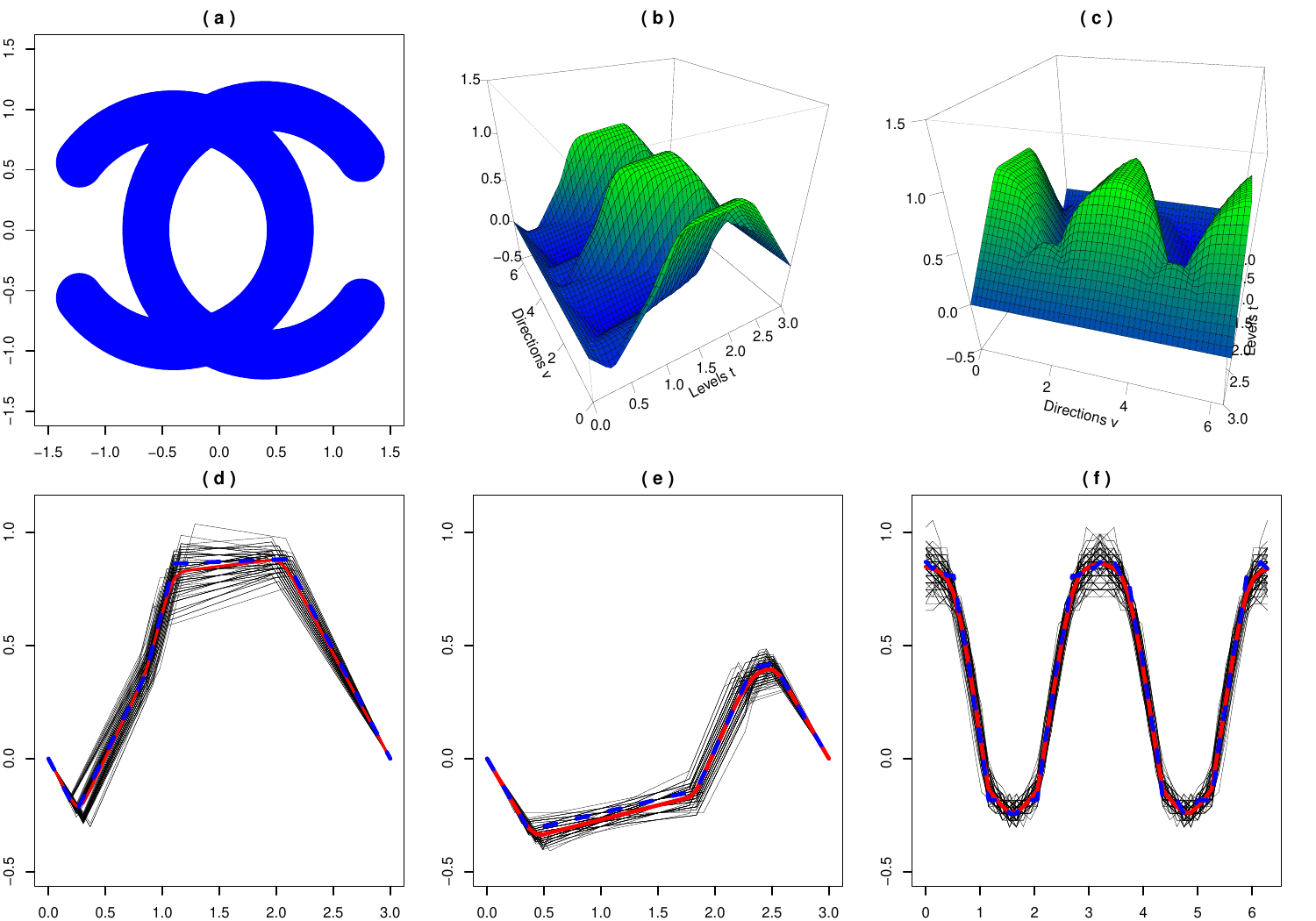}
    \caption{{\footnotesize Panel (a) presents a shape generated from the distribution in Eq.~\eqref{eq: randon shapes under null}. Panels (b) and (c) present the mean surface $(\vartheta, t) \mapsto \mathbb{E}\{\operatorname{SECT}(\nu, t)\}$ from different angles. Panels (d), (e), and (f) present the sample paths $\{\operatorname{SECT}(K_i)(\nu, t)\}_{i=1}^n$ with $\nu=(1,0)^\T$, $\nu=(0,1)^\T$, and $t=3/2$, respectively.} }
    \label{fig: SECT visualizations, random}
\end{figure}

\subsection{Computation of SECT}\label{section: Computation of SECT Using the Čech Complexes}

Here, we present an approach to computing the SECT of the shapes in Figures \ref{fig: SECT visualizations, deterministic} and \ref{fig: SECT visualizations, random}. Let $K \subseteq\mathbb{R}^d$ be a shape of interest. Suppose a finite set of points $\{x_i\}_{i=1}^I\subseteq K$ and a radius $r>0$ are properly chosen such that 
\begin{align}\label{eq: ball unions approx shapes}
    \begin{aligned}
        & K_t^\nu = \left\{x\in K \, \vert \, x\cdot\nu \le t-R \right\} \approx \bigcup_{i\in \mathfrak{I}_t^\nu} \overline{B(x_i,r)},\ \ \mbox{ for all }t\in[0,T]\mbox{ and }\nu\in\mathbb{S}^{d-1},\\
    & \mbox{ where }\, \mathfrak{I}_t^\nu \overset{\operatorname{def}}{=} \left\{i\in\mathbb{N} \, \Big\vert \, 1\le i\le I \, \mbox{ and }\, x_i\cdot\nu\le t-R\right\},
    \end{aligned}
\end{align}
and $\overline{B(x_i,r)}  \overset{\operatorname{def}}{=}  \{x\in\mathbb{R}^d:\Vert x-x_i\Vert\le r\}$ denotes a closed ball centered at $x_i$ with radius $r$. For example, when $d=2$, centers $x_i$ may be chosen as a subset of the grid points
\begin{align*}
    \left\{y_{j,j'} \overset{\operatorname{def}}{=} \left (-R+j\cdot\delta, \, -R+j'\cdot\delta \right)^\T \right\}_{j,j'=1}^J
\end{align*}
of the square $[-R,R]^2$ containing shape $K$, where $\delta=\frac{2R}{J}$ and radius $r=\delta$. Specifically, 
\begin{align}\label{eq: approximation using grid points}
    K_t^\nu \approx \bigcup_{y_{j,j'}\in K_t^\nu} \overline{B\left(y_{j,j'}, \, \delta\right)}\ \ \mbox{ for all }t\in[0,T]\mbox{ and }\nu\in\mathbb{S}^{d-1},
\end{align}
which is a special case of Eq.~\eqref{eq: ball unions approx shapes}. The shape approximation in Eq.~\eqref{eq: approximation using grid points} is illustrated by Supplementary Figures \ref{fig: Computing_SECT}(a) and (b).

\paragraph*{Čech complexes} The Čech Complex determined by the point set $\{x_i\}_{i\in \mathfrak{I}_t^\nu}$ and radius $r$ in Eq.~\eqref{eq: ball unions approx shapes} is defined as the following simplicial complex
\begin{align*}
    \check{C}_r\left( \{x_i\}_{i\in \mathfrak{I}_t^\nu} \right) \overset{\operatorname{def}}{=} \left\{ \operatorname{conv}\left(\{x_i\}_{i\in s}\right) \, \Bigg\vert \, s\in 2^{\mathfrak{I}_t^\nu} \mbox{ and } \bigcap_{i\in s}\overline{B(x_i,r)}\ne\emptyset\right\},
\end{align*}
where $\operatorname{conv}\left(\{x_i\}_{i\in s}\right)$ denotes the convex hull generated by points $\{x_i\}_{i\in s}$. The nerve theorem \citep[][Chapter III]{edelsbrunner2010computational} indicates that the Čech Complex $\check{C}_r\left( \{x_i\}_{i\in \mathfrak{I}_t^\nu} \right)$ and the union $\bigcup_{i\in \mathfrak{I}_t^\nu} \overline{B(x_i,r)}$ have the same homotopy type. Hence, they share the same Betti numbers, i.e.,
\begin{align*}
    \beta_k\Big(\check{C}_r\left( \{x_i\}_{i\in \mathfrak{I}_t^\nu} \right)\Big) = \beta_k\left(\bigcup_{i\in \mathfrak{I}_t^\nu} \overline{B(x_i,r)}\right),\ \ \mbox{ for all }k\in\mathbb{Z}.
\end{align*}
Using the shape approximation in Eq.~\eqref{eq: ball unions approx shapes}, we have the following approximation for ECC
\begin{align}\label{eq: ECC approximation using Cech complexes}
    \chi_t^{\nu}(K) \approx \sum_{k=0}^{d-1} (-1)^{k} \cdot \beta_k\Big(\check{C}_r\left( \{x_i\}_{i\in \mathfrak{I}_t^\nu} \right)\Big),\ \ \ t\in[0,T].
\end{align}
The method of computing the Betti numbers of simplicial complexes in Eq.~\eqref{eq: ECC approximation using Cech complexes} is standard and can be found in the literature \citep{edelsbrunner2010computational, niyogi2008finding}. Then, the SECT of $K$ is estimated using Eq.~\eqref{Eq: definition of SECT}. The smoothing effect of the integrals in Eq.~\eqref{Eq: definition of SECT} reduces the estimation error.

\begin{figure}[h]
    \centering
    \includegraphics[scale=0.58]{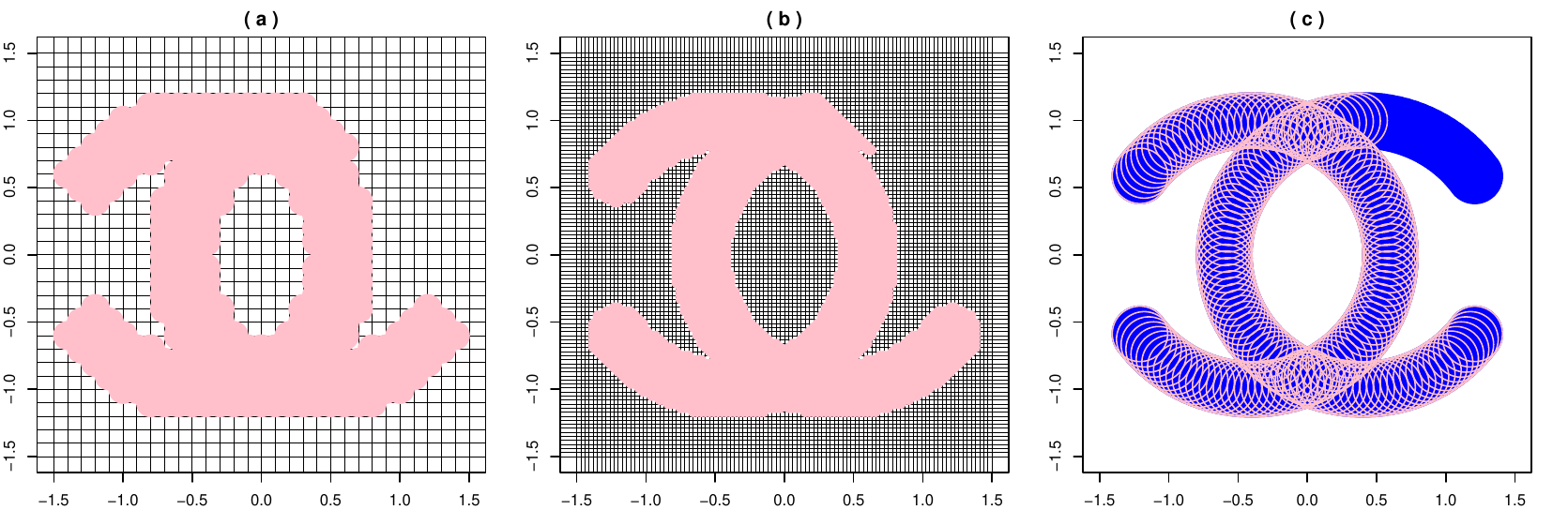}
    \caption{{\footnotesize Illustrations of the shape approximation in Eq.~\eqref{eq: ball unions approx shapes}. The shape $K$ of interest herein is equal to the $K^{(1)}$ defined in Eq.~\eqref{eq: example shapes K1 and K2} (see Supplementary Figure \ref{fig: SECT visualizations, deterministic}(a) and the blue shape in the panel (c) herein). We set $R=\frac{3}{2}$, $t=1$, and $\nu=\left(\sqrt{2}/2,\, \sqrt{2}/2\right)^\T$. Panels (a) and (b) specifically illustrate the approximation in Eq.~\eqref{eq: approximation using grid points} using grid points, and the pink shapes in the two panels present the union $\bigcup_{y_{j,j'}\in K_t^\nu} \overline{B (y_{j,j'},\, \delta )}$ in Eq.~\eqref{eq: approximation using grid points}; in panel (a), $J=30$ (i.e., $\delta=0.1$); in panel (b), $J=100$ (i.e., $\delta=0.03$). Panel (c) illustrates the the approximation in Eq.~\eqref{eq: ball unions approx shapes} with centers $\{x_i\}_i$ being the points in Eq.~\eqref{eq: the centers of our shape approximation} and $r=\frac{1}{5}$. Each pink circle in panel (c) presents a ball $\overline{B(x_i, r)}$.}}
    \label{fig: Computing_SECT}
\end{figure}

\paragraph*{Computing the SECT for our proof-of-concept and simulation examples} For the shape $K^{(1)}$ defined in Eq. \eqref{eq: example shapes K1 and K2} of Section \ref{Proof-of-Concept Simulation Examples I: Deterministic Shapes}, we estimate the SECT of $K^{(1)}$ using the aforementioned Čech complex approach with the following setup: $R=\frac{3}{2}$, $r=\frac{1}{5}$, and the point set $\{x_i\}_i$ is equal to the following collection
\begin{align}\label{eq: the centers of our shape approximation}
     \left\{\left(\frac{2}{5}+\cos t_j, \sin t_j\right) \,\Bigg\vert\, t_j=\frac{\pi}{5}+\frac{j}{J}\cdot\frac{8\pi}{5}\right\}_{j=1}^J\bigcup\left\{\left(-\frac{2}{5}+\cos t_j, \sin t_j\right) \,\Bigg\vert\, \frac{6\pi}{5}+\frac{j}{J}\cdot\frac{8\pi}{5}\right\}_{j=1}^J,
\end{align}
where $J$ is a sufficiently large integer. We implement $J=100$ in our proof-of-concept example. Supplementary Figure \ref{fig: Computing_SECT}(c) illustrates the shape approximation using this setup. The SECT for other shapes in our proof-of-concept/simulation examples is estimated in the similar way.

\section{Alignment}\label{section: ECT Alignment}

In this subsection, we introduce an important preprocessing step based on the ECT to align shapes and discuss its motivation.

\subsection{Motivation}\label{section: Motivation for the ECT-alignment}

It is widely accepted in geometric morphometrics that the intrinsic information of a shape does not change if one translates, rotates, or reflects (if handedness is ignored) the shape, which is mathematically represented as follows. Let $\operatorname{E}(d)$ be the group of rigid motions in $\mathbb{R}^d$ (i.e., the Euclidean group, it comprises translations, rotations, reflections, and finite combinations of them). Shapes $K$ and $K'$ are identified as the same if there exists $\varphi\in\operatorname{E}(d)$ such that $K=\varphi(K')$ (denoted as $K \sim K'$). One can verify that $\sim$ is an equivalence relation and the shape collection of interest is the quotient space $\mathcal{S}_{R,d}^M/\sim$. Since the $\rho$ in Eq.~\eqref{Eq: distance between shapes} is a metric, we may detect the equivalence between $K$ and $K'$ via the following
\begin{align}\label{eq: inf for alignment}
    \Tilde{\rho}\left( K, \, K' \right) \overset{\operatorname{def}}{=} \inf_{\varphi\in\operatorname{E}(d)} \left\{ \rho \left( K, \, \varphi(K') \right) \right\}.
\end{align}
That is, $K\sim K'$ if $\Tilde{\rho}\left( K, \, K' \right)=0$. Eq.~\eqref{eq: inf for alignment} resembles the Procrustes distance in statistical shape analysis \citep[][Section 2.1]{gao2019gaussianmorphometrics}.

To compute the $\Tilde{\rho}$ in Eq.~\eqref{eq: inf for alignment}, one needs $\operatorname{ECT}(\varphi(K))=\{\chi_t^\nu(\varphi(K)) \,\vert\,(\nu,t)\in\mathbb{S}^{d-1}\times[0,T]\}$ for all $\varphi\in\operatorname{E}(d)$. Computing the ECT of every shape $\varphi(K)$ for all $\varphi\in\operatorname{E}(d)$ can be computationally expensive. We may implement the following to bypass the infeasible computation: We compute the ECT of only one shape $K$; we apply the ``dual motion" $\varphi_*$ to the computed $\operatorname{ECT}(K)$ such that $\varphi_*\operatorname{ECT}(K)=\operatorname{ECT}(\varphi(K))$. Then, it suffices to derive the dual motion $\varphi_*$ and show that $\varphi_*$ is computationally efficient for every $\varphi\in\operatorname{E}(d)$. By the structure of $\operatorname{E}(d)$, we only need to derive the dual motions of translations and $\operatorname{O}(d)$-actions, where $\operatorname{O}(d)$ is the orthogonal group in dimension $d$ (rotations and reflections). 

We first consider translations. For any $\eta\in\mathbb{R}^d$, denote $K+\eta\overset{\operatorname{def}}{=}\{x+\eta \,\vert\, x\in K\}$ and assume $K+\eta\subseteq B(0,R)$ without loss of generality. One can then show the following 
\begin{align}\label{eq: translation for ECT}
    \chi_t^\nu\left( K+\eta \right) = \chi_{t-\eta\cdot\nu}^\nu(K).
\end{align}
That is, if $\varphi: x\mapsto x+\eta$, then its dual $\varphi_*: \chi_t^\nu \mapsto \chi_{t-\eta\cdot\nu}^\nu$. Secondly, we consider $\operatorname{O}(d)$-actions. For any $\pmb{A}\in\operatorname{O}(d)$, denote $\pmb{A}K\overset{\operatorname{def}}{=}\{\pmb{A}x \,\vert\,x\in K\}$. One can show the following
\begin{align}\label{eq: O(d) for ECT}
    \chi_t^\nu(\pmb{A}K)=\chi_t^{\pmb{A}^{-1}\nu}(K)
\end{align}
\citep[for a generalized version, see][]{meng2023Inference}. That is, if $\varphi:x\mapsto\pmb{A}x$, then $\varphi_*:\chi_t^\nu\mapsto\chi_t^{\pmb{A}^{-1}\nu}$. Furthermore, Eq.~\eqref{eq: translation for ECT} and Eq.~\eqref{eq: O(d) for ECT} imply that rigid motions do not influence the qualitative properties of SECT (e.g., measurability, Sobolev-regularity, and $\frac{1}{2}$-Hölder continuity).

\subsection{ECT Alignment}

In almost all shape analysis studies, the initial step is to align the shapes. The objective of the alignment is to mitigate the difference between two shapes caused by rigid motions. Suppose $K^\diamondsuit$ is the ``standard shape" as a template. Motivated by Eq.~\eqref{eq: inf for alignment}, we align each shape $K$ to be $\varphi^\blacklozenge(K)$ before any statistical inference, where
\begin{align}\label{eq: ECT alignment}
    \varphi^\blacklozenge \overset{\operatorname{def}}{=} \argmin_{\varphi\in\operatorname{E}(d)} \left\{ \rho \left( K^\diamondsuit, \, \varphi(K) \right) \right\}.
\end{align}
Following the discussion at the end of Section \ref{section: Motivation for the ECT-alignment}, the ECT alignment defined in Eq.~\eqref{eq: ECT alignment} does not change the qualitative properties of the SECT.A numerical approach for the minimization in Eq.~\eqref{eq: ECT alignment}, as well as its proof-of-concept studies, is provided in Supplementary Material of \cite{wang2021statistical} (Section 4). We apply this numerical approach in our paper.

\section{Numerical Foundation for Hypothesis Testing}\label{The Numerical foundation for Hypothesis Testing}

In Section \ref{section: hypothesis testing}, we proposed a fdANOVA approach to testing the hypotheses in Eq.~\eqref{eq: the main hypotheses} based on the $\{\xi_{l,i}\}_{l}$ defined in Eq.~\eqref{eq: def of the xi statistic}. In applications, neither the mean $m_{\nu}^{(j)}(t)$ nor the covariance $\Xi_{\nu}(s,t)$ is known. Hence, the corresponding KL expansion in Eq.~\eqref{eq: rigorous KL expansions of SECT} is unavailable, and the proposed hypothesis testing approach is not directly applicable. Here, motivated by Section 4.3.2 of \cite{williams2006gaussian}, we propose a method for estimating the $\{\xi_{l,i}\}_{l}$ defined in Eq.~\eqref{eq: def of the xi statistic}.

For random shapes $\{K_i^{(j)}\}_{i=1}^n\overset{\operatorname{i.i.d.}}{\sim}\mathbb{P}^{(j)}$, with $j\in\{1,2\}$, we compute their corresponding SECT in finitely many directions and sublevel sets as discussed at the end of Section \ref{section: distributions of Gaussian bridge} to get $\{\operatorname{SECT}(K_i^{(j)})(\nu_p;t_q)\,|\, p=1,\cdots, \Gamma \mbox{ and }q=1,\cdots,\Delta\}_{i=1}^n$ for $j\in\{1,2\}$, where $t_q  =  \frac{T}{\Delta}q$. The SECT of all shapes $K_i^{(j)}$ in the two collections are computed in the same collection of directions $\{\nu_p\}_{p=1}^\Gamma$ and at the same collection of sublevel sets $\{t_q\}_{q=1}^\Delta$. Using Eq.~(5.3) of \cite{zhang2013analysis}, we estimate the mean $\widehat{m}_{\nu_p}^{(j)}(t_q)$ of $m_{\nu_p}^{(j)}(t_q)$ at level $t_q$ by taking the sample mean of $\{\operatorname{SECT}(K_i^{(j)})(\nu_p;t_q)\}_{i=1}^n$ across $i\in\{1,\cdots,n\}$. Then, we estimate the distinguishing direction $\nu^*$ by 
\begin{align}\label{eq: estimated distinguishing direction}
    \widehat{\nu}^* \overset{\operatorname{def}}{=} \argmax_{\nu_p} \left[ \max_{t_q} \left\{\left\vert\, \widehat{m}_{\nu_p}^{(1)}(t_q) - \widehat{m}_{\nu_p}^{(2)}(t_q) \,\right\vert \right\} \right].
\end{align}

Under Assumption \ref{assumption: equal covariance assumption}, we estimate the covariance matrix $\left(\Xi_{\nu^*}(t_{q'},t_{q})\right)_{q,q'=1,\cdots,\Delta}$ using the pooled sample covariance matrix $\pmb{C}  =  \left(\widehat{\Xi}_{\nu^*}(t_{q'},t_{q}) \right)_{q,q'=1,\cdots,\Delta}$ defined by
\begin{align}\label{eq: centered sample vectors}
    \widehat{\Xi}_{\nu^*}(t_{q'},t_{q}) \overset{\operatorname{def}}{=} \frac{1}{2n-1} \sum_{j=1}^2 \sum_{i=1}^n \left( \operatorname{SECT}(K_i^{(j)})(\widehat{\nu}^*,\,t_{q'}) - \widehat{m}_{\widehat{\nu}^*}^{(j)}(t_{q'}) \right)\cdot \left( \operatorname{SECT}(K_i^{(j)})(\widehat{\nu}^*,\,t_{q}) - \widehat{m}_{\widehat{\nu}^*}^{(j)}(t_{q}) \right),
\end{align}
which is based on Eq.~(5.3) of \cite{zhang2013analysis}.

Since the eigenfunctions $\{\phi_l\}_{l=1}^\infty$ and eigenvalues $\{\lambda_l\}_{l=1}^\infty$ satisfy $\lambda_l\phi_l=\int_0^T \phi_l(s)\cdot\Xi_{\nu^*}(s,\cdot) \, ds$, we have the following approximation
\begin{align*}
    \lambda_l\phi_l(t_q)=\int_0^T \phi_l(s)\cdot\Xi_{\nu^*}(s,t_q) \, ds\approx\frac{T}{\Delta}\sum_{q'=1}^\Delta \phi_l(t_{q'})\cdot\Xi_{\nu^*}(t_{q'}, t_q) \approx \frac{T}{\Delta}\sum_{q'=1}^\Delta \phi_l(t_{q'})\cdot\widehat{\Xi}_{\nu^*}(t_{q'}, t_q),
\end{align*}
which is represented in the following matrix form
\begin{align}\label{eq: matrix form, approximate integrals by riemann sums}
    &\lambda_l\begin{pmatrix}
    \phi_l(t_1)\\
    \vdots \\
    \phi_l(t_\Delta)
    \end{pmatrix}\approx\frac{T}{\Delta}
    \begin{pmatrix}
   \widehat{\Xi}_{\nu^*}(t_{1}, t_1) & \ldots & \widehat{\Xi}_{\nu^*}(t_\Delta, t_1) \\
    \vdots & \ddots & \vdots \\
    \widehat{\Xi}_{\nu^*}(t_\Delta, t_2) & \ldots & \widehat{\Xi}_{\nu^*}(t_\Delta, t_\Delta)
    \end{pmatrix}
    \begin{pmatrix}
    \phi_l(t_1)\\
    \vdots \\
    \phi_l(t_\Delta)
    \end{pmatrix}.
\end{align}
We denote the eigenvectors and eigenvalues of $\pmb{C}$ as $\{\pmb{v}_l=(v_{l,1}, \cdots, v_{l,\Delta})^\T\}_{l=1}^\Delta$ and $\{\Lambda_l\}_{l=1}^\Delta$, respectively. The following equation motivates the estimator $\phi_l(t_q)\approx \widehat{\phi}_l(t_q) \overset{\operatorname{def}}{=} \sqrt{\frac{\Delta}{T}} \cdot v_{l,q}$, for all $l\in\{1,\cdots,\Delta\}$,
\begin{align*}
    \sum_{q=1}^\Delta v_{l,q}^2 = \Vert \pmb{v}_l \Vert^2 = 1 = \int_0^T\vert \phi_l(t)\vert^2 dt \approx \frac{T}{\Delta}\sum_{q=1}^\Delta \left(\phi_l(t_q)\right)^2 = \sum_{q=1}^\Delta \left(\sqrt{\frac{T}{\Delta}}\cdot\phi_l(t_q)\right)^2.
\end{align*}
The following equation motivates the estimator $\lambda_l\approx\widehat{\lambda}_l \overset{\operatorname{def}}{=} \frac{T}{\Delta}\Lambda_l$, for all $l\in\{1,\cdots,\Delta\}$,
\begin{align*}
    \lambda_l \left(\widehat{\phi}_l(t_1), \cdots, \widehat{\phi}_l(t_\Delta)\right)^\T &\approx \frac{T}{\Delta}\pmb{C}\left(\widehat{\phi}_l(t_1), \cdots, \widehat{\phi}_l(t_\Delta)\right)^\T  \\
    & = \sqrt{\frac{T}{\Delta}} \pmb{C} \pmb{v}_l \\
    &= \sqrt{\frac{T}{\Delta}} \Lambda_l \pmb{v}_l \\
    &= \left(\frac{T}{\Delta}\Lambda_l\right) \left(\widehat{\phi}_l(t_1), \cdots, \widehat{\phi}_l(t_\Delta)\right)^\T.
\end{align*}
Additionally, we estimate the $L$ defined in Eq.~\eqref{eq: def of L} by the following
\begin{align}\label{eq: estimated L}
    L\approx \widehat{L} \overset{\operatorname{def}}{=} \max\{1,\, \widehat{\mathcal{L}}\},\ \ \ \widehat{\mathcal{L}} \overset{\operatorname{def}}{=} \min \left\{ l=1,\cdots, \Delta \,\Bigg\vert\, \frac{\sum_{l'=1}^l \max\{ \widehat{\lambda}_{l'}, \, 0 \}}{\sum_{l^{''}=1}^\Delta \max\{ \widehat{\lambda}_{l^{''}}, \, 0\}} >0.95\right\},
\end{align}
where we use $\max\{ \widehat{\lambda}_{l^{''}}, \, 0\}$ to compensate for when the estimated eigenvalues may be numerically negative in applications. We estimate the $\xi_{l,i}$ defined in Eq.~\eqref{eq: def of the xi statistic} by the following
\begin{align}\label{eq: def of xi_hat}
    \xi_{l,i} \approx \widehat{\xi}_{l,i} \overset{\operatorname{def}}{=} \frac{1}{\sqrt{2\widehat{\lambda}_l}} \cdot \frac{T}{\Delta} \sum_{q=1}^\Delta \left\{ \operatorname{SECT}(K_i^{(1)})(\widehat{\nu}^*,\,t_q) - \operatorname{SECT}(K_i^{(2)})(\widehat{\nu}^*,\,t_q) \right\} \widehat{\phi}_l(t_q),
\end{align}
for $l=1,\ldots,\widehat{L}$ and $i=1,\ldots,n$. Then, we implement the $\chi^2$-test in Eq.~\eqref{eq: rejection region} as follows
\begin{align}\label{eq: numerical chisq rejection region}
    \sum_{l=1}^{ \widehat{L}}\left(\frac{1}{\sqrt{n}}\sum_{i=1}^n \widehat{\xi}_{l,i}\right)^2 > \chi^2_{\widehat{L}, 1-\alpha} = \mbox{ the $1-\alpha$ lower quantile of the $\chi^2_{\widehat{L}}$ distribution}.
\end{align}
We encapsulate the numerical procedures for the $\chi^2$-test above by Algorithm \ref{algorithm: testing hypotheses on mean functions; Appendix}.
\begin{algorithm}[h]
\renewcommand{\thealgorithm}{1}\caption{: $\chi^2$-test.}\label{algorithm: testing hypotheses on mean functions; Appendix}
	\begin{algorithmic}[1]
		\INPUT
        \noindent (i) SECT of two collection of shapes $\{\operatorname{SECT}(K_i^{(j)})(\nu_p,\,t_q):p=1,\cdots,\Gamma \mbox{ and } q=1,\cdots,\Delta\}_{i=1}^n$ for $j\in\{1,2\}$;
        (ii) confidence level $1-\alpha$ with $\alpha\in(0,1)$.
		\OUTPUT \texttt{Accept} or \texttt{Reject} the null hypothesis $H_0$ in Eq.~\eqref{eq: the main hypotheses}.
		\STATE For each $j\in\{1,2\}$, compute $\widehat{m}_{\nu_p}^{(j)}(t_q) \overset{\operatorname{def}}{=}$ sample mean of $\{\operatorname{SECT}(K_i^{(j)})(\nu_p;t_q)\}_{i=1}^n$ across $i\in\{1,\cdots,n\}$.
		\STATE Compute the estimated distinguishing direction $\widehat{\nu}^*$ using Eq.~\eqref{eq: estimated distinguishing direction}.
		\STATE Compute $\pmb{C}=(\widehat{\Xi}_{\nu^*}(t_{q'},t_{q}))_{q,q'=1,\cdots,\Delta}$ using Eq.~\eqref{eq: centered sample vectors}.
		\STATE Compute the eigenvectors $\{\pmb{v}_l\}_{l=1}^\Delta$ and eigenvalues $\{\Lambda_l\}_{l=1}^\Delta$ of the matrix $\pmb{C}$.
		\STATE Compute $\widehat{\phi}_l(t_q) \overset{\operatorname{def}}{=}  \sqrt{\frac{\Delta}{T}} v_{l,q}$ and $\widehat{\lambda}_l \overset{\operatorname{def}}{=}  \frac{T}{\Delta}\Lambda_l$ for all $l=1,\cdots,\Delta$.
		\STATE Compute $\widehat{L}$ using Eq.~\eqref{eq: estimated L}.
		\STATE Compute $\{\xi_{l,i}:l=1,\cdots,\widehat{L}\}_{i=1}^n$ using Eq.~\eqref{eq: def of xi_hat}, test null $H_0$ using Eq.~\eqref{eq: numerical chisq rejection region}, and report the output.
		\end{algorithmic}
\end{algorithm}

In addition to the $\chi^2$-test detailed in Algorithm \ref{algorithm: testing hypotheses on mean functions; Appendix}, we also propose a permutation-based test as an alternative approach for assessing the statistical hypotheses in Eq.~\eqref{eq: the main hypotheses}. The main idea behind the permutation test is that, under the null hypothesis, shuffling the group labels of shapes should not heavily change the test statistic of interest. To perform the permutation-based test, we first apply Algorithm \ref{algorithm: testing hypotheses on mean functions; Appendix} to our original data and then repeatedly re-apply Algorithm \ref{algorithm: testing hypotheses on mean functions; Appendix} to shapes with shuffled labels.$^\S$\footnote{$\S$: When we apply Algorithm \ref{algorithm: testing hypotheses on mean functions; Appendix} to the original SECT, the result of Eq.~\eqref{eq: estimated L} is denoted as $\widehat{L}_0$. When we apply Algorithm \ref{algorithm: testing hypotheses on mean functions; Appendix} to the shuffled SECT, the $\widehat{L}$ resulting from Eq.~\eqref{eq: estimated L} may differ from $\widehat{L}_0$. To make the comparison between $\mathfrak{S}_0$ and $\mathfrak{S}_{(k^*)}$ fair (see the last step of Algorithm \ref{algorithm: permutation-based testing hypotheses on mean functions}), we set $\widehat{L}$ to be $\widehat{L}_0$.} We then compare how the test statistics derived from the original differ from those computed on the shuffled data. The details of this permutation-based approach are provided in Algorithm \ref{algorithm: permutation-based testing hypotheses on mean functions}. Simulation studies in Section \ref{section: Simulation experiments} show that Algorithm \ref{algorithm: permutation-based testing hypotheses on mean functions} can eliminate the moderate type I error inflation of Algorithm \ref{algorithm: testing hypotheses on mean functions; Appendix}; however, the power under the alternative for Algorithm \ref{algorithm: permutation-based testing hypotheses on mean functions} is moderately weaker than that of Algorithm \ref{algorithm: testing hypotheses on mean functions; Appendix}. 

\begin{algorithm}[h]
\renewcommand{\thealgorithm}{2}\caption{: Permutation-based $\chi^2$-test.}\label{algorithm: permutation-based testing hypotheses on mean functions}
	\begin{algorithmic}[1]
		\INPUT
        \noindent (i) SECT of two collections of shapes $\{\operatorname{SECT}(K_i^{(j)})(\nu_p,\,t_q):p=1,\cdots,\Gamma \mbox{ and } q=1,\cdots,\Delta\}_{i=1}^n$ for $j\in\{1,2\}$;
        (ii) desired confidence level $1-\alpha$ with significance $\alpha\in(0,1)$; (iii) the number of permutations $\Pi$.
		\OUTPUT \texttt{Accept} or \texttt{Reject} the null hypothesis $H_0$ in Eq.~\eqref{eq: the main hypotheses}.
		\STATE Apply Algorithm \ref{algorithm: testing hypotheses on mean functions; Appendix} to the original input SECT data, compute $\widehat{L}_0$ using Eq.~\eqref{eq: estimated L} (see footnote $\S$), and compute the $\chi^2$-test statistic denoted as $\mathfrak{S}_0$ using Eq.~\eqref{eq: numerical chisq rejection region}.
		\FORALL{$k=1,\cdots,\Pi$, }
		\STATE Randomly permute the group labels $j\in\{1,2\}$ of the input SECT data.
		\STATE Apply Algorithm \ref{algorithm: testing hypotheses on mean functions; Appendix} to the permuted SECT data while setting $\widehat{L}$ to be the $\widehat{L}_0$, instead of using Eq.~\eqref{eq: estimated L}, and compute a $\chi^2$-test statistic $\mathfrak{S}_k$ using Eq.~\eqref{eq: numerical chisq rejection region}.
		\ENDFOR
  \STATE Sort the sequence $\{\mathfrak{S}_k\}_{k=1}^\Pi$ into ascending order, i.e., we have the ordered values $\{\mathfrak{S}_{(k)}\}_{k=1}^\Pi$ such that $\mathfrak{S}_{(1)} \le \mathfrak{S}_{(2)} \le \ldots \le \mathfrak{S}_{(\Pi)}$.
		\STATE Compute $k^* \overset{\operatorname{def}}{=}[(1-\alpha)\cdot\Pi]\overset{\operatorname{def}}{=}$ the largest integer smaller than $(1-\alpha)\cdot\Pi$. 
		\STATE \texttt{Reject} the null hypothesis $H_0$ if $\mathfrak{S}_0>\mathfrak{S}_{(k^*)}$ and report the output.
		\end{algorithmic}
\end{algorithm}

\section{Randomization-style Null Hypothesis\\ Significance Test}\label{Randomization-style Null Hypothesis Significance Test}

In Section \ref{section: Simulation experiments}, we compare our proposed Algorithms \ref{algorithm: testing hypotheses on mean functions; Appendix} and \ref{algorithm: permutation-based testing hypotheses on mean functions} with the ``randomization-style null hypothesis significance test (NHST)" \citep[][particularly Section 5.3 therein]{robinson2017hypothesis}, which is designed to test the following hypotheses
\begin{align}\label{eq: hypotheses for randomization-style NHST}
H_0:\ \ \mathbb{P}^{(1)}=\mathbb{P}^{(2)}\ \ \ vs.\ \ \ H_1:\ \ \mathbb{P}^{(1)}\ne\mathbb{P}^{(2)},
\end{align}
The randomization-style NHST is based on the permutation test and the following loss function
\begin{align}\label{eq: randomization-style NHST}
    F\left(\{K_i^{(1)}\}_{i=1}^n, \, \{K_i^{(2)}\}_{i=1}^n\right) \overset{\operatorname{def}}{=}\frac{1}{2n(n-1)}\sum_{k,l=1}^n\left\{\rho\left(K_k^{(1)},\, K_l^{(1)}\right)+\rho\left(K_k^{(2)},\, K_l^{(2)}\right)\right\},
\end{align}
where $\rho$ is the distance function defined in Eq.~\eqref{Eq: distance between shapes}. 

For a given discrete ECT, where $\{\operatorname{ECT}(K_i^{(j)})(\nu_p,\,t_q):p=1,\cdots,\Gamma \mbox{ and } q=1,\cdots,\Delta\}_{i=1}^n$, we may adopt the following approximation
\begin{align}\label{eq: approximate rho}
    \rho\left(K_k^{(j)},\, K_l^{(j)}\right) \approx \sup_{p=1,\ldots,\Gamma}\left(\sum_{q=1}^\Delta \left\vert \operatorname{ECT}(K_k^{(j)})(\nu_p,\, t_q)-\operatorname{ECT}(K_l^{(j)})(\nu_p,\, t_q)\right\vert^2\right)^{1/2}.
\end{align}
We apply Algorithm \ref{algorithm: randomization-style NHST} to implement the randomization-style NHST.

\begin{algorithm}[h]
\renewcommand{\thealgorithm}{3}
		\caption{: Randomization-style NHST}

 \label{algorithm: randomization-style NHST}
	\begin{algorithmic}[1]
		\INPUT
        \noindent (i) ECT of two collection of shapes $\{\operatorname{ECT}(K_i^{(j)})(\nu_p,\,t_q):p=1,\cdots,\Gamma \mbox{ and } q=1,\cdots,\Delta\}_{i=1}^n$ for $j\in\{1,2\}$;
        (ii) confidence level $1-\alpha$ with $\alpha\in(0,1)$; (iii) the number of permutations $\Pi$.
		\OUTPUT \texttt{Accept} or \texttt{Reject} the null hypothesis $H_0$ in Eq.~\eqref{eq: hypotheses for randomization-style NHST}.
		\STATE Apply Eq.~\eqref{eq: randomization-style NHST} and Eq.~\eqref{eq: approximate rho} to the original input ECT data and compute the value of the loss $\mathfrak{S}_0 \overset{\operatorname{def}}{=} F(\{K_i^{(1)}\}_{i=1}^n, \, \{K_i^{(2)}\}_{i=1}^n)$.
		\FORALL{$k=1,\cdots,\Pi$, }
		\STATE Randomly permute the group labels $j\in\{1,2\}$ of the input ECT data.
		\STATE Apply Eq.~\eqref{eq: randomization-style NHST} and Eq.~\eqref{eq: approximate rho} to the permuted ECT data and compute the corresponding value of the loss function $F$; denote the value of the loss by $\mathfrak{S}_k$.
		\ENDFOR
  \STATE Sort the sequence $\{\mathfrak{S}_k\}_{k=1}^\Pi$ into ascending order, i.e., we have the ordered values $\{\mathfrak{S}_{(k)}\}_{k=1}^\Pi$ such that $\mathfrak{S}_{(1)} \le \mathfrak{S}_{(2)} \le \ldots \le \mathfrak{S}_{(\Pi)}$.
		\STATE Compute $k^* \overset{\operatorname{def}}{=}[\alpha\cdot\Pi]\overset{\operatorname{def}}{=}$ the largest integer smaller than $\alpha\cdot\Pi$. 
		\STATE \texttt{Reject} the null hypothesis $H_0$ if $\mathfrak{S}_0<\mathfrak{S}_{(k^*)}$ and report the output.
		\end{algorithmic}
\end{algorithm}

\section{Landmark-based Permutation Test}\label{section: Landmark-based Permutation Test}

Landmarks are widely used in geometric morphometrics \citep{kendall1977diffusion, kendall1984shape, kendall1989survey, gao2019gaussian, gao2019gaussianmorphometrics}. In this section, we introduce a landmark-based hypothesis testing approach for distinguishing between shape collections. We compare this approach with our proposed Algorithm \ref{algorithm: testing hypotheses on mean functions; Appendix} and \ref{algorithm: permutation-based testing hypotheses on mean functions} using the mandibular molar data presented in Figure \ref{fig: Teeth} (see Section \ref{section: Applications}). This section is divided into the following two subsections:
\begin{itemize}
    \item In Section \ref{section: Permutation Test Using Landmark-based Procrustes Distances}, motivated by \cite{robinson2017hypothesis}, we propose a permutation test based on Procrustes distances, which is presented in Algorithm \ref{algorithm: Permutation test using landmark-based Procrustes distances}.


    \item In Section \ref{section: Gaussian Process-based Landmarks and Correspondence}, we first use the ``Gaussian process landmarking (GPL) algorithm" \citep{gao2019gaussian, gao2019gaussianmorphometrics} to generate landmarks on the mandibular molars analyzed in Section \ref{section: Applications}. Then, we apply the ``bounded distortion Gaussian process landmark matching" method \citep[][Section 4.2]{gao2019gaussianmorphometrics} to generate the correspondence between each pair of mandibular molars. Lastly, we implement the continuous Procrustes distance $\varrho_P$ induced by the correspondence. The Procrustes distance $\varrho_P$ can be used as an input to Algorithm \ref{algorithm: Permutation test using landmark-based Procrustes distances}.
\end{itemize}

\subsection{Permutation Test Using Landmark-based Distances}\label{section: Permutation Test Using Landmark-based Procrustes Distances}

In this subsection, to test the hypotheses in Eq.~\eqref{eq: hypotheses for randomization-style NHST}, we propose a permutation test using landmark-based distances. Similar to Algorithm \ref{algorithm: randomization-style NHST}, we need the following loss function
\begin{align}\label{eq: loss function for the Procrustes distance-based permutation test}
    G\left(\{K_i^{(1)}\}_{i=1}^n, \, \{K_i^{(2)}\}_{i=1}^n\right) \overset{\operatorname{def}}{=}\frac{1}{2n(n-1)}\sum_{k,l=1}^n\left\{\varrho\left(K_k^{(1)},\, K_l^{(1)}\right)+\varrho\left(K_k^{(2)},\, K_l^{(2)}\right)\right\},
\end{align}
where $\{K_i^{(j)}\}_{i=1}^n\overset{\operatorname{i.i.d.}}{\sim}\mathbb{P}^{(j)}$, for $j\in\{1,2\}$, and the $\varrho$ is a landmark-based distance. A comprehensive description of the permutation test using landmark-based distances is given in Algorithm \ref{algorithm: Permutation test using landmark-based Procrustes distances}. An example of the landmark-based distance $\varrho$ is the continuous Procrustes distance $\varrho_{P}$ defined in Section \ref{section: Gaussian Process-based Landmarks and Correspondence}. 

\begin{algorithm}[h]
\renewcommand{\thealgorithm}{4}
		\caption{: Permutation test using landmark-based Procrustes distances}

 \label{algorithm: Permutation test using landmark-based Procrustes distances}
	\begin{algorithmic}[1]
		\INPUT
        \noindent (i) Landmark-based distance $\varrho$;
        (ii) confidence level $1-\alpha$ with significance $\alpha\in(0,1)$; (iii) the number of permutations $\Pi$.
		\OUTPUT \texttt{Accept} or \texttt{Reject} the null hypothesis $H_0$ in Eq.~\eqref{eq: hypotheses for randomization-style NHST}.
		\STATE Compute the value of the loss function in Eq.~\eqref{eq: loss function for the Procrustes distance-based permutation test} using the original (un-shuffled) shapes. That is, compute the distances $\{\varrho(K_k^{(1)}, K_l^{(1)})\}_{k,l=1}^n$ and $\{\varrho(K_k^{(2)}, K_l^{(2)})\}_{k,l=1}^n$; then, compute the value $\mathfrak{S}_0 \overset{\operatorname{def}}{=} G(\{K_i^{(1)}\}_{i=1}^n, \, \{K_i^{(2)}\}_{i=1}^n)$ using Eq.~\eqref{eq: loss function for the Procrustes distance-based permutation test}.
		\FORALL{$k=1,\cdots,\Pi$, }
		\STATE Randomly permute the group labels $j\in\{1,2\}$ of the shapes $\{K_i^{(j)}\}_{i,j}$.
		\STATE Apply Eq.~\eqref{eq: loss function for the Procrustes distance-based permutation test} to the permuted shapes and compute the corresponding value of the loss function $G$; denote the value of the loss by $\mathfrak{S}_k$.
		\ENDFOR
  \STATE Sort the sequence $\{\mathfrak{S}_k\}_{k=1}^\Pi$ into ascending order, i.e., we have the ordered values $\{\mathfrak{S}_{(k)}\}_{k=1}^\Pi$ such that $\mathfrak{S}_{(1)} \le \mathfrak{S}_{(2)} \le \ldots \le \mathfrak{S}_{(\Pi)}$.
		\STATE Compute $k^* \overset{\operatorname{def}}{=}[\alpha\cdot\Pi]\overset{\operatorname{def}}{=}$ the largest integer smaller than $\alpha\cdot\Pi$. 
		\STATE \texttt{Reject} the null hypothesis $H_0$ if $\mathfrak{S}_0<\mathfrak{S}_{(k^*)}$ and report the output.
		\end{algorithmic}
\end{algorithm}

\subsection{Gaussian Process-based Landmarks and Correspondence}\label{section: Gaussian Process-based Landmarks and Correspondence}

This subsection is designed for the landmark-based analysis of mandibular molars presented in Section \ref{section: Applications}. In this section, we briefly introduce the ``Gaussian process landmarking (GPL) algorithm" and ``bounded distortion GPL matching" method that were developed in \cite{gao2019gaussian} and \cite{gao2019gaussianmorphometrics}. Following this, a continuous Procrustes distance \citep{al2013continuous} is defined to measure the dissimilarity between each pair of mandibular molars.

\subsubsection{Gaussian process landmarking algorithm}

The GPL algorithm and its applications to the anatomical surfaces represented by discrete triangular meshes have been developed in \cite{gao2019gaussian} and \cite{gao2019gaussianmorphometrics}. The algorithm is comprehensively described in Algorithm 2.1 of \cite{gao2019gaussianmorphometrics}, which is the one implemented in our paper. Briefly, the GPL algorithm takes a triangular mesh of interest and the number of desired landmarks (which is set to be 40 in both \cite{gao2019gaussianmorphometrics} and our paper) as inputs; it returns feature vertices (called landmarks) of the triangular mesh as the output. Figure 1 of \cite{gao2019gaussianmorphometrics} gives an example visually showing the performance of the GPL algorithm.

\subsubsection{Procrustes distance induced by the Gaussian Process-based Landmarks and Correspondence}

In addition to the GPL algorithm, \cite{gao2019gaussianmorphometrics} also introduced the ``bounded distortion GPL matching" method (Section 4.2 therein) to construct a correspondence map $f$ using the sampled landmarks. Specifically, suppose we are interested in two anatomical surfaces $K_1$ and $K_2$, which are the shapes of interest; landmarks $\{\zeta_i^{(1)}\}_{i=1}^{40}$ and $\{\zeta_i^{(2)}\}_{i=1}^{40}$ are sampled from the $K_1$ and $K_2$, respectively, using the GPL algorithm;\footnote{\cite{gao2019gaussianmorphometrics} set the number of landmarks on each anatomical surface to be 40 in the examples illustrated therein; we adopt the same choice.} the bounded distortion GPL matching method returns a correspondence map $f: K_1\rightarrow K_2$ using these landmarks as input.

With the obtained correspondence map $f:K_1\rightarrow K_2$, we may compute the following continuous Procrustes distance (also see \cite{al2013continuous} and Eq.~(4.3) of \cite{gao2019gaussianmorphometrics})
\begin{align}\label{eq: def of the continuous Procrustes distance}
    \varrho_{P} \left( K_1, \,K_2 \right) \overset{\operatorname{def}}{=} \left( \inf_{\varphi\in\operatorname{E}(3)} \int_{K_1} \left\Vert f(x)-\varphi(x) \right\Vert^2 d\operatorname{vol}_{K_1}(x)
    \right)^{1/2},
\end{align}
where $d\operatorname{vol}_{K_1}$ denotes the volume form of the surface $K_1$.

In this paper, we compute $\varrho_{P}$ using the approach proposed in \cite{gao2019gaussianmorphometrics}. Specifically, we implement the code provided in the GitHub repository of the authors.

\section{Applications to Silhouette Database}\label{section: Silhouette Database}

We use a subset of the silhouette database that includes three classes of shapes: apples, hearts, and children (see Supplementary Figure \ref{fig: Silhouette Database}; each class has 20 shapes). For each shape shown in Supplementary Figure \ref{fig: Silhouette Database}, we compute its SECT. Specifically, we compute the ECCs for 72 directions, evenly sampled over the interval $[0,2\pi]$. For each direction, we analyze 100 sublevel sets. We apply Algorithm \ref{algorithm: testing hypotheses on mean functions; Appendix} and \ref{algorithm: permutation-based testing hypotheses on mean functions} to test the hypothesis that shapes differ between classes and present the results in Table \ref{tab: Silhouette Database}. The p-values in Table \ref{tab: Silhouette Database} are either $\chi^2$-test p-values (from Algorithm \ref{algorithm: testing hypotheses on mean functions; Appendix}) or permutation-test p-values (from Algorithm \ref{algorithm: permutation-based testing hypotheses on mean functions} with $\Pi=1000$). In addition to testing differences between shape classes, we also apply the algorithms within each individual shape class. Specifically, for each shape class, we randomly split the class into two halves and test for differences between them using the algorithms. We repeat this random splitting procedure 100 times and present the corresponding p-values in Table \ref{tab: Silhouette Database} (rows 4-6). For each shape class, we summarize the 100 p-values by reporting their mean with the standard deviation given in parentheses.

\begin{figure}[h]
    \centering
    \includegraphics[scale=0.25]{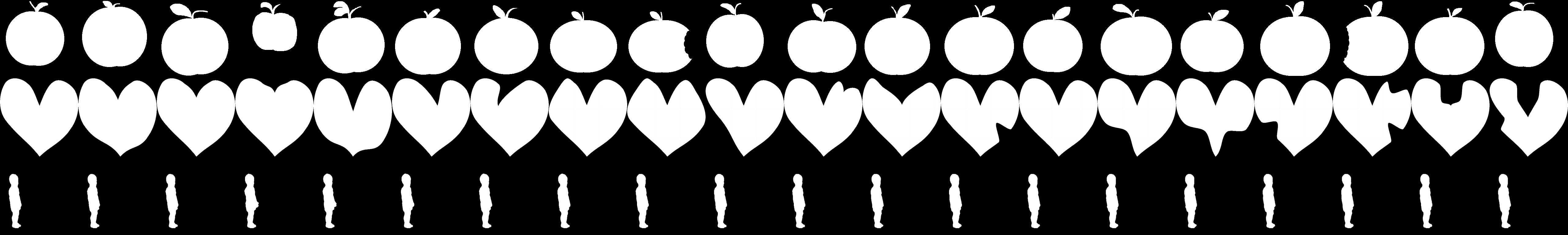}
    \caption{{\footnotesize Each row corresponds to one of the shape classes: apples, hearts, and children.}}
    \label{fig: Silhouette Database}
\end{figure}

\begin{table}[h]
\vspace*{0.5em}
\small
\def\arraystretch{0.8}
\caption{P-values of Algorithm \ref{algorithm: testing hypotheses on mean functions; Appendix} and \ref{algorithm: permutation-based testing hypotheses on mean functions} for the silhouette database.}
\label{tab: Silhouette Database}
\centering
\begin{tabular}{|c|c|c|}
\hline
& Algorithm \ref{algorithm: testing hypotheses on mean functions; Appendix}                                         & Algorithm \ref{algorithm: permutation-based testing hypotheses on mean functions}        \\[2pt] \hline

Apples vs. Hearts                       & $<0.01$ &  $<0.01$    \\[2pt]
Apples vs. Children                      & $<0.01$   & $<0.01$                   \\ [2pt]
Hearts vs. Children &$<0.01$ &$<0.01$\\[2pt]\hline
Apples vs. Apples &   0.26 (0.23)   &      0.46 (0.27)  \\[2pt]
Hearts vs. Hearts &    0.17 (0.16) &    0.47 (0.29) \\[2pt]
Children vs. Children & 0.39 (0.28) &   0.49 (0.30)\\[2pt]
\hline
\end{tabular}
\end{table}

Rows 1-3 of Table \ref{tab: Silhouette Database} show that our proposed Algorithm \ref{algorithm: testing hypotheses on mean functions; Appendix} and \ref{algorithm: permutation-based testing hypotheses on mean functions} can distinguish the shape classes of apples, hearts, and children presented in Supplementary Figure \ref{fig: Silhouette Database}. Rows 4-6 show that our proposed algorithms tend to minimize type I errors for the shape-valued data presented in Supplementary Figure \ref{fig: Silhouette Database}. The within-class p-values for each shape class reflect its homogeneity or heterogeneity (e.g., the class of children has the highest homogeneity among all three shape classes). Algorithm \ref{algorithm: permutation-based testing hypotheses on mean functions} tends to have larger p-values than Algorithm \ref{algorithm: testing hypotheses on mean functions; Appendix} when applied within each shape class, which is essentially due to the permutation procedure implemented in Algorithm \ref{algorithm: permutation-based testing hypotheses on mean functions}.

\section{Numerical Experiments on the ANOVA for Functional Data --- Existing Methods vs. Our Proposed Methods}\label{appendix: Numerical Experiments on One-way ANOVA --- Existing Methods vs. Our Proposed Methods}

In this section, we compare our proposed Algorithms \ref{algorithm: testing hypotheses on mean functions; Appendix} and \ref{algorithm: permutation-based testing hypotheses on mean functions} with twelve existing fdANOVA methods, which are listed below \citep[also see][]{gorecki2019fdanova}:
\begin{itemize}
    \item FP: permutation test based on a basis function representation and an F-statistic \citep{gorecki2015comparison}.

    \item CH and CS: $L^2$-norm-based parametric bootstrap tests for homoscedastic and heteroscedastic samples, respectively \citep{cuevas2004anova}.

    \item L2N and L2B: $L^2$-norm-based test with naive and bias-reduced method of estimation, respectively \citep{faraway1997regression, zhang2007statistical, zhang2013analysis}.

    \item L2b: $L^2$-norm-based bootstrap test \citep{zhang2013analysis}.

    \item FN and FB: F-type test with naive and bias-reduced method of estimation, respectively \citep{shen2004f, zhang2011statistical}.

    \item Fb: F-type bootstrap test \citep{zhang2013analysis}.

    \item TRP: tests based on random projections \citep{cuesta2010simple}. TRP-ANOVA indicates the TRP test using the ANOVA F statistic; TRP-ATS indicates the TRP test using the ANOVA-type statistic (ATS) proposed by \cite{brunner1997box}. TRP-WTPS indicates the TRP test using the Wald-type permutation statistic (WTPS) of \cite{pauly2015asymptotic}.
\end{itemize}
These existing methods are implemented using the \texttt{R} package \texttt{fdANOVA} \citep{gorecki2019fdanova}. Further details regarding these methods are available in \cite{gorecki2019fdanova}. As stated within \cite{gorecki2019fdanova}: ``satisfactory results are usually obtained by using the default values"; as a result, we use the default parameters provided by the package \texttt{fdANOVA}.

The application of the fdANOVA methods to the SECT is described in Section \ref{section: hypothesis testing}. The distributions that generate the simulated shapes in our numerical experiments are described in Section \ref{section: Simulation experiments}, particularly in Eq.~\eqref{eq: explicit P varepsilon}. The rejection rates of Algorithms \ref{algorithm: testing hypotheses on mean functions; Appendix}, \ref{algorithm: permutation-based testing hypotheses on mean functions}, \ref{algorithm: randomization-style NHST}, and the existing fdANOVA methods in different settings are presented in Table \ref{table: comparison with the existing ANOVA methods}, which complements the results in Table \ref{table: epsilon vs. rejection rates}. Notably, $\varepsilon=0$ indicates that the simulated shapes are generated under the null hypothesis. These rejection rates are also depicted in Figure \ref{fig: simulation visualizations}.

\begin{table}[h]
\centering
\caption{\footnotesize{Rejection rates (from 1000 experiments) for different indices $\varepsilon$ (significance $\alpha=0.05$).}}\label{table: comparison with the existing ANOVA methods}
\vspace*{0.5em}
\small
\def\arraystretch{0.8}
\begin{tabular}{|c|c|c|c|c|c|c|c|c|c|}
\hline
Indices $\varepsilon$ & 0.00 & 0.01 & 0.02 & 0.03 & 0.04 & 0.05 & 0.06 & 0.08 & 0.10 \\ \hline
Algorithm \ref{algorithm: testing hypotheses on mean functions; Appendix} & 0.118 & 0.161 & 0.315 & 0.519 & 0.785 & 0.910 &  0.975  & 0.990  & 1.000   \\
Algorithm \ref{algorithm: permutation-based testing hypotheses on mean functions} & 0.046  & 0.054 & 0.162 & 0.343  & 0.612   & 0.789  & 0.931   & 0.994 & 1.000 \\
Algorithm \ref{algorithm: randomization-style NHST} & 0.050 & 0.050 & 0.111 & 0.185 & 0.335 & 0.535 & 0.739 & 0.983 & 0.999 \\
FP         & 0.136 & 0.153 & 0.308  & 0.539  & 0.810  & 0.924 & 0.986 & 0.997 & 1.000  \\
CH         & 0.136 & 0.141 & 0.309  & 0.539  & 0.806  & 0.927 & 0.985 & 0.997  & 1.000 \\
CS         & 0.139 & 0.148 & 0.318  & 0.547  & 0.807  & 0.933 & 0.985 & 0.997  & 1.000  \\
L2N        & 0.138 & 0.147 & 0.319  & 0.550  & 0.809  & 0.931 & 0.985 & 0.997  & 1.000  \\
L2B        & 0.139 & 0.155 & 0.321  & 0.557  & 0.810  & 0.933 & 0.985 & 0.997  & 1.000  \\
L2b        & 0.138 & 0.147 & 0.319  & 0.551  & 0.803  & 0.929 & 0.985 & 0.997  & 1.000  \\
FN         & 0.136 & 0.144 & 0.316  & 0.544  & 0.807  & 0.929 & 0.985 & 0.997  & 1.000  \\
FB         & 0.138 & 0.147 & 0.318  & 0.546  & 0.808  & 0.929 & 0.985 & 0.997  & 1.000  \\
Fb         & 0.138 & 0.139 & 0.309  & 0.541  & 0.802  & 0.925 & 0.984 & 0.997  & 1.000  \\
TRP-ANOVA  & 0.073 & 0.091 & 0.256  & 0.510  & 0.782  & 0.931 & 0.980 & 0.997  & 1.000  \\
TRP-ATS    & 0.076 & 0.093 & 0.258  & 0.519  & 0.785  & 0.932 & 0.982 & 0.997  & 1.000  \\
TRP-WTPS   & 0.075 & 0.091 & 0.261  & 0.515  & 0.790  & 0.929 & 0.980 & 0.997  & 1.000  \\ \hline
\end{tabular}
\end{table}

\section{Trade-off Studies for Algorithms \ref{algorithm: testing hypotheses on mean functions; Appendix} and \ref{algorithm: permutation-based testing hypotheses on mean functions} --- Numbers of Directions and Levels, Sample Sizes, and Computational Cost}\label{section: Runtime}

In this section, we study the trade-offs among the following quantities using simulations:
\begin{itemize}
\item Runtimes/computational cost.
    \item Number $\Gamma$ of directions.
    \item Number $\Delta$ of levels.
    \item Number $n$ of shape pairs.
\end{itemize}

Using the random shape-generating mechanism introduced in Section \ref{section: Simulation experiments}, we simulate a data set for each fixed sample size \( n \in \{25, 50, 100\} \) and fixed \( \varepsilon \in \{0,\,0.05\} \). Specifically, the simulated data set is given by \(\{(K_i^{(0)},\, K_i^{(\varepsilon)})\}_{i=1}^n \overset{iid}{\sim} \mathbb{P}^{(0)}\otimes \mathbb{P}^{(\epsilon)}\). Using the simulated data set, we then apply Algorithm \ref{algorithm: testing hypotheses on mean functions; Appendix} and \ref{algorithm: permutation-based testing hypotheses on mean functions} to test the following the hypotheses, considering different combinations of \( \Gamma \in \{2,4,8\} \) and \( \Delta \in \{25, 50, 100\} \).
\begin{align}\label{eq: runtime study HT; Appendix}
    \begin{aligned}
        & H_0: m_\nu^{(0)}(t)=m_\nu^{(\varepsilon)}(t)\mbox{ for all }(\nu,t)\in\mathbb{S}^{d-1}\times[0,T]\ \ \ \\
        & vs. \ \ \ H_1:  m_\nu^{(0)}(t)\ne m_\nu^{(\varepsilon)}(t)\mbox{ for some }(\nu,t).
    \end{aligned}
\end{align} 
We repeat this procedure 20 times and explore all the combinations of \( n \in \{25, 50, 100\} \), \( \Gamma \in \{2,4,8\} \), and \( \Delta \in \{25, 50, 100\} \). The simulation results are presented in Tables \ref{table: runtime of Algorithm 1}, \ref{table: pvalue of Algorithm 1(0.05)}, and \ref{table: pvalue of Algorithm 1(0)} and summarized as follows:
\begin{enumerate}
    \item For each input combination of $n$, $\Gamma$, and $\Delta$ for each algorithm, the mean and standard deviation (in parentheses) of the runtimes across the 20 simulations are presented in Table \ref{table: runtime of Algorithm 1}.
    
    \item For each input combination of $n$, $\Gamma$, and $\Delta$ for each algorithm and when $\epsilon=0.05$, the mean and standard deviation (in parentheses) of the p-values across the 20 simulations are presented in Table \ref{table: pvalue of Algorithm 1(0.05)}. This table presents the trade-offs between the accuracy of our algorithms (in terms of p-values), sample size $n$, number $\Gamma$ of directions, and number $\Delta$ of levels under the alternative hypothesis $H_1$ in Eq.~\eqref{eq: runtime study HT; Appendix} with $\epsilon=0.05$. The p-values in Table \ref{table: pvalue of Algorithm 1(0.05)} demonstrate that the larger the $n$, $\Delta$, and $\Gamma$, the more likely Algorithm \ref{algorithm: testing hypotheses on mean functions; Appendix} and \ref{algorithm: permutation-based testing hypotheses on mean functions} reject the null hypothesis $H_0$ in Eq.~\eqref{eq: runtime study HT; Appendix}. Particularly, with a significance of $\alpha=0.05$, only $\Gamma=4$ directions are sufficient for Algorithm \ref{algorithm: testing hypotheses on mean functions; Appendix} and \ref{algorithm: permutation-based testing hypotheses on mean functions} to reject $H_0$ when $\Delta=100$ and $n=50$. Therefore, we adopt $\Gamma=4$ for the proof-of-concept purpose in our simulation studies presented in Section \ref{section: Simulation experiments}.
    
    \item For each input combination of $n$, $\Gamma$, and $\Delta$ for each algorithm and when $\epsilon=0$, the mean and standard deviation (in parentheses) of the p-values across the 20 simulations are presented in Table \ref{table: pvalue of Algorithm 1(0)}. This table presents the trade-offs between the accuracy of our algorithms (in terms of p-values), sample size $n$, number $\Gamma$ of directions, and number $\Delta$ of levels under the null hypothesis $H_0: \mathbb{P}^{(0)} =  \mathbb{P}^{(\epsilon)}$ with $\epsilon=0$. The p-values in Table \ref{table: pvalue of Algorithm 1(0)} demonstrate that no matter how small or large the $n$, $\Delta$, and $\Gamma$ are, neither Algorithm \ref{algorithm: testing hypotheses on mean functions; Appendix} nor \ref{algorithm: permutation-based testing hypotheses on mean functions} tends to distinguish $\mathbb{P}^{(0)}$ and $\mathbb{P}^{(\epsilon)}$ with $\epsilon=0$ falsely.
\end{enumerate}

The source code for the this study is publically available online through the link provided in the ``Software Availability" section. This study is conducted on a computer with an AMD Ryzen 7 5800H processor running at 3200 MHz using 16 GB of RAM, running Windows version 21H2.

\begin{table}[H]
\centering
\caption{Runtimes of Algorithm \ref{algorithm: testing hypotheses on mean functions; Appendix} and \ref{algorithm: permutation-based testing hypotheses on mean functions} (in seconds) based on $\Gamma, \Delta$, and $n$.}\label{table: runtime of Algorithm 1}
    \vspace*{0.5em}
\small
\def\arraystretch{0.8}
\begin{tabular}{|ccccc|}
\hline
\multicolumn{5}{|c|}{Algorithm \ref{algorithm: testing hypotheses on mean functions; Appendix}}\\
\hline
\multicolumn{1}{|c|}{Number of directions}                   & \multicolumn{1}{c|}{Number of levels}  & \multicolumn{1}{c|}{\textit{n} = 25}  & \multicolumn{1}{c|}{\textit{n} = 50}  &\multicolumn{1}{c|}{\textit{n} = 100}  \\ \hline
\multicolumn{1}{|c|}{} & \multicolumn{1}{c|}{$\Delta=25$} & \multicolumn{1}{c|}{0.75 (0.10)} & \multicolumn{1}{c|}{1.40 (0.08)} & \multicolumn{1}{c|}{2.23 (0.09)} \\ \cline{2-5} 
\multicolumn{1}{|c|}{$\Gamma=2$}                    & \multicolumn{1}{c|}{$\Delta=50$} & \multicolumn{1}{c|}{1.26 (0.09)} & \multicolumn{1}{c|}{2.36 (0.08)} & \multicolumn{1}{c|}{4.25 (0.10)} \\ \cline{2-5} 
\multicolumn{1}{|c|}{}                    & \multicolumn{1}{c|}{$\Delta=100$} & \multicolumn{1}{c|}{2.32 (0.09)} & \multicolumn{1}{c|}{4.52 (0.15)} & \multicolumn{1}{c|}{8.10 (0.21)} \\ \hline
\multicolumn{1}{|c|}{} & \multicolumn{1}{c|}{$\Delta=25$} & \multicolumn{1}{c|}{1.25 (0.09)} & \multicolumn{1}{c|}{2.26 (0.08)} & \multicolumn{1}{c|}{4.13 (0.15)} \\ \cline{2-5} 
\multicolumn{1}{|c|}{$\Gamma=4$}                    & \multicolumn{1}{c|}{$\Delta=50$} & \multicolumn{1}{c|}{2.26 (0.09)} & \multicolumn{1}{c|}{4.29 (0.14)} &\multicolumn{1}{c|}{8.00 (0.18)} \\ \cline{2-5} 
\multicolumn{1}{|c|}{}                    & \multicolumn{1}{c|}{$\Delta=100$} & \multicolumn{1}{c|}{4.28 (0.13)} & \multicolumn{1}{c|}{8.31 (0.17)} & \multicolumn{1}{c|}{15.77 (0.29)} \\ \hline
\multicolumn{1}{|c|}{} & \multicolumn{1}{c|}{$\Delta=25$} & \multicolumn{1}{c|}{2.27 (0.10)} & \multicolumn{1}{c|}{4.25 (0.12)} & \multicolumn{1}{c|}{8.06 (0.21)} \\ \cline{2-5} 
\multicolumn{1}{|c|}{$\Gamma=8$}                    & \multicolumn{1}{c|}{$\Delta=50$} & \multicolumn{1}{c|}{4.32 (0.13)} & \multicolumn{1}{c|}{8.39 (0.17)} &  \multicolumn{1}{c|}{15.81 (0.14)} \\ \cline{2-5} 
\multicolumn{1}{|c|}{}                    & \multicolumn{1}{c|}{$\Delta=100$} & \multicolumn{1}{c|}{8.38 (0.12)} & \multicolumn{1}{c|}{16.69 (0.34)} & \multicolumn{1}{c|}{32.69 (1.16)} \\ \hline
\multicolumn{5}{|c|}{Algorithm \ref{algorithm: permutation-based testing hypotheses on mean functions}}\\
\hline
\multicolumn{1}{|c|}{Number of directions}                   & \multicolumn{1}{c|}{Number of levels}  & \multicolumn{1}{c|}{\textit{n} = 25}  & \multicolumn{1}{c|}{\textit{n} = 50}  &\multicolumn{1}{c|}{\textit{n} = 100}  \\ \hline
\multicolumn{1}{|c|}{} & \multicolumn{1}{c|}{$\Delta=25$} & \multicolumn{1}{c|}{1.53 (0.14)} & \multicolumn{1}{c|}{2.89 (0.11)} & \multicolumn{1}{c|}{4.20 (0.13)} \\ \cline{2-5} 
\multicolumn{1}{|c|}{$\Gamma=2$}                    & \multicolumn{1}{c|}{$\Delta=50$} & \multicolumn{1}{c|}{2.69 (0.12)} & \multicolumn{1}{c|}{4.18 (0.15)} & \multicolumn{1}{c|}{7.34 (0.24)} \\ \cline{2-5} 
\multicolumn{1}{|c|}{}                    & \multicolumn{1}{c|}{$\Delta=100$} & \multicolumn{1}{c|}{4.89 (0.25)} & \multicolumn{1}{c|}{9.12 (0.26)} & \multicolumn{1}{c|}{13.89 (0.83)} \\ \hline
\multicolumn{1}{|c|}{} & \multicolumn{1}{c|}{$\Delta=25$} & \multicolumn{1}{c|}{2.07 (0.17)} & \multicolumn{1}{c|}{3.55 (0.15)} & \multicolumn{1}{c|}{6.16 (0.18)} \\ \cline{2-5} 
\multicolumn{1}{|c|}{$\Gamma=4$}                    & \multicolumn{1}{c|}{$\Delta=50$} & \multicolumn{1}{c|}{3.52 (0.21)} & \multicolumn{1}{c|}{6.31 (0.13)} &\multicolumn{1}{c|}{11.19 (0.28)} \\ \cline{2-5} 
\multicolumn{1}{|c|}{}                    & \multicolumn{1}{c|}{$\Delta=100$} & \multicolumn{1}{c|}{6.97 (0.22)} & \multicolumn{1}{c|}{12.60 (0.67)} & \multicolumn{1}{c|}{21.62 (1.39)} \\ \hline
\multicolumn{1}{|c|}{} & \multicolumn{1}{c|}{$\Delta=25$} & \multicolumn{1}{c|}{3.12 (0.13)} & \multicolumn{1}{c|}{5.64 (0.19)} & \multicolumn{1}{c|}{10.22 (0.22)} \\ \cline{2-5} 
\multicolumn{1}{|c|}{$\Gamma=8$}                    & \multicolumn{1}{c|}{$\Delta=50$} & \multicolumn{1}{c|}{5.64 (0.17)} & \multicolumn{1}{c|}{10.46 (0.16)} &  \multicolumn{1}{c|}{19.36 (0.48)} \\ \cline{2-5} 
\multicolumn{1}{|c|}{}                    & \multicolumn{1}{c|}{$\Delta=100$} & \multicolumn{1}{c|}{11.05 (0.17)} & \multicolumn{1}{c|}{20.42 (0.31)} & \multicolumn{1}{c|}{37.85 (1.87)} \\ \hline
\end{tabular}
\end{table}

\begin{table}[H]
\centering
\caption{P-value of Algorithm \ref{algorithm: testing hypotheses on mean functions; Appendix} and \ref{algorithm: permutation-based testing hypotheses on mean functions} based on $\Gamma, \Delta$, and $n$. ($\varepsilon = 0.05$)}\label{table: pvalue of Algorithm 1(0.05)}
    \vspace*{0.5em}
\small
\def\arraystretch{0.8}
\begin{tabular}{|ccccc|}
\hline
\multicolumn{5}{|c|}{Algorithm \ref{algorithm: testing hypotheses on mean functions; Appendix}}\\
\hline
\multicolumn{1}{|c|}{Number of directions}                   & \multicolumn{1}{c|}{Number of levels}  & \multicolumn{1}{c|}{\textit{n} = 25}  & \multicolumn{1}{c|}{\textit{n} = 50}  &\multicolumn{1}{c|}{\textit{n} = 100}  \\ \hline
\multicolumn{1}{|c|}{} & \multicolumn{1}{c|}{$\Delta=25$} & \multicolumn{1}{c|}{0.3487 (0.2618)} & \multicolumn{1}{c|}{0.1804 (0.2659)} & \multicolumn{1}{c|}{0.0612 (0.1072)} \\ \cline{2-5} 
\multicolumn{1}{|c|}{$\Gamma=2$}                    & \multicolumn{1}{c|}{$\Delta=50$} & \multicolumn{1}{c|}{0.1696 (0.2205)} & \multicolumn{1}{c|}{0.1502 (0.2350)} & \multicolumn{1}{c|}{0.0179 (0.0357)} \\ \cline{2-5} 
\multicolumn{1}{|c|}{}                    & \multicolumn{1}{c|}{$\Delta=100$} & \multicolumn{1}{c|}{0.1224 (0.1383)} & \multicolumn{1}{c|}{0.0565 (0.0948)} & \multicolumn{1}{c|}{0.0123 (0.0251)} \\ \hline
\multicolumn{1}{|c|}{} & \multicolumn{1}{c|}{$\Delta=25$} & \multicolumn{1}{c|}{0.1909 (0.1964)} & \multicolumn{1}{c|}{0.1367 (0.1626)} & \multicolumn{1}{c|}{0.0656 (0.0933)} \\ \cline{2-5} 
\multicolumn{1}{|c|}{$\Gamma=4$}                    & \multicolumn{1}{c|}{$\Delta=50$} & \multicolumn{1}{c|}{0.1414 (0.1693)} & \multicolumn{1}{c|}{0.1250 (0.1861)} &\multicolumn{1}{c|}{0.0070 (0.0119)} \\ \cline{2-5} 
\multicolumn{1}{|c|}{}                    & \multicolumn{1}{c|}{$\Delta=100$} & \multicolumn{1}{c|}{0.1747 (0.2432)} & \multicolumn{1}{c|}{0.0517 (0.0703)} & \multicolumn{1}{c|}{0.0056 (0.0147)} \\ \hline
\multicolumn{1}{|c|}{} & \multicolumn{1}{c|}{$\Delta=25$} & \multicolumn{1}{c|}{0.2773 (0.2425)} & \multicolumn{1}{c|}{0.1130 (0.1770)} & \multicolumn{1}{c|}{0.0294 (0.0561)} \\ \cline{2-5} 
\multicolumn{1}{|c|}{$\Gamma=8$}                    & \multicolumn{1}{c|}{$\Delta=50$} & \multicolumn{1}{c|}{0.1064 (0.1140)} & \multicolumn{1}{c|}{0.0533 (0.0929)} &  \multicolumn{1}{c|}{0.0179 (0.0235)} \\ \cline{2-5} 
\multicolumn{1}{|c|}{}                    & \multicolumn{1}{c|}{$\Delta=100$} & \multicolumn{1}{c|}{0.1371 (0.2053)} & \multicolumn{1}{c|}{0.0296 (0.0694)} & \multicolumn{1}{c|}{0.0044 (0.0125)} \\ \hline
\multicolumn{5}{|c|}{Algorithm \ref{algorithm: permutation-based testing hypotheses on mean functions}}\\
\hline
\multicolumn{1}{|c|}{Number of directions}                   & \multicolumn{1}{c|}{Number of levels}  & \multicolumn{1}{c|}{\textit{n} = 25}  & \multicolumn{1}{c|}{\textit{n} = 50}  &\multicolumn{1}{c|}{\textit{n} = 100}  \\ \hline
\multicolumn{1}{|c|}{} & \multicolumn{1}{c|}{$\Delta=25$} & \multicolumn{1}{c|}{0.3582 (0.2710)} & \multicolumn{1}{c|}{0.1942 (0.1807)} & \multicolumn{1}{c|}{0.0600 (0.1052)} \\ \cline{2-5} 
\multicolumn{1}{|c|}{$\Gamma=2$}                    & \multicolumn{1}{c|}{$\Delta=50$} & \multicolumn{1}{c|}{0.2156 (0.1863)} & \multicolumn{1}{c|}{0.1704 (0.2030)} & \multicolumn{1}{c|}{0.0385 (0.0950)} \\ \cline{2-5} 
\multicolumn{1}{|c|}{}                    & \multicolumn{1}{c|}{$\Delta=100$} & \multicolumn{1}{c|}{0.2067 (0.1984)} & \multicolumn{1}{c|}{0.0883 (0.1186)} & \multicolumn{1}{c|}{0.0601 (0.2046)} \\ \hline
\multicolumn{1}{|c|}{} & \multicolumn{1}{c|}{$\Delta=25$} & \multicolumn{1}{c|}{0.3102 (0.2727)} & \multicolumn{1}{c|}{0.3664 (0.2714)} & \multicolumn{1}{c|}{0.1970 (0.2644)} \\ \cline{2-5} 
\multicolumn{1}{|c|}{$\Gamma=4$}                    & \multicolumn{1}{c|}{$\Delta=50$} & \multicolumn{1}{c|}{0.3282 (0.3152)} & \multicolumn{1}{c|}{0.1445 (0.1841)} &\multicolumn{1}{c|}{0.0364 (0.1141)} \\ \cline{2-5} 
\multicolumn{1}{|c|}{}                    & \multicolumn{1}{c|}{$\Delta=100$} & \multicolumn{1}{c|}{0.3250 (0.3197)} & \multicolumn{1}{c|}{0.0728 (0.1225)} & \multicolumn{1}{c|}{0.0194 (0.0676)} \\ \hline
\multicolumn{1}{|c|}{} & \multicolumn{1}{c|}{$\Delta=25$} & \multicolumn{1}{c|}{0.3704 (0.3408)} & \multicolumn{1}{c|}{0.2533 (0.2406)} & \multicolumn{1}{c|}{0.2138 (0.2670)} \\ \cline{2-5} 
\multicolumn{1}{|c|}{$\Gamma=8$}                    & \multicolumn{1}{c|}{$\Delta=50$} & \multicolumn{1}{c|}{0.2108 (0.2354)} & \multicolumn{1}{c|}{0.1437 (0.2147)} &  \multicolumn{1}{c|}{0.0373 (0.0665)} \\ \cline{2-5} 
\multicolumn{1}{|c|}{}                    & \multicolumn{1}{c|}{$\Delta=100$} & \multicolumn{1}{c|}{0.2107 (0.2843)} & \multicolumn{1}{c|}{0.1613 (0.2141)} & \multicolumn{1}{c|}{0.0186 (0.0567)} \\ \hline
\end{tabular}
\end{table}

\begin{table}[H]
\centering
\caption{P-value of Algorithm \ref{algorithm: testing hypotheses on mean functions; Appendix} and \ref{algorithm: permutation-based testing hypotheses on mean functions} based on $\Gamma, \Delta$, and $n$. ($\varepsilon = 0$)}\label{table: pvalue of Algorithm 1(0)}
    \vspace*{0.5em}
\small
\def\arraystretch{0.8}
\begin{tabular}{|ccccc|}
\hline
\multicolumn{5}{|c|}{Algorithm \ref{algorithm: testing hypotheses on mean functions; Appendix}}\\
\hline
\multicolumn{1}{|c|}{Number of directions}                   & \multicolumn{1}{c|}{Number of levels}  & \multicolumn{1}{c|}{\textit{n} = 25}  & \multicolumn{1}{c|}{\textit{n} = 50}  &\multicolumn{1}{c|}{\textit{n} = 100}  \\ \hline
\multicolumn{1}{|c|}{} & \multicolumn{1}{c|}{$\Delta=25$} & \multicolumn{1}{c|}{0.3826 (0.2486)} & \multicolumn{1}{c|}{0.3712 (0.2770)} & \multicolumn{1}{c|}{0.4374 (0.2333)} \\ \cline{2-5} 
\multicolumn{1}{|c|}{$\Gamma=2$}                    & \multicolumn{1}{c|}{$\Delta=50$} & \multicolumn{1}{c|}{0.3943 (0.2425)} & \multicolumn{1}{c|}{0.3725 (0.2448)} & \multicolumn{1}{c|}{0.3234 (0.2493)} \\ \cline{2-5} 
\multicolumn{1}{|c|}{}                    & \multicolumn{1}{c|}{$\Delta=100$} & \multicolumn{1}{c|}{0.3209 (0.2122)} & \multicolumn{1}{c|}{0.3947 (0.2432)} & \multicolumn{1}{c|}{0.5072 (0.3393)} \\ \hline
\multicolumn{1}{|c|}{} & \multicolumn{1}{c|}{$\Delta=25$} & \multicolumn{1}{c|}{0.2513 (0.2415)} & \multicolumn{1}{c|}{0.4025 (0.2906)} & \multicolumn{1}{c|}{0.2910 (0.2347)} \\ \cline{2-5} 
\multicolumn{1}{|c|}{$\Gamma=4$}                    & \multicolumn{1}{c|}{$\Delta=50$} & \multicolumn{1}{c|}{0.3060 (0.2171)} & \multicolumn{1}{c|}{0.2694 (0.2351)} &\multicolumn{1}{c|}{0.4245 (0.2508)} \\ \cline{2-5} 
\multicolumn{1}{|c|}{}                    & \multicolumn{1}{c|}{$\Delta=100$} & \multicolumn{1}{c|}{0.2855 (0.2420)} & \multicolumn{1}{c|}{0.2357 (0.2230)} & \multicolumn{1}{c|}{0.3426 (0.2939)} \\ \hline
\multicolumn{1}{|c|}{} & \multicolumn{1}{c|}{$\Delta=25$} & \multicolumn{1}{c|}{0.2431 (0.2237)} & \multicolumn{1}{c|}{0.2711 (0.2784)} & \multicolumn{1}{c|}{0.3265 (0.2509)} \\ \cline{2-5} 
\multicolumn{1}{|c|}{$\Gamma=8$}                    & \multicolumn{1}{c|}{$\Delta=50$} & \multicolumn{1}{c|}{0.3199 (0.2110)} & \multicolumn{1}{c|}{0.2430 (0.1694)} &  \multicolumn{1}{c|}{0.3367 (0.2481)} \\ \cline{2-5} 
\multicolumn{1}{|c|}{}                    & \multicolumn{1}{c|}{$\Delta=100$} & \multicolumn{1}{c|}{0.2507 (0.2215)} & \multicolumn{1}{c|}{0.3212 (0.2307)} & \multicolumn{1}{c|}{0.4346 (0.2546)} \\ \hline
\multicolumn{5}{|c|}{Algorithm \ref{algorithm: permutation-based testing hypotheses on mean functions}}\\
\hline
\multicolumn{1}{|c|}{Number of directions}                   & \multicolumn{1}{c|}{Number of levels}  & \multicolumn{1}{c|}{\textit{n} = 25}  & \multicolumn{1}{c|}{\textit{n} = 50}  &\multicolumn{1}{c|}{\textit{n} = 100}  \\ \hline
\multicolumn{1}{|c|}{} & \multicolumn{1}{c|}{$\Delta=25$} & \multicolumn{1}{c|}{0.4278 (0.2947)} & \multicolumn{1}{c|}{0.4764 (0.2830)} & \multicolumn{1}{c|}{0.4921 (0.2359)} \\ \cline{2-5} 
\multicolumn{1}{|c|}{$\Gamma=2$}                    & \multicolumn{1}{c|}{$\Delta=50$} & \multicolumn{1}{c|}{0.5637 (0.3072)} & \multicolumn{1}{c|}{0.5866 (0.2655)} & \multicolumn{1}{c|}{0.5510 (0.3005)} \\ \cline{2-5} 
\multicolumn{1}{|c|}{}                    & \multicolumn{1}{c|}{$\Delta=100$} & \multicolumn{1}{c|}{0.3028 (0.2026)} & \multicolumn{1}{c|}{0.5540 (0.2932)} & \multicolumn{1}{c|}{0.4155 (0.2892)} \\ \hline
\multicolumn{1}{|c|}{} & \multicolumn{1}{c|}{$\Delta=25$} & \multicolumn{1}{c|}{0.5738 (0.3467)} & \multicolumn{1}{c|}{0.5293 (0.2335)} & \multicolumn{1}{c|}{0.4714 (0.2967)} \\ \cline{2-5} 
\multicolumn{1}{|c|}{$\Gamma=4$}                    & \multicolumn{1}{c|}{$\Delta=50$} & \multicolumn{1}{c|}{0.4489 (0.2582)} & \multicolumn{1}{c|}{0.4639 (0.3050)} &\multicolumn{1}{c|}{0.4669 (0.3310)} \\ \cline{2-5} 
\multicolumn{1}{|c|}{}                    & \multicolumn{1}{c|}{$\Delta=100$} & \multicolumn{1}{c|}{0.5193 (0.2830)} & \multicolumn{1}{c|}{0.5579 (0.2984)} & \multicolumn{1}{c|}{0.5557 (0.2215)} \\ \hline
\multicolumn{1}{|c|}{} & \multicolumn{1}{c|}{$\Delta=25$} & \multicolumn{1}{c|}{0.4786 (0.2669)} & \multicolumn{1}{c|}{0.4782 (0.3068)} & \multicolumn{1}{c|}{0.4721 (0.3296)} \\ \cline{2-5} 
\multicolumn{1}{|c|}{$\Gamma=8$}                    & \multicolumn{1}{c|}{$\Delta=50$} & \multicolumn{1}{c|}{0.3867 (0.2768)} & \multicolumn{1}{c|}{0.4543 (0.3116)} &  \multicolumn{1}{c|}{0.5324 (0.2943)} \\ \cline{2-5} 
\multicolumn{1}{|c|}{}                    & \multicolumn{1}{c|}{$\Delta=100$} & \multicolumn{1}{c|}{0.6138 (0.2375)} & \multicolumn{1}{c|}{0.5110 (0.2581)} & \multicolumn{1}{c|}{0.4109 (0.2592)} \\ \hline
\end{tabular}
\end{table}

\section{Proofs}\label{section: appendix, proofs}

\subsection{Proof of Theorem \ref{thm: the separability of C(Shere;H)}}

\begin{proof}[Proof of Theorem \ref{thm: the separability of C(Shere;H)}]
Since $H_0^1([0,T])$ is a separable Hilbert space \citep[][Section 8.3]{brezis2011functional}, it suffices to show the results (i) and (ii). 

The separability of $\mathcal{H}$ implies that $\mathcal{H}$ has an orthonormal basis $\{\boldsymbol{e}_j\}_{j=1}^\infty$ \citep[][Theorem 5.11]{brezis2011functional}. Since $C(\mathbb{S}^{d-1})=C(\mathbb{S}^{d-1};\mathbb{R})$ is separable \citep[][Section 3.6]{brezis2011functional}, $C(\mathbb{S}^{d-1})$ has a dense and countable subset $D$. Then, the linear hull $\Tilde{D}\overset{\operatorname{def}}{=}\operatorname{span}\{g\boldsymbol{e}_j\,\vert\, g\in D \mbox{ and } j=1,2,\cdots\}$ is a dense and countable subset of $C(\mathbb{S}^{d-1};\mathcal{H})$, and the reasoning is the following.

For any $f\in C(\mathbb{S}^{d-1};\mathcal{H})$, we have
\begin{align*}
f(\nu)=\sum_{j=1}^\infty \langle f(\nu), \boldsymbol{e}_j\rangle \boldsymbol{e}_j,\ \ \mbox{ for each }\nu\in\mathbb{S}^{d-1},
\end{align*}
where $\langle \cdot, \cdot \rangle$ denotes the inner product of $\mathcal{H}$ and $\sum_{j=1}^\infty$ converges in the $\mathcal{H}$-topology. It is straightforward that the function $\nu\mapsto \langle f(\nu), \boldsymbol{e}_j\rangle$ is an element of $C(\mathbb{S}^{d-1})$, for each fixed $j=1,2,\cdots$. Hence, for any $\epsilon>0$, there exists $\{g_j\}_{j=1}^\infty\subseteq D$ such that 
\begin{align*}
\sup_{\nu\in\mathbb{S}^{d-1}}\left\vert \langle f(\nu), \boldsymbol{e}_j\rangle -g_j(\nu) \right\vert<\frac{\epsilon}{2^{j+1}},
\end{align*}
for all $j=1,2,\cdots$, which implies
\begin{align*}
\left\Vert f - \sum_{j=1}^\infty g_j \boldsymbol{e}_j \right\Vert_{C(\mathbb{S}^{d-1};\mathcal{H})} &= \sup_{\nu\in\mathbb{S}^{d-1}} \left\Vert f(\nu) - \sum_{j=1}^\infty g_j(\nu) \boldsymbol{e}_j \right\Vert_{\mathcal{H}} \\
& = \sup_{\nu\in\mathbb{S}^{d-1}} \left\Vert \sum_{j=1}^\infty \Big(\langle f(\nu), \boldsymbol{e}_j\rangle-g_j(\nu)\Big) \boldsymbol{e}_j \right\Vert_{\mathcal{H}}\\
&\le \sup_{\nu\in\mathbb{S}^{d-1}} \sum_{j=1}^\infty \left\vert \langle f(\nu), \boldsymbol{e}_j\rangle -g_j(\nu) \right\vert <\epsilon.
\end{align*}
Since $\{\sum_{j=1}^n g_j \boldsymbol{e}_j\}_{n=1}^\infty \subseteq \Tilde{D}$, the proof of result (i) is complete. 

The result (ii) can be proved using the same trick implemented in the proof of result (i). The proof is complete.
\end{proof}

\subsection{An Elementary Proof of Eq.~\eqref{eq: Sobolev embedding from Morrey}}\label{proof: simple proof of the Sobolev embedding}

\begin{proof}[Proof of Eq.~\eqref{eq: Sobolev embedding from Morrey}]
For any $f\in\mathcal{H}  =  H_0^1([0,T]) = \{f\in L^2([0,1])\,\vert\, f'\in L^2([0,T]) \mbox{ and }f(0)=f(T)=0\}$, we identify $f$ as a continuous function defined on the compact interval $[0,T]$ (see Section \ref{section: notation for closed vs. open}, also Theorem 8.2 of \cite{brezis2011functional}, for a justification). It suffices to show $\Vert f\Vert_{C^{0,\frac{1}{2}}([0,T])} \le \Tilde{C}_T \cdot \Vert f\Vert_{\mathcal{H}}$ for some constant $\Tilde{C}_T$ depending only on $T$.

Theorem 8.2 of \cite{brezis2011functional} implies
\begin{align}\label{eq: NL formula for weak derivatives}
    f(t)-f(s) = \int_s^t f'(\tau) \,d\tau, \ \ \text{ for all }s,t\in[0,T],
\end{align}
where $f'$ denotes the weak derivative of $f$. Without loss of generality, we assume $s<t$ in Eq.~\eqref{eq: NL formula for weak derivatives}. $f(0)=0$ and Eq.~\eqref{eq: NL formula for weak derivatives} imply the following inequalities for all $t\in[0,T]$
\begin{align*}
    \vert f(t) \vert & = \left\vert \int_0^t f'(\tau) \,d\tau\right\vert \\
    & \le \int_0^T \left\vert f'(\tau) \right\vert \,d\tau \\
    & \le \sqrt{T}\cdot\left( \int_0^T \left\vert f'(\tau) \right\vert^2 \,d\tau \right)^{1/2} \\
    & = \sqrt{T}\cdot\Vert f\Vert_{\mathcal{H}}.
\end{align*}
Hence, $\sup_{t\in[0,T]}\vert f(t) \vert \le \sqrt{T}\cdot\Vert f\Vert_{\mathcal{H}}$. Again, Eq.~\eqref{eq: NL formula for weak derivatives} implies 
\begin{align*}
    \vert f(t) - f(s) \vert & = \left\vert \int_s^t f'(\tau) \,d\tau\right\vert \\
    & \le \sqrt{t-s}\cdot\left( \int_s^t \left\vert f'(\tau) \right\vert^2 \,d\tau \right)^{1/2} \\
    & \le \sqrt{t-s}\cdot\left( \int_0^T \left\vert f'(\tau) \right\vert^2 \,d\tau \right)^{1/2} \\
    & = \sqrt{t-s}\cdot \Vert f\Vert_{\mathcal{H}}.
\end{align*}
Therefore, we have
\begin{align*}
    \sup_{s,t\in[0,T] \text{ and } s\ne t} \left(\frac{\vert f(t) - f(s) \vert}{\sqrt{\vert t-s \vert}}\right) \le \Vert f\Vert_{\mathcal{H}},
\end{align*}
which implies
\begin{align*}
    \Vert f\Vert_{C^{0,\frac{1}{2}}([0,T])} = \sup_{t\in[0,T]}\vert f(t) \vert + \sup_{s,t\in[0,T] \text{ and } s\ne t} \left(\frac{\vert f(t) - f(s) \vert}{\sqrt{\vert t-s \vert}}\right) \le (1+\sqrt{T}) \cdot \Vert f\Vert_{\mathcal{H}}.
\end{align*}
The desired Eq.~\eqref{eq: Sobolev embedding from Morrey} follows (i.e., $\Tilde{C}_T=1+\sqrt{T}$).
\end{proof}

\subsection{Proof of Theorem \ref{thm: boundedness topological invariants theorem}}

We recall that the discontinuities of $t\mapsto \beta_k(K_t^\nu)$ and $t\mapsto\chi(K_t^\nu)=\chi_t^\nu(K)$ are the HCPs of $K$ in direction $\nu$ (see Section \ref{section: Homeomorphism Critical Points}).

\begin{proof}[Proof of Theorem \ref{thm: boundedness topological invariants theorem}]
    For any fixed $\nu\in\mathbb{S}^{d-1}$, the following inclusion is straightforward
\begin{align}\label{eq: PD inclusion}
\Big(\operatorname{Dgm}_k(K;\phi_{\nu})\cap (-\infty,t)\times(t,\infty)\Big) \subseteq \left\{\xi\in \operatorname{Dgm}_k(K;\phi_{\nu}) \, \vert \, \operatorname{pers}(\xi)>0\right\},
\end{align}
where the function $\phi_\nu$ is defined in Eq.~\eqref{Eq: Morse function 1}, and the definitions of $\operatorname{Dgm}_k(K;\phi_{\nu})$ and $\operatorname{pers}(\xi)$ are given in Appendix \ref{The Relationship between PHT and SECT}. Together with the k-triangle lemma \citep{edelsbrunner2000topological, cohen2007stability}, the inclusion in Eq.~\eqref{eq: PD inclusion} and Condition \ref{condition: the condition for defining S_{R,d}^M} imply
\begin{align}\label{eq: boundedness of Betti numbers}
    \begin{aligned}
        \beta_k(K_t^\nu) &=\# \Big(\operatorname{Dgm}_k(K;\phi_{\nu})\cap (-\infty,t)\times(t,\infty)\Big) \\
    & \le \# \{\xi\in \operatorname{Dgm}_k(K;\phi_{\nu}) \, \vert \, \operatorname{pers}(\xi)>0\} \\
    & \le\frac{M}{d},
    \end{aligned}
\end{align}
for all $k\in\{0,1,\cdots,d-1\}$ and all $t$ that are not HCPs in direction $\nu$, where the cardinality $\#\{\cdot\}$ counts the multiplicity of the multisets. Eq.~\eqref{Eq: first def of Euler characteristic curve} implies
\begin{align}\label{eq: Chi bounded by Betti}
    \begin{aligned}
    \left\vert\chi_{t}^\nu(K)\right\vert & = \left\vert\sum_{k=0}^{d-1} (-1)^{k}\cdot\beta_k(K_t^{\nu})\right\vert \\ 
    & \le d\cdot\sup_{k\in\{0,\cdots,d-1\}}\beta_k(K_t^{\nu}) \\ 
    & \le M,
    \end{aligned}
\end{align}
for all $t$ that are not HCPs in direction $\nu$. The right continuity of $t\mapsto\chi(K_t^\nu)$ stated in Lemma \ref{lemma: right continuity of the ECT}, together with Eq.~\eqref{eq: Chi bounded by Betti}, implies that $\left\vert\chi_{t}^\nu(K)\right\vert \le M$ holds for all $t\in[0,T]$. Then, we have
\begin{align*}
    \sup_{\nu\in\mathbb{S}^{d-1}}\left(\sup_{0\le t\le T}\left\vert\chi_{t}^\nu(K)\right\vert\right) \le M.
\end{align*}
The proof is complete. 
\end{proof}

\subsection{Proof of Lemma \ref{thm: Sobolev function paths; general, appendix}}

To prove Lemma \ref{thm: Sobolev function paths; general, appendix}, we need the following lemma as a preparation.
\begin{lemma}\label{lemma: weak derivative formula}
    For any $K\in\mathcal{S}_{R,d}^M$ and fixed $\nu\in\mathbb{S}^{d-1}$, the function $t\mapsto \int_0^t\chi_\tau^\nu(K) \,d\tau$ has its first-order weak derivative $t\mapsto\chi_t^\nu(K)$.
\end{lemma}
\begin{proof}
Because of 
\begin{align*}
    \left\{\int_0^t \chi_\tau^\nu(K) d\tau\right\}_{t\in[0,T]} &\in\{\mbox{all absolutely continuous functions on }[0,T]\}\\
    &=\{x\in L^1([0,T]): \mbox{the weak derivative $x'$ exists and }x'\in L^1([0,T])\}\\
    & \overset{\operatorname{def}}{=} W^{1,1}([0,T])
\end{align*}
(see the Remark 8 after Proposition 8.3 in \cite{brezis2011functional} for details), the weak derivative of $\{\int_0^t \chi_\tau^\nu(K) d\tau\}_{t\in[0,T]}$ exists. Lemma \ref{lemma: weak derivative formula} follows from Theorem 8.2 of \cite{brezis2011functional}.
\end{proof}
\noindent\textbf{Remark:} Using Lemma \ref{thm: tameness property}, one can verify that $\chi_t^\nu(K)$ is the classical derivative of $\int_0^t\chi_\tau^\nu(K) \,d\tau$ for all $t$ except for the finitely many HCPs of $K$ in direction $\nu$.

With Lemma \ref{lemma: weak derivative formula}, we prove Lemma \ref{thm: Sobolev function paths; general, appendix} as follows.

\begin{proof}[Proof of Lemma \ref{thm: Sobolev function paths; general, appendix}
.] 
For the simplicity, we denote
\begin{align*}
    F(t) \overset{\operatorname{def}}{=} \int_0^t\chi_\tau^\nu(K) d\tau - \frac{t}{T} \int_0^T \chi_\tau^\nu(K) d\tau, \ \ \mbox{ for }t\in[0,T].
\end{align*}
Theorem \ref{thm: boundedness topological invariants theorem} implies 
\begin{align*}
    \vert F(t)\vert \le \int_0^T \vert\chi_\tau^\nu(K) \vert d\tau + \frac{t}{T} \int_0^T \vert\chi_\tau^\nu(K)\vert d\tau \le 2TM,\ \ \mbox{ for }t\in[0,T].
\end{align*}
Hence, $F\in L^p([0,T])$ for $p\in[1,\infty)$. Lemma \ref{lemma: weak derivative formula} implies that the weak derivative of $F$ exists and is $F'(t)=\chi_t^\nu(K)-\frac{1}{T}\int_0^T \vert\chi_\tau^\nu(K)\vert d\tau$. We have the boundedness
\begin{align*}
    \vert F'(t)\vert \le \vert \chi_{t}^\nu(K) \vert + \frac{1}{T} \int_{0}^T \vert \chi_\tau^\nu(K)\vert d\tau \le 2M, \ \ \mbox{ for }t\in[0,T],
\end{align*}
which implies $F'\in L^p([0,T])$ for $p\in[1,\infty)$. Furthermore, $F(0)=F(T)=0$, together with the discussion above, implies $F\in W^{1,p}_0([0,T])$ for all $p\in[1,\infty)$ \citep[][Theorem 8.12]{brezis2011functional}. Theorem 8.8 and the Remark 8 after Proposition 8.3 in \cite{brezis2011functional} imply $W^{1,p}_0([0,T]) \subseteq \mathcal{B}$ for $p\in[1,\infty)$. The proof of Lemma \ref{thm: Sobolev function paths; general, appendix} is complete.   
\end{proof}

\subsection{Proof of Eq.~\eqref{Eq: lemma for the continuity inequality}}

This subsection gives the proof of the first half of Theorem \ref{lemma: The continuity lemma; Appendix}, i.e., Eq.~\eqref{Eq: lemma for the continuity inequality}. The following lemmas are prepared for the proof of Eq.~\eqref{Eq: lemma for the continuity inequality}.

\begin{lemma}\label{lemma: stability lemma 1}
Suppose $K\in\mathcal{S}_{R,d}^M$. We have the following estimate for all $t$ that are neither HCPs of $K$ in direction $\nu_1$ nor HCPs of $K$ in direction $\nu_2$.
\begin{align}\label{Eq: counting estimate}
\begin{aligned}
    & \Upsilon_k(t;\nu_1, \nu_2) \overset{\operatorname{def}}{=} \left\vert \beta_k(K_t^{\nu_1}) - \beta_k(K_t^{\nu_2}) \right\vert  \\ 
    & \le \#\left\{x\in \operatorname{Dgm}_k(K;\phi_{\nu_1}) \, \Big\vert \, x\ne\gamma^*(x) \mbox{ and } \underline{(x,\gamma^*(x))}\bigcap\partial\big((-\infty,t)\times(t,\infty)\big)\ne\emptyset\right\},
    \end{aligned}
\end{align}
where $\underline{(x,\gamma^*(x))}$ denotes the straight line segment connecting points $x$ and $\gamma^*(x)$ in $\mathbb{R}^2$, the map $\gamma^*$ is any optimal bijection such that 
\begin{align}\label{eq: optimal bijection condition}
    W_\infty \Big(\operatorname{Dgm}_k(K;\phi_{\nu_1}),\, \operatorname{Dgm}_k(K;\phi_{\nu_2}) \Big) = \sup \Big\{\Vert \xi - \gamma^*(\xi) \Vert_{l^\infty} \, \Big\vert \, \xi\in \operatorname{Dgm}_k(K;\phi_{\nu_1})\Big\}
\end{align}
(see Definition \ref{def: bottleneck distance}, and $\Vert\cdot\Vert_{l^\infty}$ is defined in Eq.~\eqref{eq: def of l infinity norm}), and the cardinality $\#$ counts the corresponding multiplicity.
\end{lemma}
\begin{remark} 
Because $(\mathscr{D}, W_\infty)$ is a geodesic space, the optimal bijection $\gamma^*$ does exist \citep[][Proposition 1 and its proof]{turner2013means}.
\end{remark}
\begin{proof}[Proof of Lemma \ref{lemma: stability lemma 1}]
Since $t$ is not an HCP, neither $\operatorname{Dgm}_k(K;\phi_{\nu_1})$ nor $\operatorname{Dgm}_k(K;\phi_{\nu_2})$ has a point on the boundary $\partial\big((-\infty,t)\times(t,\infty)\big)$. If $\beta_k(K_t^{\nu_1}) = \beta_k(K_t^{\nu_2})$, Eq.~\eqref{Eq: counting estimate} is true. Otherwise, without loss of generality, we assume $\beta_k(K_t^{\nu_1}) > \beta_k(K_t^{\nu_2})$. Notice
\begin{align*}
    \beta_k(K_t^{\nu_i})=\#\left\{\operatorname{Dgm}_k(K;\phi_{\nu_i})\bigcap (-\infty,t)\times(t,\infty) \right\},\ \ \mbox{ for }i\in\{1,2\}.
\end{align*}
Let $\gamma^*$ be any optimal bijection, then there should be at least $\beta_k(K_t^{\nu_1}) - \beta_k(K_t^{\nu_2})$ straight line segments $\underline{(x,\gamma^*(x))}$ crossing $\partial\big((-\infty,t)\times(t,\infty)\big)$; otherwise, $\gamma^*$ is not bijective. Hence, 
\begin{align*}
    &\beta_k(K_t^{\nu_1}) - \beta_k(K_t^{\nu_2}) \\ 
    & \le \#\left\{x\in \operatorname{Dgm}_k(K;\phi_{\nu_1})\,\Big\vert\, x\ne\gamma^*(x) \mbox{ and } \underline{(x,\gamma^*(x))}\bigcap\partial\big((-\infty,t)\times(t,\infty)\big)\ne\emptyset\right\},
\end{align*}
and Eq.~\eqref{Eq: counting estimate} follows.
\end{proof}

\begin{lemma}\label{lemma: stability lemma 2}
Suppose $K\in\mathcal{S}_{R,d}^M$. Except for finitely many $t$, we have 
\begin{align*}
    \Upsilon_k(t;\nu_1, \nu_2) \le \frac{2M}{d} \cdot \mathbbm{1}_{\mathcal{T}_k},\ \ \mbox{where}
\end{align*}
\begin{align*}
    \mathcal{T}_k \overset{\operatorname{def}}{=} \left\{t\in[0,T] \mbox{ not an HCP in direction }\nu_1\mbox{ or }\nu_2 \,\Big\vert\, \mbox{there exists } x\in \operatorname{Dgm}_k(K;\phi_{\nu_1}) \mbox{ such that } x\ne\gamma^*(x) \right.
    \\ \left. \mbox{ and } \underline{(x,\gamma^*(x))}\bigcap\partial\big((-\infty,t)\times(t,\infty)\big)\ne\emptyset\right\},
\end{align*}
and $\gamma^*: \operatorname{Dgm}_k(K;\phi_{\nu_1})\rightarrow \operatorname{Dgm}_k(K;\phi_{\nu_2})$ is any optimal bijection satisfying Eq.~\eqref{eq: optimal bijection condition}.
\end{lemma}
\begin{proof}[Proof of Lemma \ref{lemma: stability lemma 2}]
Eq.~\eqref{eq: boundedness of Betti numbers} implies
\begin{align*}
    \Upsilon_k(t;\nu_1, \nu_2) = \left\vert \beta_k(K_t^{\nu_1}) - \beta_k(K_t^{\nu_2}) \right\vert \le 2M/d.
\end{align*}
Furthermore, the inequality in Eq.~\eqref{Eq: counting estimate} indicates that $\Upsilon_k(t;\nu_1, \nu_2)=0$ if $t\notin\mathcal{T}_k$, except for finitely many HCPs in directions $\nu_1$ and $\nu_2$. Then the desired estimate follows.
\end{proof}

\begin{proof}[Proof of Eq.~\eqref{Eq: lemma for the continuity inequality}.] 
Eq.~\eqref{Eq: first def of Euler characteristic curve} and Lemma \ref{lemma: stability lemma 2} imply the following for $p\in[1,\infty)$
\begin{align}\label{Eq: estimate in proof}
    \begin{aligned}
        & \int_0^T \left\vert\Big\{\chi_\tau^{\nu_1}(K)-\chi_\tau^{\nu_2}(K)\Big\} \right\vert^p d\tau \\
        & = \int_0^T \left\vert \sum_{k=0}^{d-1} (-1)^k\cdot\Big(\beta_k(K_\tau^{\nu_1}) - \beta_k(K_\tau^{\nu_2}) \Big) \right\vert^p d\tau \\
    & \le \int_0^T \left(\sum_{k=0}^{d-1}\Upsilon_k(\tau;\nu_1,\nu_2)\right)^p d\tau \\
    & \le d^{(p-1)} \cdot \sum_{k=0}^{d-1} \int_0^T \Big(\Upsilon_k(\tau;\nu_1,\nu_2)\Big)^p d\tau \\
    & \le \frac{(2M)^p}{d} \cdot \sum_{k=0}^{d-1} \int_{\mathcal{T}_k} d\tau \\
    & \le \frac{(2M)^p}{d} \cdot \sum_{k=0}^{d-1} \left(\sum_{\xi\in \operatorname{Dgm}_k(K;\phi_{\nu_1})} 2\cdot\Vert \xi-\gamma^*(\xi)\Vert_{l^\infty}\right),
    \end{aligned}
\end{align}
where the last inequality follows from the definition of $\mathcal{T}_k$. Since $\Vert \xi-\gamma^*(\xi)\Vert_{l^\infty}$ can be positive only if $\operatorname{pers}(\xi)>0$ or $\operatorname{pers}(\gamma^*(\xi))>0$, there are at most $N$ terms $\Vert \xi-\gamma^*(\xi)\Vert_{l^\infty}>0$, where the condition in Eq.~\eqref{Eq: topological invariants boundedness condition} implies
\begin{align*}
    N \overset{\operatorname{def}}{=} \sum_{i=1}^2 \#\{\xi\in \operatorname{Dgm}_k(K;\phi_{\nu_i}) \,\vert \, \operatorname{pers}(\xi)>0\} \le 2M/d.
\end{align*}
Therefore, the inequality in Eq.~\eqref{Eq: estimate in proof} implies
\begin{align*}
    \int_0^T \left\vert\Big\{\chi_\tau^{\nu_1}(K)-\chi_\tau^{\nu_2}(K)\Big\} \right\vert^p d\tau & \le \frac{2 \cdot (2M)^{(p+1)}}{d} \cdot \sup\Big\{ \Vert \xi-\gamma^*(\xi)\Vert_{l^\infty} \,\Big\vert\, \xi\in \operatorname{Dgm}_k(K;\phi_{\nu_1})\Big\} \\ 
    & = \frac{2 \cdot (2M)^{(p+1)}}{d} \cdot W_\infty \Big(\operatorname{Dgm}_k(K;\phi_{\nu_1}),\, \operatorname{Dgm}_k(K;\phi_{\nu_2}) \Big).
\end{align*}
Then, Theorem \ref{thm: bottleneck stability} implies 
\begin{align*}
    \int_0^T \left\vert\Big\{\chi_\tau^{\nu_1}(K)-\chi_\tau^{\nu_2}(K)\Big\} \right\vert^p d\tau \le \frac{2\cdot(2M)^{(p+1)}}{d} \cdot \sup_{x\in K} \vert x\cdot(\nu_1-\nu_2)\vert.
\end{align*}
Additionally, $\vert x\cdot(\nu_1-\nu_2)\vert\le \Vert x\Vert\cdot\Vert \nu_1-\nu_2\Vert$ and $K\subseteq B(0,R)$ provide
\begin{align}\label{Eq: Continuity inequality lemma}
    \int_0^T \left\vert\Big\{\chi_\tau^{\nu_1}(K)-\chi_\tau^{\nu_2}(K)\Big\} \right\vert^p d\tau \le \frac{2 \cdot R \cdot (2M)^{(p+1)}}{d} \cdot \Vert \nu_1-\nu_2\Vert.
\end{align}
Define the constant $C^*_{M,R,d}$ as follows
\begin{align}\label{eq: C^*_{M,R,d}}
    C^*_{M,R,d} \overset{\operatorname{def}}{=}  \sqrt{ \frac{16M^3R}{d} + \frac{32M^3R}{d} + \frac{64 M^4 R}{d^2} } .
\end{align}
(The constant $C^*_{M,R,d}$ defined in Eq.~\eqref{eq: C^*_{M,R,d}} will also be implemented in other proofs.)

Setting $p=2$, Eq.~\eqref{Eq: Continuity inequality lemma} implies the following
\begin{align*}
    \left( \int_0^T \left\vert\Big\{\chi_\tau^{\nu_1}(K)-\chi_\tau^{\nu_2}(K)\Big\} \right\vert^2 d\tau \right)^{1/2} & \le \sqrt{ \frac{16M^3R}{d} \cdot \Vert \nu_1-\nu_2\Vert } \\
    & \le C^*_{M,R,d} \cdot \sqrt{\Vert \nu_1-\nu_2\Vert},
\end{align*}
which is the inequality in Eq.~\eqref{Eq: lemma for the continuity inequality}. 
\end{proof}

\subsection{Proof of Theorem \ref{lemma: The continuity lemma}}

\begin{proof}[Proof of result (i), i.e., Eq.~\eqref{Eq: continuity inequality}.] 
The definition of $\operatorname{SECT}(K)$, together with Eq.~\eqref{Eq: Continuity inequality lemma}, implies
\begin{align*}
& \Big\Vert \operatorname{SECT}(K)(\nu_1) - \operatorname{SECT}(K)(\nu_2) \Big\Vert^2_{\mathcal{H}} \\
    & = \int_0^T \left\vert  \frac{d}{dt}\operatorname{SECT}(K)(\nu_1;t) - \frac{d}{dt}\operatorname{SECT}(K)(\nu_2;t) \right\vert^2 dt \\
    & = \int_0^T \left\vert  \Big(\chi_t^{\nu_1}(K)-\chi_t^{\nu_2}(K)\Big) - \frac{1}{T} \int_0^T \Big(\chi_\tau^{\nu_1}(K)-\chi_\tau^{\nu_2}(K)\Big) d\tau \right\vert^2 dt\\
    & \le \int_0^T \left( \Big\vert\chi_t^{\nu_1}(K)-\chi_t^{\nu_2}(K)\Big\vert + \frac{1}{T} \int_0^T \Big\vert\chi_\tau^{\nu_1}(K)-\chi_\tau^{\nu_2}(K)\Big\vert d\tau \right)^2 dt \\
    & \le 2\int_0^T  \Big\vert\chi_t^{\nu_1}(K)-\chi_t^{\nu_2}(K)\Big\vert^2 dt + \frac{2}{T}\left(\int_0^T \Big\vert\chi_\tau^{\nu_1}(K)-\chi_\tau^{\nu_2}(K)\Big\vert d\tau \right)^2 \\
    & \le \frac{32M^3R}{d} \cdot \Vert \nu_1 - \nu_2\Vert + \frac{64 M^4 R}{d^2} \cdot \Vert 
    \nu_1-\nu_2 \Vert^2,
\end{align*}
where the last inequality above comes from Eq.~\eqref{Eq: Continuity inequality lemma}. Then, we have
\begin{align*}
    & \Big\Vert \operatorname{SECT}(K)(\nu_1) - \operatorname{SECT}(K)(\nu_2) \Big\Vert_{\mathcal{H}} \\
    & \le \sqrt{ \frac{32M^3R}{d} \cdot \Vert \nu_1 - \nu_2\Vert + \frac{64 M^4 R}{d^2} \cdot \Vert 
    \nu_1-\nu_2 \Vert^2 } \\
    & \le C^*_{M,R,d} \cdot \sqrt{ \Vert \nu_1 - \nu_2\Vert + \Vert 
    \nu_1-\nu_2 \Vert^2 },
\end{align*}
where $C^*_{M,R,d}$ is defined in Eq.~\eqref{eq: C^*_{M,R,d}}. The proof of result (i), i.e., Eq.~\eqref{Eq: continuity inequality}, is complete. 
\end{proof}

\begin{proof}[Proof of result (ii).] 
The law of cosines and Taylor's expansion indicates
\begin{align*}
    \frac{\Vert \nu_1-\nu_2 \Vert}{d_{\mathbb{S}^{d-1}}(\nu_1, \nu_2)} & = \sqrt{ 2 \cdot \frac{1-\cos\left(d_{\mathbb{S}^{d-1}}(\nu_1, \nu_2)\right)}{\left\{d_{\mathbb{S}^{d-1}}(\nu_1, \nu_2)\right\}^2} }\\
    & = \sqrt{ 2 \cdot \left[ \sum_{n=1}^\infty \frac{(-1)^{n+1}}{(2n)!} \cdot \Big\{d_{\mathbb{S}^{d-1}}(\nu_1, \nu_2)\Big\}^{2n-2} \right] } = O(1).
\end{align*}
Then, result (ii) comes from the following
\begin{align}\label{eq: 1/2 Holder argument}
    \begin{aligned}
        & \frac{\Vert \operatorname{SECT}(K)(\nu_1)-\operatorname{SECT}(K)(\nu_2) \Vert_{\mathcal{H}}}{\sqrt{d_{\mathbb{S}^{d-1}}(\nu_1, \nu_2)}}\\ 
    & \le C^*_{M,R,d} \cdot \sqrt{ \frac{\Vert \nu_1 - \nu_2\Vert}{d_{\mathbb{S}^{d-1}}(\nu_1, \nu_2)} + \frac{\Vert 
    \nu_1-\nu_2 \Vert^2}{d_{\mathbb{S}^{d-1}}(\nu_1, \nu_2)} }=O(1).
    \end{aligned}
\end{align}
The proof of Theorem \ref{lemma: The continuity lemma} is complete.
\end{proof}

\subsection{Proof of Lemma \ref{lemma: uniform boundedness of the SECT}}

\begin{proof}[Proof of Lemma \ref{lemma: uniform boundedness of the SECT}.]
    Theorem \ref{thm: boundedness topological invariants theorem}, Lemma \ref{lemma: weak derivative formula}, and the definition of $\Vert \cdot\Vert_\mathcal{H}$ imply the following
\begin{align*}
\left\Vert \operatorname{SECT}(K) \right\Vert_{C(\mathbb{S}^{d-1};\mathcal{H})} &= \sup_{\nu\in\mathbb{S}^{d-1}} \left\{\Vert \operatorname{SECT}(K)(\nu)\Vert_{\mathcal{H}} \right\} \\
&= \sup_{\nu\in\mathbb{S}^{d-1}} \left(\int_0^T \left\vert \chi_t^\nu(K) - \frac{1}{T}\int_0^T \chi_\tau^\nu(K)d\tau \right\vert^2 dt\right)^{1/2} \\ 
&\le 2M\cdot\sqrt{T}.
\end{align*}
The proof of Lemma \ref{lemma: uniform boundedness of the SECT} is complete.
\end{proof}

\subsection{Proof of Eq.~\eqref{eq: bivariate Holder continuity}}

This subsection gives the proof of the second half of Theorem \ref{lemma: The continuity lemma; Appendix}, i.e., Eq.~\eqref{eq: bivariate Holder continuity}.

\begin{proof}[Proof of Eq.~\eqref{eq: bivariate Holder continuity}.]
    We consider the following inequality for all $\nu_1,\nu_2\in\mathbb{S}^{d-1}$ and $t_1, t_2\in[0,T]$
\begin{align}\label{eq: pm terms for Holder}
        \begin{aligned}
            &\vert \operatorname{SECT}(K)(\nu_1; t_1)-\operatorname{SECT}(K)(\nu_2; t_2)\vert\\
    &\le \vert \operatorname{SECT}(K)(\nu_1; t_1)-\operatorname{SECT}(K)(\nu_1; t_2)\vert\\
    & + \vert \operatorname{SECT}(K)(\nu_1; t_2)-\operatorname{SECT}(K)(\nu_2; t_2)\vert\\
    & \overset{\operatorname{def}}{=} I+II.
        \end{aligned}
\end{align}
From the definition of $\Vert\cdot\Vert_{C^{0,\frac{1}{2}}([0,T])}$ and Eq.~\eqref{eq: Sobolev embedding from Morrey}, we have
\begin{align*}
    & \sup_{t_1, t_2\in[0,T], \, t_1\ne t_2}\frac{\vert \operatorname{SECT}(K)(\nu_1; t_1)-\operatorname{SECT}(K)(\nu_1; t_2)\vert}{\vert t_1-t_2 \vert^{1/2}} \\ 
    & \le \Vert \operatorname{SECT}(K)(\nu_1)\Vert_{C^{0,\frac{1}{2}}([0,T])}\\
    & \le \Tilde{C}_T  \Vert \operatorname{SECT}(K)(\nu_1)\Vert_\mathcal{H}\\
    & \le \Tilde{C}_T  \Vert \operatorname{SECT}(K)\Vert_{C(\mathbb{S}^{d-1};\mathcal{H})},
\end{align*}
which implies the following for all $t_1, t_2\in[0,T]$ 
\begin{align*}
    I &\le \Tilde{C}_T  \Vert \operatorname{SECT}(K)\Vert_{C(\mathbb{S}^{d-1};\mathcal{H})} \cdot \vert t_1-t_2 \vert^{1/2} \\
    &\le \Tilde{C}_T \cdot 2M\sqrt{T} \cdot \sqrt{\vert t_1-t_2 \vert},
\end{align*}
where the second inequality follows from Lemma \ref{lemma: uniform boundedness of the SECT}.

Applying Eq.~\eqref{eq: Sobolev embedding from Morrey} again, we have 
\begin{align*}
    II & \le \Vert \operatorname{SECT}(K)(\nu_1)-\operatorname{SECT}(K)(\nu_2)\Vert_{\mathcal{B}} \\
    & \le \Tilde{C}_T \Vert \operatorname{SECT}(K)(\nu_1)-\operatorname{SECT}(K)(\nu_2)\Vert_{\mathcal{H}} \\
    & \le \Tilde{C}_T \cdot C^*_{M,R,d} \cdot \sqrt{ \Vert \nu_1 - \nu_2\Vert + \Vert 
    \nu_1-\nu_2 \Vert^2 },
\end{align*}
where the last inequality follows from Theorem \ref{lemma: The continuity lemma}(i). Then, the inequality in Eq.~\eqref{eq: bivariate Holder continuity} follows from Eq.~\eqref{eq: pm terms for Holder}.
\end{proof}

\subsection{Proof of Theorem \ref{thm: invertibility}}\label{section: proof of the injectivity of SECT}

Recall the following concepts discussed in Appendix \ref{section: Further Theorems}:
\begin{itemize}
    \item For any given o-minimal structure $\mathcal{S}$, any elements of $\mathcal{S}$ are called definable sets.
    \item Compact definable sets are called constructible sets. The collection of constructible subsets of $\mathbb{R}^d$ is denoted by $\operatorname{CS}(\mathbb{R}^d)$.
    \item If $\mathcal{S}$ satisfies Assumption \ref{Assumption: basic requirements for o-minimal structures of interest}, we have $\mathcal{S}_{R,d}^M \subseteq \operatorname{CS}(\mathbb{R}^d)$.
\end{itemize}
Then, Theorem \ref{thm: invertibility} follows directly from Corollary \ref{corollary: Corollary 1 of Ghrist et all.(2018)} (i.e., Corollary 1 of \cite{ghrist2018persistent}).

We have the following as a further explanation for Corollary \ref{corollary: Corollary 1 of Ghrist et all.(2018)}: Using a Morse theory-like result, \cite{ji2023euler} showed that the ECT and SECT determine each other; then, the injectivity of the ECT stated in Theorem \ref{thm: injectivity of the ECT} (i.e., Theorem 1 of \cite{ghrist2018persistent} or Theorem 3.5 of \cite{curry2022many}) implies the injectivity of the SECT.

\subsection{Proof of Theorem \ref{Thm: metric theorem for shapes}}\label{proof: measurability of SECT and PECT; appendix}

\begin{proof}[Proof of Theorem \ref{Thm: metric theorem for shapes}.]
The proof needs the concept of PECT defined in Eq.~\eqref{Eq: def of PECT}.

The compactness of $K$, together with $K \subseteq B(0,R)$, implies that $\chi(K_t^\nu)=0$ for all $t$ satisfying the following
\begin{align*}
    0\le t< \operatorname{dist}\left(K, \partial B(0,R) \right)=\inf\left\{\Vert x-y\Vert \,\Big\vert\, x\in K \text{ and }y\in \partial B(0,R) \right\} \overset{\operatorname{def}}{=} \sigma>0.
\end{align*}
Therefore, $\operatorname{PECT}(K)(\nu,t)= \int_0^t \chi_\tau^\nu(K) \,d\tau =0$ for all $0\le t<\sigma$, which implies
\begin{align*}
    \frac{d^+}{dt}\Big\vert_{t=0} \operatorname{SECT}(K)(\nu, t) = - \frac{1}{T} \operatorname{PECT}(K)(\nu, T),
\end{align*}
where $\frac{d^+}{dt}$ denotes the right derivative with respect to $t$. Hence, we have
\begin{align*}
    \operatorname{PECT}(K)(\nu, t) = \operatorname{SECT}(K)(\nu, t) - t\cdot \frac{d^+}{dt}\Big\vert_{t=0} \operatorname{SECT}(K)(\nu, t).
\end{align*}
That is, $\operatorname{PECT}(K)$ and $\operatorname{SECT}(K)$ determine each other. Then, Theorem \ref{thm: invertibility} implies that the $\operatorname{PECT}$ defined in Eq.~\eqref{Eq: def of PECT} is injective. (Alternatively, Lemma \ref{lemma: weak derivative formula}, together with Lemma \ref{lemma: right continuity of the ECT} and Theorem \ref{thm: injectivity of the ECT}, also implies that the $\operatorname{PECT}$ in Eq.~\eqref{Eq: def of PECT} is injective.)

The triangle inequalities and symmetry of $\rho$ follow from that of the metric of $C(\mathbb{S}^{d-1};\mathcal{H})$. Equation $\rho(K_1, K_2)=0$ indicates $\Vert \operatorname{PECT}(K_1)(\nu)-\operatorname{PECT}(K_2)(\nu)\Vert_{\mathcal{H}_{BM}}=0$ for all $\nu\in\mathbb{S}^{d-1}$. \cite{evans2010partial} (Theorem 5 of Chapter 5.6) implies $\Vert \operatorname{PECT}(K_1)(\nu)-\operatorname{PECT}(K_2)(\nu)\Vert_{\mathcal{B}}=0$ for all $\nu\in\mathbb{S}^{d-1}$. Then, we have $\int_0^t\chi_\tau^\nu(K_1) d\tau=\int_0^t\chi_\tau^\nu(K_2) d\tau$ for all $t\in[0,T]$ and $\nu\in\mathbb{S}^{d-1}$; hence, $\operatorname{SECT}(K_1)=\operatorname{SECT}(K_2)$. Then, Theorem \ref{thm: invertibility} implies $K_1=K_2$. Therefore, $\rho$ is a distance.

The proof of that $\mathscr{F}=\mathscr{B}(\rho)$ satisfies Assumption \ref{assumption: the measurability of ECC} is motivated by the following chain of maps for any fixed $\nu\in\mathbb{S}^{d-1}$ and $t\in[0,T]$.
\begin{align*}
    & \mathcal{S}_{R,d}^M \ \ \ \xrightarrow{\operatorname{PECT}}\ \ \  C(\mathbb{S}^{d-1};\mathcal{H}_{BM})\ \ \  \xrightarrow{\text{projection}}\ \ \  \mathcal{H}_{BM}\mbox{, which is embedded into }\mathcal{B} \ \ \  \xrightarrow{\text{projection}} \ \ \mathbb{R}, \\
    & K \mapsto \left\{\operatorname{PECT}(K)(\nu')\right\}_{\nu'\in\mathbb{S}^{d-1}} \mapsto \left\{\operatorname{PECT}(K)(\nu, t') \right\}_{t'\in[0,T]} \ \mapsto \operatorname{PECT}(K)(\nu, t)=\int_0^{t} \chi_{\tau}^\nu(K) d\tau,
\end{align*}
where all spaces above are metric spaces and equipped with their Borel algebras. We notice the following facts: 
\begin{itemize}
    \item the mapping $\operatorname{PECT}: \mathcal{S}_{R,d}^M \rightarrow C(\mathbb{S}^{d-1}; \mathcal{H}_{BM})$ is isometric;
    \item the projection $C(\mathbb{S}^{d-1};\mathcal{H}_{BM})\rightarrow\mathcal{H}_{BM}, \, \{F(\nu')\}_{\nu'\in\mathbb{S}^{d-1}}\mapsto F(\nu)$ is continuous for each fixed direction $\nu$;
    \item applying \cite{evans2010partial} (Theorem 5 of Chapter 5.6) again, the embedding $\mathcal{H}_{BM}\rightarrow\mathcal{B}, \, F(\nu) \mapsto F(\nu)$ is continuous;
    \item projection $\mathcal{B}\rightarrow \mathbb{R},\ \{x(t')\}_{t'\in[0,T]}\mapsto x(t)$ is continuous.
\end{itemize}
Therefore, $\mathcal{S}_{R,d}^M\rightarrow\mathbb{R},\ \ K\mapsto \operatorname{PECT}(K)(\nu,t)$ is continuous, hence, measurable. 

For any $K\in\mathcal{S}_{R,d}^M$, Lemmas \ref{thm: tameness property} and \ref{lemma: right continuity of the ECT} imply the following for all $t$ and $\nu$
\begin{align*}
    \chi^\nu_{t}(K)=\lim_{n\rightarrow\infty}\left[\frac{1}{\delta_n}\left
\{\operatorname{PECT}(K)(\nu, t+\delta_n)-\operatorname{PECT}(K)(\nu, t)\right\}\right],
\end{align*}
where $\lim_{n\rightarrow\infty}\delta_n=0$ and $\delta_n>0$. The measurability of $\operatorname{PECT}(\nu, t+\delta_n)$ and $\operatorname{PECT}(\nu, t)$ implies that $\chi_t^\nu: \mathcal{S}_{R,d}^M\rightarrow\mathbb{R}, K\mapsto \chi_t^\nu(K)$ is measurable, for any fixed $\nu$ and $t$. 

The proof of Theorem \ref{Thm: metric theorem for shapes} is complete. 
\end{proof}

\subsection{Proof of Lemma \ref{assumption: existence of second moments}}
\begin{proof}[Proof of Lemma \ref{assumption: existence of second moments}.]
    Lemma \ref{lemma: uniform boundedness of the SECT} implies the following for all $K\in\mathcal{S}_{R,d}^M$ 
    \begin{align*}
        \sup_{\nu\in\mathbb{S}^{d-1}} \Vert \operatorname{SECT}(K)(\nu)\Vert_{\mathcal{H}}^2 \le 4M^2T.
    \end{align*}
    Then, we have
    \begin{align*}
        \int_{\mathcal{S}_{R,d}^M} \left\{ \sup_{\nu\in\mathbb{S}^{d-1}} \Vert \operatorname{SECT}(K)(\nu)\Vert_{\mathcal{H}}^2 \right\} \,\mathbb{P}(dK) \le 4M^2T <\infty.
    \end{align*}
    The proof of Lemma \ref{assumption: existence of second moments} is complete.
\end{proof}

\subsection{Proof of Lemma \ref{thm: mean is in H}}

The proof of Lemma \ref{thm: mean is in H} is divided into five small proofs.

\begin{proof}[Proof of result (i).]
    For each fixed direction $\nu\in\mathbb{S}^{d-1}$, Theorem \ref{thm: SECT distribution theorem in each direction} indicates that the mapping $\operatorname{SECT}(\nu): K\mapsto \operatorname{SECT}(K)(\nu)$ is an $\mathcal{H}$-valued measurable function defined on the probability space $(\mathcal{S}_{R,d}^M, \mathscr{F}, \mathbb{P})$. We first show the Bochner $\mathbb{P}$-integrability of $\operatorname{SECT}(\nu)$ (see Section 5 in Chapter V of \cite{yosida1965functional} for the definition of Bochner $\mathbb{P}$-integrability), and the Bochner integral of $\operatorname{SECT}(\nu)$ will be fundamental to our proof. Lemma 1.3 of \cite{da2014stochastic} indicates that $\operatorname{SECT}(\nu)$ is strongly $\mathscr{F}$-measurable (see Section 4 in Chapter V of \cite{yosida1965functional} for the definition of strong $\mathscr{F}$-measurability). Then, Lemma \ref{assumption: existence of second moments} indicates that the Bochner integral
\begin{align*}
    m^*_\nu \overset{\operatorname{def}}{=} \int_{\mathcal{S}_{R,d}^M} \operatorname{SECT}(K)(\nu) \,\mathbb{P}(dK)
\end{align*}
is Bochner $\mathbb{P}$-integrable and $m^*_\nu \in\mathcal{H}$ \citep[][Section 5 of Chapter V, Theorem 1]{yosida1965functional}. The Corollary 2 in Section 5 of Chapter V of \cite{yosida1965functional}, together with that $\mathcal{H}$ is the RKHS generated by the kernel $\kappa(s,t)=\min\{s,t\}-\frac{st}{T}$ \citep[][Example 4.9]{lifshits2012lectures}, implies
\begin{align*}
    m^*_\nu(t) &= \langle \kappa(t,\cdot),\, m^*_\nu \rangle \\
    & = \int_{\mathcal{S}_{R,d}^M} \Big\langle \kappa(t,\cdot), \, \operatorname{SECT}(K)(\nu) \Big\rangle \,\mathbb{P}(dK) \\
    & =\int_{\mathcal{S}_{R,d}^M} \operatorname{SECT}(K)(\nu, t) \,\mathbb{P}(dK) \\
    & = \mathbb{E}\left\{ \operatorname{SECT}(\nu, t)\right\} =m_\nu(t),\ \ \mbox{ for all }t\in[0,T],
\end{align*}
where $\langle\cdot, \cdot\rangle$ denotes the inner product of $\mathcal{H}$. Therefore, $m_\nu=m^*_\nu\in\mathcal{H}$. The proof of result (i) is complete.
\end{proof}

\begin{proof}[Proof of result (ii).]
    To prove result (ii), we first show the product measurability of the following map for each fixed direction $\nu\in\mathbb{S}^{d-1}$
\begin{align}\label{eq: product measurability}
    \begin{aligned}
        & \Big(\mathcal{S}_{R,d}^M \times [0,T], \mathscr{F}\otimes \mathscr{B}([0,T])\Big)\ \  \rightarrow \ \ (\mathbb{R}, \mathscr{B}(\mathbb{R})),\\
    & (K,t) \ \ \mapsto \ \ \operatorname{SECT}(K)(\nu, t),
    \end{aligned}
\end{align}
where $\mathscr{F}\otimes \mathscr{B}([0,T])$ denotes the product $\sigma$-algebra generated by $\mathscr{F}$ and $\mathscr{B}([0,T])$. Define the filtration $\{\mathscr{F}_t\}_{t\in[0,T]}$ by $\mathscr{F}_t \overset{\operatorname{def}}{=}\sigma(\{\operatorname{SECT}(\nu, t')\,\vert\, t'\in[0,t]\}) \subseteq \mathscr{F}$ for $t\in[0,T]$. Because the sample paths of $\operatorname{SECT}(\nu)$ are in $\mathcal{H}$, these sample paths are continuous (see the Sobolev embedding in Eq.~\eqref{eq: H, Holder, B embeddings}. Proposition 1.13 of \cite{karatzas2012brownian} implies that the stochastic process $\operatorname{SECT}(\nu)$ is progressively measurable with respect to the filtration $\{\mathscr{F}_t\}_{t\in[0,T]}$. Then, the mapping in Eq.~\eqref{eq: product measurability} is measurable with respect to the product $\sigma$-algebra $\mathscr{F}\otimes \mathscr{B}([0,T])$ \citep[][Definitions 1.6 and 1.11, also the paragraph right after Definition 1.11 therein]{karatzas2012brownian}. Lemma \ref{assumption: existence of second moments} implies
\begin{align*}
    \int_0^T \int_{\mathcal{S}_{R,d}^M} \vert \operatorname{SECT}(K)(\nu, t)\vert^2 \,\mathbb{P}(dK) dt \le T \cdot \Tilde{C}^2_T \cdot \mathbb{E}\Vert \operatorname{SECT}(\nu)\Vert^2_{\mathcal{H}}<\infty,
\end{align*}
where the double integral is well-defined because of the product measurability of the mapping in Eq.~\eqref{eq: product measurability} and the Fubini's theorem. Then, the proof of result (ii) is complete. 
\end{proof}

\begin{proof}[Proof of result (iii).]
    Eq.~\eqref{eq: bivariate Holder continuity} implies
    \begin{align*}
        \mathbb{E}\left\vert  \operatorname{SECT}(\nu,t+\epsilon) - \operatorname{SECT}(\nu,t) \right\vert^2 \le \Tilde{C}_T^2\cdot 4M^2T\cdot\vert\epsilon\vert \rightarrow 0,
    \end{align*}
    as $\epsilon\rightarrow0$. The proof is complete.
\end{proof}

\begin{proof}[Proof of result (iv).]
    Result (iv) follows from Lemma 4.2 of \cite{alexanderian2015brief}.
\end{proof}

\begin{proof}[Proof of result (v).]
    For any $\nu_1,\nu_2\in\mathbb{S}^{d-1}$, the proof of result (i) implies the following Bochner integral representation
\begin{align*}
    \Vert m_{\nu_1} - m_{\nu_2} \Vert_{\mathcal{H}} & = \left\Vert \int_{\mathcal{S}_{R,d}^M} \operatorname{SECT}(K)(\nu_1) - \operatorname{SECT}(K)(\nu_2) \,\mathbb{P}(dK) \right\Vert_{\mathcal{H}} \\
    & \overset{(1)}{\le} \int_{\mathcal{S}_{R,d}^M} \Big\Vert \operatorname{SECT}(K)(\nu_1) - \operatorname{SECT}(K)(\nu_2) \Big\Vert_{\mathcal{H}} \mathbb{P}(dK) \\ 
    & \overset{(2)}{\le} C^*_{M,R,d} \cdot \sqrt{ \Vert \nu_1 - \nu_2\Vert + \Vert 
    \nu_1-\nu_2 \Vert^2 },
\end{align*}
where the inequality (1) follows from the Corollary 1 in Section 5 of Chapter V of \cite{yosida1965functional}, and the inequality (2) follows from Eq.~\eqref{Eq: continuity inequality}. With the argument in Eq.~\eqref{eq: 1/2 Holder argument}, the proof of result (v) is complete.
\end{proof}

\subsection{Proof of Theorem \ref{thm: KL expansions of SECT}}\label{section: proof of KL expansion; Appendix}

\begin{proof}[Proof of Theorem \ref{thm: KL expansions of SECT}]
    Lemma \ref{thm: mean is in H} implies that, for each $j\in\{1,2\}$, the stochastic process $\{\operatorname{SECT}(\nu^*;t)-m_{\nu^*}^{(j)}(t)\}_{t\in[0,T]}$ is of mean zero, mean-square continuous, and belongs to $L^2(\mathcal{S}_{R,d}^M\times[0,T],\, \mathbb{P}^{(j)}(dK)\otimes dt)$. Then, result (i) follows from Theorem 7.3.5 of \cite{hsing2015theoretical} (equivalently, Corollary 5.5 of \cite{alexanderian2015brief}). 

To prove result (ii), we denote the following
\begin{align*}
    &\ \ \ \ \ D_L(K^{(1)},K^{(2)}; \,t) \\
    & \overset{\operatorname{def}}{=}\left\{ \operatorname{SECT}(K^{(1)})(\nu^*;t) - \operatorname{SECT}(K^{(2)})(\nu^*;t) \right\} \\ 
    &\ \ \ \ \ - \left[ \left\{m^{(1)}_{\nu^*}(t) + \sum_{l'=1}^L \sqrt{\lambda_{l'}} \cdot Z_{l'}^{(1)}(K^{(1)}) \cdot \phi_{l'}(t)\right\} -\left\{m^{(2)}_{\nu^*}(t) + \sum_{l'=1}^L \sqrt{\lambda_{l'}} \cdot Z_{l'}^{(2)}(K^{(2)}) \cdot \phi_{l'}(t)\right\} \right]
\end{align*}
Then, result (i) implies the following
\begin{align*}
    &\ \ \ \lim_{L\rightarrow\infty} \left\{\sup_{t\in[0,T]} \left\Vert D_L(\cdot, \cdot; \, t)\right\Vert^2_{L^2}\right\} \\
    & =\lim_{L\rightarrow\infty}\left\{\sup_{t\in[0,T]} \int_{\mathcal{S}_{R,d}^M \times \mathcal{S}_{R,d}^M}  \left\vert D(K^{(1)},K^{(2)}; \, t) \right\vert^2 \mathbb{P}^{(1)}\otimes\mathbb{P}^{(2)}(dK^{(1)}, dK^{(2)})\right\} = 0,
\end{align*}
where $L^2$ is the abbreviation for $L^2(\mathcal{S}_{R,d}^M\times \mathcal{S}_{R,d}^M, \mathscr{F}\otimes\mathscr{F}, \mathbb{P}^{(1)}\otimes\mathbb{P}^{(2)})$. For each fixed $l=1,2,\ldots$, we have
\begin{align}\label{eq: zero L2 norm}
    \begin{aligned}
        & \ \ \ \ \left\Vert \frac{1}{\sqrt{2\lambda_l}}\int_0^T D_L(\cdot,\cdot; \, t)\phi_l(t) dt\right\Vert_{L^2} \\
        & \le \frac{1}{\sqrt{2\lambda_l}}\int_0^T \left\Vert D_L(\cdot,\cdot; \, t) \right\Vert_{L^2} \vert\phi_l(t)\vert dt \\
    &\le \sup_{t\in[0,T]}\left\Vert D_L(\cdot,\cdot; \, t) \right\Vert_{L^2} \cdot\frac{1}{\sqrt{2\lambda_l}}\int_0^T \vert\phi_l(\tau)\vert d\tau \rightarrow0, \ \ \mbox{ as }L\rightarrow\infty.
    \end{aligned}
\end{align}
In addition, for each fixed $l=1,2,\ldots$ and $L>l$, we have
\begin{align*}
    & \ \ \ \ \frac{1}{\sqrt{2\lambda_l}}\int_0^T D_L(K^{(1)}, K^{(2)};t)\phi_l(t) dt\\ 
    &= \frac{1}{\sqrt{2\lambda_l}}\int_0^T \left\{ \operatorname{SECT}(K^{(1)})(\nu^*;t) - \operatorname{SECT}(K^{(2)})(\nu^*;t) \right\}\phi_l(t)dt \\
    &\ \ -\frac{1}{\sqrt{2\lambda_l}}\int_0^T \left\{m^{(1)}_{\nu^*}(t)-m^{(2)}_{\nu^*}(t)\right\}\phi_l(t)dt \\
    &\ \ -\frac{1}{\sqrt{2\lambda_l}}\int_0^T \left[ \sum_{l'=1}^L \sqrt{\lambda_{l'}} \cdot \left\{ Z_{l'}^{(1)}(K^{(1)})-Z_{l'}^{(2)}(K^{(2)}) \right\} \cdot \phi_{l'}(t) \right]\cdot\phi_l(t) dt \\
    &=\delta_l\left(K^{(1)}, \, K^{(2)}\right) - \left[ \theta_l + \left( \frac{Z_{l}^{(1)}(K^{(1)})-Z_{l}^{(2)}(K^{(2)})}{\sqrt{2}} \right) \right],
\end{align*}
where $\delta_l\left(K^{(1)}, \, K^{(2)}\right)$ is defined in Eq.~\eqref{eq: KL expansions of SECT}, and the last equality follows from the $L^2([0,T])$-orthonormality of $\{\phi_l\}_{l=1}^\infty$. The limit in Eq.~\eqref{eq: zero L2 norm} imply
\begin{align*}
    & \ \ \ \ \int_{\mathcal{S}_{R,d}^M \times \mathcal{S}_{R,d}^M} \left\vert \, \delta_l\left(K^{(1)}, \, K^{(2)}\right) - \left[ \theta_l + \left( \frac{Z_{l}^{(1)}(K^{(1)})-Z_{l}^{(2)}(K^{(2)})}{\sqrt{2}} \right) \right] \right\vert^2 \mathbb{P}^{(1)}\otimes\mathbb{P}^{(2)}(dK^{(1)}, dK^{(2)}) \\
    & = \lim_{L\rightarrow\infty} \left\Vert \frac{1}{\sqrt{2\lambda_l}}\int_0^T D_L(\cdot,\cdot; \, t)\phi_l(t) dt\right\Vert_{L^2} \\
    &=0.
\end{align*}
Then, there exists $\mathcal{N}_l\in\mathscr{F}\otimes\mathscr{F}$, which depends on $l$, such that $\mathbb{P}^{(1)}\otimes\mathbb{P}^{(2)}(\mathcal{N}_l)=0$ and
\begin{align}\label{eq: ultimate goal of the KL expansion}
    \delta_l\left(K^{(1)}, \, K^{(2)}\right) = \theta_l + \left( \frac{Z_{l}^{(1)}(K^{(1)})-Z_{l}^{(2)}(K^{(2)})}{\sqrt{2}} \right),
\end{align}
for any $(K^{(1)}, \, K^{(2)}) \notin \mathcal{N}_l$. Define $\mathcal{N}\overset{\operatorname{def}}{=}\bigcup_{l=1}^\infty \mathcal{N}_l$, we have $\mathbb{P}^{(1)}\otimes\mathbb{P}^{(2)}(\mathcal{N})=0$ and Eq.~\eqref{eq: ultimate goal of the KL expansion} holds for all $(K^{(1)}, \, K^{(2)}) \notin \mathcal{N}$ and $l=1,2,\ldots$. The proof of result (ii) is complete. 
\end{proof}

\subsection{Proof of Lemma \ref{lemma: representing H0}}

\begin{proof}[Proof of Lemma \ref{lemma: representing H0}.]
We have shown that the null $H_0$ is equivalent to $m_{\nu^*}^{(1)}(t)=m_{\nu^*}^{(2)}(t)$ for all $t\in[0,T]$, where $\nu^*$ is defined in Eq.~\eqref{eq: def of distinguishing direction}. The null $H_0$ directly implies that $\theta_l=0$ for all $l$. On the other hand, if $\theta_l=0$ for all $l$, that $\{\phi_l\}_l$ is an orthonormal basis of $L^2([0,T])$ indicates that $m_{\nu^*}^{(1)}=m_{\nu^*}^{(2)}$ almost everywhere with respect to the Lebesgue measure $dt$. Part (i) of Lemma \ref{thm: mean is in H} and the embedding $\mathcal{H}\subseteq\mathcal{B}$ in Eq.~\eqref{eq: H, Holder, B embeddings} imply that $m_{\nu^*}^{(1)}$ and $m_{\nu^*}^{(2)}$ are continuous functions. As a result, $m_{\nu^*}^{(1)}(t)=m_{\nu^*}^{(2)}(t)$ for all $t\in[0,T]$. The proof is complete.
\end{proof}

\subsection{Proof of Theorem \ref{eq: the SECT is non-Gaussian}}

\begin{proof}[Proof of Theorem \ref{eq: the SECT is non-Gaussian}] 
    The independence condition implies that the stochastic process $\operatorname{PECT}(\nu)=\{\operatorname{PECT}(\nu,t)=\int_0^t \chi_\tau^\nu \,d\tau\}_{t\in[0,T]}$ has independent increments. The continuity of the sample paths of $\operatorname{PECT}(\nu)$ and that $\operatorname{PECT}(\nu,0)=\int_0^0 \chi_\tau^\nu \,d\tau = 0$ imply $\operatorname{PECT}(\nu)$ is a Gaussian process \citep[][Theorem 14.4]{kallenberg2021foundations}.
\end{proof}

\section{Potential Future Research Areas}\label{section: potential future research areas}

In this section, we list several potential future research areas that we believe are related to our work.

\subsection{Generative Models for Complex Shapes}\label{appendix: Generative Models for Complex Shapes}

Theorem \ref{thm: KL expansions of SECT} also holds for any fixed direction $\nu$; that is, the distinguishing direction $\nu^*$ in Theorem \ref{thm: KL expansions of SECT} can be replaced with any fixed $\nu\in\mathbb{S}^{d-1}$ (see the corresponding proof in Appendix \ref{section: proof of KL expansion; Appendix}). The first result of Theorem \ref{thm: KL expansions of SECT} can be formally represented as follows 
\begin{align}\label{eq: formal KL expansion}
    \operatorname{SECT}(K)(\nu,t) = m^{(j)}_{\nu}(t) + \sum_{l=1}^\infty \sqrt{\lambda_l} \cdot Z_{l}^{(j)}(K) \cdot \phi_l(t).
\end{align}
Using Eq.~\eqref{eq: formal KL expansion}, the random sampling of shapes may be considered. This involves sampling the stochastic process on the right-hand side of Eq.~\eqref{eq: formal KL expansion} and reconstructing a shape by applying the inverse of the injective map $
\operatorname{SECT}:\mathcal{S}_{R,d}^M\rightarrow C(\mathbb{S}^{d-1};\mathcal{H})$. Still, several challenges arise: 
\begin{enumerate}
    \item the map \( \operatorname{SECT} \) is not surjective, and the characterization of the image \( \operatorname{SECT}(\mathcal{S}_{R,d}^M) \) remains to be developed;

    \item one must properly select the covariance function \( \Xi_\nu \) and the distribution of \( \{Z_l^{(j)}\}_{l=1}^\infty \) to ensure that the sample paths of $ m^{(j)}_{\nu}(t) + \sum_{l=1}^\infty \sqrt{\lambda_l} \cdot Z_{l}^{(j)} \cdot \phi_l(t)$ belong to \( \operatorname{SECT}(\mathcal{S}_{R,d}^M) \);

    \item reconstructing shapes \( K \) from \( \operatorname{SECT}(K) \), as discussed in \cite{fasy2018challenges}, is still an open question in the field. The random sampling of shapes using Eq.~\eqref{eq: formal KL expansion} is left for future research. 
\end{enumerate}

\subsection{Definition of Mean Shapes}

The existence and uniqueness of the mean shapes $K_\oplus$ in the Fréchet sense as defined in Eq.~\eqref{Frechet mean shape} are still unknown. If mean shapes $K_\oplus$ do exist, the relationship between $\operatorname{SECT}(K_\oplus)$ and $\mathbb{E}\{\operatorname{SECT}\}$ is of particular interest. In addition, the relationship of $\mathbb{E}\{\operatorname{SECT}\}$ to the theory for ``expectations of random sets" \citep{molchanov2005theory} is of interest and left for future research. 

The Fréchet mean in Eq.~\eqref{Frechet mean shape} may be extended to the conditional Fréchet mean and implemented in the Fréchet regression --- predicting shapes $K$ using multiple scalar predictors \citep{petersen2019frechet}. For example, predicting molecular shapes and structures from scalar-valued indicators or sequences has become of high interest in biology \citep{jumper2021highly, yang2020improved}. As an example of the other way around, predicting clinical outcomes from the tumors is also of interest \citep{moon2020predicting, somasundaram2021persistent, vipond2021multiparameter}, which is potentially relevant to the Wasserstein regression \citep{chen2021wasserstein}.



\subsection{Two-sample Test via the Reproducing Kernel Hilbert Space embedding or Optimal Transport}

From the viewpoint of statistical inference, the ultimate goal of this paper is to solve the following two-sample test problem (also see Eq.~\eqref{eq: original hypotheses}):
\begin{align}\label{eq: oringial two-sample test problem; Appendix}
    H_0^*:\ \ \mathbb{P}^{(1)} = \mathbb{P}^{(2)},\ \ \ vs. \ \ \ H_1^*: \ \ \mathbb{P}^{(1)} \ne \mathbb{P}^{(2)}, 
\end{align} 
where the observed shapes $\{K_i^{(j)}\}_{i=1}^n\overset{\operatorname{i.i.d.}}{\sim}\mathbb{P}^{(j)}$, for $j\in\{1,2\}$. Through the distinguishing direction $\nu^*$ defined in Eq.~\eqref{eq: def of distinguishing direction}, the two-sample test problem in Eq.~\eqref{eq: oringial two-sample test problem; Appendix} can be transformed into testing the equality of probability measures defined on the RKHS $\mathcal{H}$, i.e.,
\begin{align}\label{eq: pushed forward two-sample test problem; Appendix}
    H_0:\ \ \mathbf{P}^{(1)} = \mathbf{P}^{(2)},\ \ vs. \ \ \  \mathbf{P}^{(1)} \ne \mathbf{P}^{(2)},
\end{align}
where $\mathbf{P}^{(j)}(B)\overset{\operatorname{def}}{=}\mathbb{P}^{(j)}\left\{K\in\mathcal{S}_{R,d}^M \,\vert\, \operatorname{SECT}(K)(\nu^*)\in B\right\}$ for all $B\in\mathscr{B}(\mathcal{H})$ and $j\in\{1,2\}$.

In the literature, numerous powerful frameworks for two-sample test problems have been developed in the past decade. One notable framework is the ``kernel two-sample test" \citep{gretton2006kernel, gretton2012kernel, hagrass2023spectral}, which is based on the concept of ``RKHS embedding of probability measures" \citep{smola2007hilbert, fukumizu2009kernel, sriperumbudur2010hilbert, sriperumbudur2011universality, muandet2017kernel}. Another framework is the rank-based distribution-free test framework proposed by \cite{deb2023multivariate}, rooted in the theory of optimal transport \citep{villani2009optimal}. 

With appropriate adjustments, it is possible to adapt the existing two-sample test frameworks from the literature to address the specific two-sample test problems outlined in Eq.~\eqref{eq: oringial two-sample test problem; Appendix} and Eq.~\eqref{eq: pushed forward two-sample test problem; Appendix}. The exploration of this avenue is left for our future research.

\subsection{Euler Characteristic-based Statistical Inference on Grayscale Images: Theory and Applications}

Each shape $K$ can be viewed as a binary-valued image, e.g., the points inside $K$ are assigned to be 1 while the points outside $K$ are assigned to be 0. That is, the shape can be equivalently represented as the indicator function $\mathbbm{1}_K$ of the shape $K$. Many images in applications are real-valued instead of binary-valued, e.g., the computed tomography (CT) scans of lung cancer tumors \citep{maldonado2021validation}. The real-valued images are referred to as grayscale images.

Over the past several years, some Euler characteristic-based representations of grayscale images have been proposed. \cite{jiang2020weighted} proposed the weighted Euler curve transform (WECT) for the analysis of MNIST digit images \citep{lecun1998gradient} and GBM tumor data \citep{holland2000glioblastoma}. \cite{kirveslahti2023representing} introduced three representations: the lifted ECT (LECT), super LECT (SELECT), and marginal Euler curve (MEC) for grayscale images. \cite{kirveslahti2023representing} also demonstrated that the MEC coincides with the WECT on weighted simplicial complexes. \cite{meng2023Inference} introduced the Euler-Radon transform (ERT) for modeling grayscale images using Euler integration over real-valued functions \citep{baryshnikov2010euler}. Notably, they found that the MEC coincides with the ``floor version" of the ERT. However, one key question remains unresolved in this series of frameworks: a probability space has yet to be constructed to mathematically characterize the randomness of the grayscale images of interest, which is left for future research.

\subsection{Euler Characteristic-based Topological Data Analysis}

Euler characteristic-based descriptors, especially the ECT, have been pivotal for TDA due to the following:
\begin{enumerate}
    \item Euler calculus \citep{schapira1988cycles, viro1988some, schapira1991operations, schapira1995tomography, van1998tame, baryshnikov2010euler, ghrist2014elementary} provides rich mathematical machinery for the Euler characteristic-based descriptors in TDA. For example, \cite{ghrist2018persistent} and \cite{curry2022many} applied Schapira's inversion formula \citep{schapira1995tomography} to show the injectivity of the ECT.

    \item Compared to persistence diagram-based descriptors, computing Euler characteristic-based descriptors is more efficient \citep{milosavljevic2011zigzag, hacquard2023euler, munch2023invitation}.

    \item As illustrated in this paper, Euler characteristic-based descriptors may allow the implementation of functional analysis tools for mathematical purposes and the implementation of functional data analysis for statistical purposes.
\end{enumerate}

\cite{munch2023invitation} provided a comprehensive overview of the ECT from both theoretical and applied perspectives. In addition, persistent homology-based descriptors have been applied to brain networks for more than a decade \citep{lee2011discriminative, wang2023topological, li2023tree}. We may also consider applying Euler characteristic-based descriptors to weighted graphs, particularly focusing on their applications to brain connectivity \citep{friston1993functional, lee2011discriminative, li2023tree, chen2024gradient, meng2024population}. There is still much room for the future development of Euler characteristic-based descriptors.

\end{appendix}




\bibliography{sample}

\end{document}